\documentclass[smallcondensed]{svjour3}    %
\date{Published 2020}
\def\makeheadbox{\relax}
\usepackage[utf8]{inputenc}
\usepackage[colorlinks=true]{hyperref}
\usepackage{enumitem}
\usepackage{mathtools}
\usepackage{amssymb}
\usepackage[dvipsnames]{xcolor}
\usepackage{pgfplots}
\pgfplotsset{compat=1.13} 
\usepackage{float}
\usepackage{color}
\usepackage{tikz}
\usetikzlibrary{arrows.meta}
\usetikzlibrary{arrows,automata}
\usetikzlibrary{shapes,decorations,patterns,positioning,angles}
\usepackage{makecell}
\usepackage{subcaption}
\usepackage{tikz-3dplot}
\usepackage{listings}
\usepackage{multirow}
\usepackage{xspace}
\usepackage{algorithm}
\usepackage{algcompatible}
\usepackage{tcolorbox}
\usepackage{kbordermatrix}
\renewcommand{\kbldelim}{(}
\renewcommand{\kbrdelim}{)}
\tcbuselibrary{listings,skins}
\smartqed

\algnewcommand{\PROCEDURE}[1]{\STATE \textbf{procedure}\space #1}
\algnewcommand{\ENDPROCEDURE}{\STATE \textbf{end procedure}}
\algnewcommand\RETURN{\STATE \textbf{return}\space}
\algnewcommand\TRUE{\textbf{true}}
\algnewcommand\FALSE{\textbf{false}}
\algnewcommand\OR{\textbf{or}\space}
\algnewcommand\AND{\textbf{and}\space}
\algnewcommand{\IFTHENELSE}[3]{\STATE \algorithmicif\ #1\ \algorithmicthen\ #2\ \algorithmicelse\ #3 \algorithmicend\ \algorithmicif}

\lstdefinelanguage{Prism} {
    keywords=[0]{const,global,bool,int,double,dtmc,ctmc,mdp,csg,rewards,endrewards,formula,module,endmodule,player,endplayer,init,true,false},
    basicstyle=\linespread{1.0}\footnotesize\ttfamily\color{black},
    keywordstyle=[0]\ttfamily\color{black},
    comment=[l][\color{gray}]{//},
    sensitive=true,
    frame=single,	
    numbers=left,
    numberstyle=\tiny\color{gray},
    morestring=[b]"
}   

\title{Automatic Verification of Concurrent Stochastic Systems}
\author{Marta Kwiatkowska \and Gethin Norman \and David Parker \and Gabriel Santos}
\institute{Marta Kwiatkowska \and Gabriel Santos \at Department of Computing Science, University of Oxford, UK
\and
Gethin Norman \at School of Computing Science, University of Glasgow, UK
\and
David Parker
\at School of Computer Science, University of Birmingham, UK
}

\renewcommand{\emptyset}{\varnothing}

\DeclareMathOperator{\opt}{opt}

\usepackage{scalerel}

\newcommand\scale[2]{\vstretch{#1}{\hstretch{#1}{#2}}}

\newcommand{\appref}[1]{Appendix~\ref{#1}}
\newcommand{\sectref}[1]{Section~\ref{#1}}

\newcommand{\figref}[1]{Figure~\ref{#1}}
\newcommand{\tabref}[1]{Table~\ref{#1}}
\newcommand{\egref}[1]{Example~\ref{#1}}

\newcommand{\eqnref}[1]{(\ref{#1})}

\newcommand{\propref}[1]{Proposition~\ref{#1}}
\newcommand{\lemref}[1]{Lemma~\ref{#1}}
\newcommand{\defref}[1]{Definition~\ref{#1}}
\newcommand{\assumref}[1]{Assumption~\ref{#1}}

\newcommand{\sectsectref}[2]{Sections~\ref{#1} and~\ref{#2}}
\newcommand{\figfigref}[2]{Figures~\ref{#1} and~\ref{#2}}
\newcommand{\tabtabref}[2]{Tables~\ref{#1} and~\ref{#2}}
\newcommand{\defdefref}[2]{Definitions~\ref{#1} and~\ref{#2}}

\newcommand{\appappref}[2]{Appendices~\ref{#1} and~\ref{#2}}
\newcommand{\assumassumref}[2]{Assumptions~\ref{#1} and~\ref{#2}}

\spnewtheorem{assumption}{Assumption}{\bfseries}{\itshape}
\newcounter{exampcount}
\newenvironment{examp}
{\refstepcounter{exampcount}
\vskip6pt\noindent
{\bf Example \arabic{exampcount}.}}
{\hfill$\blacksquare$\vskip6pt}

\newcommand{\startpara}[1]{{%
\vskip6pt\noindent
{\bf #1.}}}

\def\ra{{\rightarrow}}

\def\cC{{\mathcal{C}}}

\def\cF{{\mathcal{F}}}

\def\Nset{\mathbb{N}}
\def\Qset{\mathbb{Q}}

\def\Eset{\mathbb{E}}

\def\ra{\rightarrow} %
\def\rmdef{\,\stackrel{\mbox{\rm {\tiny def}}}{=}}
\newcommand{\true}{\mathtt{true}} %

\renewcommand{\leq}{\leqslant}
\renewcommand{\geq}{\geqslant}

\newcommand\game{{\sf G}}
\newcommand\nfgame{{\sf N}}
\newcommand\mgame{{\sf Z}}

\newcommand\sinit{{\bar{s}}}

\newcommand\dist{{\mathit{Dist}}}

\newcommand\Prob{{\mathit{Prob}}}

\newcommand\val{{\mathit{val}}}

\newcommand{\last}{\mathit{last}}
\newcommand{\ipaths}{\mathit{IPaths}}
\newcommand{\fpaths}{\mathit{FPaths}}

\newcommand{\swne}{SWNE\xspace}
\newcommand{\scne}{SCNE\xspace}
\newcommand{\eswne}{$\varepsilon$-SWNE\xspace}
\newcommand{\escne}{$\varepsilon$-SCNE\xspace}

\def\AP{{\mathit{AP}}}
\def\sat{{\,\models\,}}

\def\sateps{{\,\models\,}}
\def\notsateps{{\,\not\models\,}}

\def\Sat{{\mathit{Sat}}}
\newcommand{\coalition}[1]{\langle \! \langle {#1} \rangle \! \rangle}
\def\notsat{{\,\not\models\,}}

\def\next{{X\,}}
\def\until{{\ {\cal U}\ }}
\def\buntil{{\ {\cal U}^{\leq k}\ }}

\def\next{{\mathtt X\,}}
\def\until{{\ \mathtt{U}\ }}
\def\untilop{{\mathtt{U}}}

\def\buntil{{\ \mathtt{U}^{\leq k}\ }}
\def\buntilop{{\mathtt{U}^{\leq k}}}

\newcommand{\buntilp}[1]{{\ \mathtt{U}^{\leq #1}\, }}

\def\future{{\mathtt{F}\ }}
\def\futureop{{\mathtt{F}}}

\def\bfuture{{\mathtt{F}^{\leq k}\ }}

\def\bfutureop{{\mathtt{F}^{\leq k}}}

\newcommand\bfuturep[1]{{\mathtt{F}^{\leq #1}\ }}
\def\globally{{\mathtt{G}\ }}

\newcommand{\scumul}[1]{\mathtt{C}^{#1}} %

\newcommand{\ap}{\mathsf{a}}
\newcommand{\sinstant}[1]{\mathtt{I}^{#1}} %

\newcommand\V{{\mathtt V}}
\newcommand\probopP{{\mathtt P}}
\newcommand\nashop[3]{\coalition{#1}_{#2}(#3)}
\newcommand\snashop[3]{\coalition{#1}_{\scale{.75}{#2}}(#3)}

\newcommand\probop[2]{\probopP_{#1}[\,{#2}\,]}

\newcommand\rewopR{{\mathtt R}}
\newcommand\rewop[3]{\rewopR^{#1}_{#2}[\,{#3}\,]}

\newcommand{\lab}{{\mathit{L}}}

\newcommand{\rew}{\mathit{rew}}
\makeatletter
\patchcmd{\@addmarginpar}{\ifodd\c@page}{\ifodd\c@page\@tempcnta\m@ne}{}{}
\makeatother
\reversemarginpar
\begin{document}

\maketitle

\begin{abstract}
Automated verification techniques for stochastic games
allow formal reasoning about %
systems
that feature competitive or collaborative behaviour among
rational agents in uncertain or probabilistic settings.
Existing tools and techniques focus on turn-based games,
where each state of the game is controlled by a single player,
and on zero-sum properties, where two players or coalitions have directly opposing objectives.
In this paper, we present automated verification techniques
for concurrent stochastic games (CSGs), which provide a more natural
model of concurrent decision making and interaction.
We also consider (social welfare) Nash equilibria, to formally identify
scenarios where two players or coalitions with distinct goals can collaborate
to optimise their joint performance.
We propose an extension of the temporal logic rPATL for specifying
quantitative %
properties in this setting
and present corresponding algorithms for verification and strategy synthesis for a variant of stopping games. For finite-horizon properties the computation is exact, while for infinite-horizon it is approximate using value iteration.
For zero-sum properties it requires solving matrix games via linear programming, and for equilibria-based properties we find social welfare or social cost Nash equilibria of bimatrix games via the method of labelled polytopes through an SMT encoding.
We implement this approach in PRISM-games, which required extending the tool's modelling language for CSGs,
and apply it to case studies from domains including
robotics, computer security and computer networks,  
explicitly demonstrating the benefits of both CSGs and equilibria-based properties.

\end{abstract}

\section{Introduction}

Stochastic multi-player games are a versatile modelling framework for systems that exhibit cooperative or competitive behaviour in the presence of adversarial or uncertain environments.
They can be viewed as a collection of players (agents) with strategies for determining their actions based on the execution so far.
These models combine \emph{nondeterminism}, representing the adversarial, cooperative and competitive choices,
\emph{stochasticity}, modelling uncertainty due to noise, failures or randomness,
and \emph{concurrency}, representing simultaneous execution of interacting agents.
Examples of such systems appear in many domains, from robotics and autonomous transport, to security and computer networks.
A game-theoretic approach also facilitates the design of protocols
that use penalties or incentives to ensure robustness against selfish participants.
However, the complex interactions involved in such systems make their correct construction a challenge.

\emph{Formal verification} for stochastic games provides a means of producing quantitative guarantees on the correctness of these systems (e.g.\ ``the control software can always safely stop the vehicle with probability at least 0.99, regardless of the actions of other road users''), where the required behavioural properties are specified precisely in quantitative extensions of temporal logic. The closely related problem of \emph{strategy synthesis} constructs an optimal strategy for a player, or coalition of players, which guarantees that such a property is satisfied.

A variety of verification algorithms for stochastic games have been devised, e.g.,~\cite{Cha07b,Umm10,AHK07,AM04,CAH13}.
In recent years, further progress has been made:
verification and strategy synthesis algorithms have been developed
for various temporal logics~\cite{CFK+13b,BKW17,KNPS18,KKKW18}
and implemented in the PRISM-games tool~\cite{KPW18},
an extension of the PRISM probabilistic model checker~\cite{KNP11}.
This has allowed modelling and verification of stochastic games
to be used for a variety of non-trivial applications,
in which competitive or collaborative behaviour between entities is a crucial ingredient, including computer security and energy management.

A limitation of the techniques implemented in PRISM-games to date
is that they focus on \emph{turn-based} stochastic multi-player games (TSGs),
whose states are partitioned among a set of players,
with exactly one player taking control of each state.
In this paper, we propose and implement techniques for
\emph{concurrent stochastic multi-player games} (CSGs),
which generalise TSGs by permitting players to choose their actions
simultaneously in each state. This provides a more realistic model of interactive agents operating concurrently, and making action choices without already knowing the actions being taken by other agents.
Although algorithms for CSGs
have been known for some time (e.g., \cite{AHK07,AM04,CAH13}),
their implementation and application to real-world examples has been lacking.

A further limitation of existing work is that it focuses on \emph{zero-sum} properties,
in which one player (or a coalition of players) aims to optimise some objective,
while the remaining players have the directly opposing goal.
In PRISM-games, properties are specified in the logic rPATL (probabilistic alternating-time temporal logic with rewards)~\cite{CFK+13b}, a quantitative extension of the game logic ATL~\cite{AHK02}. This allows us to specify that a coalition of players can achieve a high-level objective, regarding the probability of an event's occurrence or the expectation of reward measure, irrespective of the other players' strategies. Extensions have allowed players to optimise multiple objectives~\cite{CFKSW13,BKW17}, but again in a zero-sum fashion.

In this work, we move beyond zero-sum properties
and consider situations where two players (or two coalitions of players)
in a CSG have distinct objectives to be maximised or minimised.
The goals of the players (or coalitions) are not necessarily directly opposing,
and so it may be beneficial for players to collaborate.
For these \emph{nonzero-sum} scenarios,
we use the well studied notion of \emph{Nash equilibria} (NE),
where it is not beneficial for any player to unilaterally change their strategy.
In particular, we use subgame-perfect NE~\cite{OR04},
where this equilibrium criterion holds in every state of the game,
and we focus on two specific variants of equilibria:
\emph{social welfare} and \emph{social cost} NE,
which maximise and minimise, respectively,
the sum of the objectives of the players.

We propose an extension of the rPATL logic
which adds the ability to express quantitative \emph{nonzero-sum} properties
based on these notions of equilibria,
for example ``the two robots have navigation strategies which form a
(social cost) Nash equilibrium, and under which the combined expected energy usage until completing their tasks is below $k$''.
We also include some additional reward properties
that have proved to be useful when applying our methods to various case studies.

We provide a formal semantics for the new logic
and propose algorithms for CSG verification and strategy synthesis for a variant of stopping games,
including both zero-sum and nonzero-sum properties.
Our algorithms extend the existing approaches for rPATL model checking,
and employ a combination of
exact computation through backward induction for finite-horizon properties
and approximate computation through value iteration for infinite-horizon properties.
Both approaches require the solution of games for each state of the model
in each iteration of the computation:
we solve matrix games for the zero-sum case
and find optimal NE for bimatrix games for the nonzero-sum case.
The former can be done with linear programming;
we perform the latter using labelled polytopes~\cite{LH64} and a reduction to SMT.

We have implemented our verification and strategy synthesis algorithms
in a new release, version 3.0, of PRISM-games~\cite{KNPS20},
extending both the modelling and property specification languages to
support CSGs and nonzero-sum properties.
In order to investigate the performance, scalability and applicability
of our techniques, we have developed a large selection of case studies
taken from a diverse set of application domains including:
finance, computer security, computer networks, communication systems,
robot navigation and power control.

These illustrate examples of systems whose modelling and analysis requires stochasticity
\emph{and} competitive or collaborative behaviour between concurrent components or agents. %
We demonstrate that our CSG modelling and verification techniques facilitate
insightful analysis of quantitative aspects of such systems.
Specifically, we show cases where CSGs allow more accurate modelling of concurrent
behaviour than their turn-based counterparts
and where our equilibria-based extension of rPATL
allows us to synthesise better performing strategies
for collaborating agents than can be achieved using the zero-sum version.

\vskip0.7em
The paper combines and extends the conference papers~\cite{KNPS18} and \cite{KNPS19}.
In particular, we:
(i) introduce the definition of social cost Nash equilibria for CSGs and model checking algorithms for verifying temporal logic specifications using this definition;
(ii) provide additional details and proofs of model checking algorithms,
for example for combinations of finite- and infinite-horizon objectives;
(iii) present an expanded experimental evaluation, 
including a wider range of properties, extended analysis of the case studies
and a more detailed evaluation of performance, including efficiency
improvements with respect to~\cite{KNPS18,KNPS19}.

\startpara{Related work}
Various verification algorithms have been proposed for CSGs,
e.g.~\cite{AHK07,AM04,CAH13},
but without implementations, tool support or case studies.
PRISM-games~2.0~\cite{KPW18}, which we have built upon in this work, 
provided modelling and verification for a wide range of properties of stochastic multi-player games,
including those in the logic rPATL, and multi-objective extensions of it,
but focusing purely on the turn-based variant of the model (TSGs) in the context of two-coalitional zero-sum properties.
GIST~\cite{CHJ+10} allows the analysis of $\omega$-regular properties on probabilistic games,
but again focuses on turn-based, not concurrent, games.
GAVS+~\cite{CKL+11} is a general-purpose tool for algorithmic game solving,
supporting TSGs and (non-stochastic) concurrent games, but not CSGs.
Three further tools, PRALINE~\cite{BRE13}, EAGLE~\cite{TGW15}  and EVE~\cite{GNP+18}, 
support the computation of NE~\cite{Nas50} for the restricted class of (non-stochastic) concurrent games. In addition, EVE has recently been extended to verify if an LTL property holds on some or all NE~\cite{GNPW20}.
Computing NE is also supported by MCMAS-SLK~\cite{CLM+14} via strategy logic and general purpose tools such as Gambit~\cite{GAMB} can compute a variety of equilibria but, again, not for stochastic games.

Work concerning nonzero-sum properties includes~\cite{CMJ04,Umm10}, in which the existence of and the complexity of finding NE for stochastic
games is studied,
but without practical algorithms. The complexity of finding subgame-perfect NE for quantitative reachability properties is studied in~\cite{BBG+19}, while~\cite{GNP+19} considers the complexity of equilibrium design for temporal logic properties and lists social welfare requirements and implementation as future work. In \cite{PPB15}, a learning-based algorithm for finding NE for discounted properties of CSGs is presented and evaluated. Similarly,~\cite{LP17} studies NE for discounted properties and introduces iterative algorithms for strategy synthesis. A theoretical framework for price-taking equilibria of CSGs is given in~\cite{AY17}, where players try to minimise their costs which include a price common to all players and dependent on the decisions of all players.
A notion of strong NE for a restricted class of CSGs is formalised in~\cite{DDJ+18} and an approximation algorithm for checking the existence of such NE for discounted properties is introduced and evaluated. The existence of stochastic equilibria with imprecise deviations for CSGs and a PSPACE algorithm to compute such equilibria is considered in~\cite{BMS16}. Finally, we mention the fact that the concept of equilibrium is used to analyze different applications such as cooperation among agents in stochastic games~\cite{HSC+18} and to design protocols based on quantum secret sharing \cite{QTT18}.

\section{Preliminaries}\label{prelim-sect}

We begin with some basic background from game theory,
and then describe CSGs, illustrating each with examples.
For any finite set $X$, we will write $\dist(X)$ for the set of probability distributions over $X$ and for any vector $v \in \Qset^n$ for $n \in \Nset$ we use $v(i)$ to denote the $i$th entry of the vector.

\subsection{Game Theory Concepts}

We first introduce \emph{normal form games},
which are simple one-shot games where players make their choices concurrently.
\begin{definition}[Normal form game] A (finite, $n$-person) \emph{normal form game} (NFG) is a tuple $\nfgame = (N,A,u)$ where:
\begin{itemize}
\item $N=\{1,\dots,n\}$ is a finite set of \emph{players};
\item $A = A_1 \times \cdots \times A_n$ and $A_i$ is a finite set of \emph{actions} available to player $i \in N$;
\item $u = (u_1,\dots,u_n)$ and $u_i \colon A \rightarrow \Qset$ is a \emph{utility function} for player $i \in N$.
\end{itemize}
\end{definition}
In a game $\nfgame$, players select actions simultaneously, with player $i \in N$ choosing from the action set $A_i$. If each player $i$ selects action $a_i$, then player $j$ receives the utility $u_j(a_1,\dots,a_n)$.
\begin{definition}[Strategies and strategy profile] 
A \emph{(mixed) strategy} $\sigma_i$ for player $i$ in an NFG $\nfgame$ is a distribution over its action set, i.e., $\sigma_i \in \dist(A_i)$. We let $\Sigma^i_\nfgame$ denote the set of all strategies for player $i$.
A \emph{strategy profile} (or just \emph{profile})
$\sigma = (\sigma_1,\dots,\sigma_n)$ is a tuple of strategies for each player.
\end{definition}
Under a strategy profile $\sigma = (\sigma_1,\dots,\sigma_n)$ of an NFG $\nfgame$, the expected utility of player $i$ is defined as follows:
\[ \begin{array}{c}
u_i(\sigma) \ \rmdef \ \sum_{(a_1,\dots,a_n) \in A} u_i(a_1,\dots,a_n) \cdot \left( \prod_{j=1}^n \sigma_j(a_j) \right) \, .
\end{array} \]
A two-player NFG is also called a \emph{bimatrix game} as it can be represented by two distinct matrices $\mgame_1, \mgame_2 \in \Qset^{l \times m}$ where $A_1 = \{a_1,\dots,a_l\}$, $A_2 = \{b_1,\dots,b_m\}$, $z^1_{ij} = u_1(a_i,b_j)$ and $z^2_{ij} = u_2(a_i,b_j)$. 

A two-player NFG is \emph{constant-sum} if there exists $c \in \Qset$ such that $u_1(\alpha) {+} u_2(\alpha) = c$ for all $\alpha \in A$ and \emph{zero-sum} if $c = 0$. A zero-sum, two-player NFG is often called a \emph{matrix game} as it can be represented by a single matrix $\mgame \in \Qset^{l \times m}$ where $A_1 = \{a_1,\dots,a_l\}$, $A_2 = \{b_1,\dots,b_m\}$ and $z_{ij} = u_1(a_i,b_j) = - u_2(a_i,b_j)$. For zero-sum, two-player NFGs, in the bimatrix game representation we have $\mgame_1 =-\mgame_2$.

\subsubsection{Matrix Games}\label{matrix-sect}

\noindent
We require the following classical result concerning matrix games,
which introduces the notion of the \emph{value} of a matrix game (and zero-sum NFG).

\begin{theorem}[Minimax theorem~\cite{NEU28,NMK+44}]\label{minimax-thm}
For any zero-sum NFG $\nfgame = (N,A,u)$ and corresponding matrix game $\mgame$, there exists $v^\star \in \Qset$, called the \emph{value} of the game and denoted $\val(\mgame)$, such that:
\begin{itemize}
\item there is a strategy $\sigma_1^\star$ for player 1, called an optimal strategy of player 1, such that under this strategy the player's expected utility is at least $v^\star$ regardless of the strategy of player 2, i.e.\ $\inf_{\sigma_2 \in \Sigma^2_\nfgame} u_1(\sigma_1^\star,\sigma_2) \geq v^\star$;
\item there is a strategy $\sigma_2^\star$ for player 2, called an optimal strategy of player 2, such that under this strategy the player's expected utility is at least $-v^\star$ regardless of the strategy of player 1, i.e.\ $\inf_{\sigma_1 \in \Sigma^1_\nfgame} u_2(\sigma_1,\sigma_2^\star) \geq -v^\star$.
\end{itemize}
\end{theorem}
The value of a matrix game $\mgame \in \Qset^{l \times m}$ can be found by solving the following linear programming (LP) problem~\cite{NEU28,NMK+44}.
Maximise $v$ subject to the constraints:
\begin{eqnarray*}
x_1 {\cdot} z_{1j} + \cdots + x_l {\cdot} z_{lj} & \geq & v \;\; \mbox{for all} \; 1 \leq j \leq m \\
x_i & \geq & 0 \;\; \mbox{for all} \; 1 \leq i \leq l \\ 
 x_1 + \cdots + x_l & = & 1
\end{eqnarray*}
In addition, the solution for $(x_1,\dots,x_l)$ yields an optimal strategy for player 1. The value of the game can also be found by solving the following dual LP problem. Minimise $v$ subject to the constraints:
\begin{eqnarray*}
y_1 {\cdot} z_{i1} + \cdots + y_m {\cdot} z_{im} & \leq & v \;\; \mbox{for all} \; 1 \leq i \leq l \\
y_j & \geq & 0 \;\; \mbox{for all} \; 1 \leq j \leq m \\ 
y_1 + \cdots + y_m & = & 1 
\end{eqnarray*}
and in this case the solution $(y_1,\dots,y_m)$ yields an optimal strategy for player 2.

\begin{examp}\label{rps-eg}
Consider the (zero-sum) NFG corresponding to the well known rock-paper-scissors game,
where each player $i\in\{1,2\}$ chooses ``rock'' ($r_i$), ``paper'' ($p_i$) or ``scissors'' ($s_i$).
The matrix game representation is:
\[
\mgame = \kbordermatrix{
 & r_2 & p_2 & s_2 \cr 
 r_1 & 0 & {-}1 & 1 \cr 
 p_1 & 1 & 0 & {-}1 \cr
 s_1 & {-}1 & 1 & 0}
\]
where the utilities for winning, losing and drawing are $1$, $-1$ and $0$ respectively.
The value for this matrix game is the solution to the following LP problem. Maximise $v$ subject to the constraints:
\begin{eqnarray*}
x_2 - x_3 & \geq & v \\ 
x_3 - x_1 & \geq & v \\ 
x_1 - x_2 & \geq & v \\
x_1,x_2,x_3 & \geq & 0 \\
x_1 + x_2 + x_3 & = & 1 
\end{eqnarray*}
which yields the value $v^\star=0$ with optimal strategy $\sigma_1^\star = (1/3,1/3,1/3)$ for player 1 (the optimal strategy for player 2 is the same).
\end{examp}

\subsubsection{Bimatrix Games}\label{bimatrix-sect}

\noindent
For bimatrix games (and nonzero-sum NFGs), we use the concept of \emph{Nash equilibria} (NE),
which represent scenarios for players with distinct objectives
in which it is not beneficial for any player to unilaterally change their strategy.
In particular, we will use variants called \emph{social welfare optimal} NE and \emph{social cost optimal} NE.
These variants are equilibria that maximise or minimise, respectively, the total utility of the players,
i.e., the sum of the individual player utilities.
\begin{definition}[Best and least response] 
For NFG $\nfgame = (N,A,u)$, strategy profile $\sigma = (\sigma_1,\dots,\sigma_n)$ and player $i$ strategy $\sigma_i'$, we define the sequence of strategies $\sigma_{-i} = (\sigma_1,\dots,\sigma_{i-1},\sigma_{i+1},\dots,\sigma_n)$ and profile $\sigma_{-i}[\sigma_i'] = (\sigma_1,\dots,\sigma_{i-1},\sigma_i',\sigma_{i+1},\dots,\sigma_n)$. For player $i$ and strategy sequence $\sigma_{-i}$:
\begin{itemize}
\item
a \emph{best response} for player $i$ to $\sigma_{-i}$ is a strategy $\sigma^\star_i$ for player $i$ such that $u_i(\sigma_{-i}[\sigma^\star_i]) \geq u_i(\sigma_{-i}[\sigma_i])$ for all strategies $\sigma_i$ of player $i$;
\item
a \emph{least response} for player $i$ to $\sigma_{-i}$ is a strategy $\sigma^\star_i$ for player $i$ such that $u_i(\sigma_{-i}[\sigma^\star_i]) \leq u_i(\sigma_{-i}[\sigma_i])$ for all strategies $\sigma_i$ of player $i$.
\end{itemize}
\end{definition}
\begin{definition}[Nash equilibrium]\label{nash-def}
For NFG $\nfgame = (N,A,u)$, a strategy profile $\sigma^\star$ of $\nfgame$ is a \emph{Nash equilibrium} (NE) and $\langle u_i(\sigma^\star)\rangle_{i \in N}$ \emph{NE values}
if $\sigma_i^\star$ is a best response to $\sigma_{-i}^\star$ for all $i \in N$.
\end{definition}
\begin{definition}[Social welfare NE]\label{swne-def}
For NFG $\nfgame = (N,A,u)$, an NE $\sigma^\star$ of $\nfgame$ is a \emph{social welfare optimal NE} (SWNE) and $\langle u_i(\sigma^\star)\rangle_{i \in N}$ corresponding \emph{SWNE values} if $u_1(\sigma^\star){+}\cdots{+}u_n(\sigma^\star)\geq u_1(\sigma){+} \cdots{+}u_n(\sigma)$ for all NE $\sigma$ of $\nfgame$.
\end{definition}
\begin{definition}[Social cost NE]\label{scne-def}
For NFG $\nfgame = (N,A,u)$, a profile $\sigma^\star$ of $\nfgame$ is a \emph{social cost optimal NE} (SCNE) and $\langle u_i(\sigma^\star)\rangle_{i \in N}$ corresponding \emph{SCNE values} if it is an NE of $\nfgame^{-}= (N,A,{-}u)$ and $u_1(\sigma^\star){+}\cdots{+}u_n(\sigma^\star)\leq u_1(\sigma){+} \cdots{+}u_n(\sigma)$ for all NE $\sigma$ of $\nfgame^{-}= (N,A,{-}u)$.
\end{definition}
The notion of SWNE is standard~\cite{NRTV07} and corresponds to the case where utility values represent profits or rewards.
We introduce the dual notion of SCNE for the case where utility values correspond to losses or costs.
In our experience of modelling with stochastic games,
such situations are common: example objectives in this category include
minimising the probability of a fault occurring or minimising the expected time to complete a task.
Representing SCNE directly is a more natural approach than the alternative of simply negating utilities, as above.

The following demonstrates the relationship between SWNE and SCNE.
\begin{lemma}\label{duality-lem}
For NFG $\nfgame = (N,A,u)$, a strategy profile $\sigma^\star$ of $\nfgame$ is an NE of $\nfgame^{-} = (N,A,{-}u)$ if and only if $\sigma^\star_i$ is a least response to $\sigma^\star_{-i}$ of player $i$ in $\nfgame$ for all $i \in N$. Furthermore, $\sigma^\star$ is a SWNE of $\nfgame^{-}$ if and only if $\sigma^\star$ is a SCNE of $\nfgame$.
\end{lemma}
\lemref{duality-lem} can be used to reduce the computation of SCNE profiles and values to those of SWNE profiles and values (or vice versa). This is achieved by negating all utilities in the NFG or bimatrix game, computing an SWNE profile and corresponding SWNE values, and then negating the SWNE values to obtain an SCNE profile and corresponding SCNE values for the original NFG or bimatrix game.

Finding NE and NE values in bimatrix games is in the class of \emph{linear complementarity} problems (LCPs). More precisely, $(\sigma_1,\sigma_2)$ is an NE profile and $(u,v)$ are the corresponding NE values of the bimatrix game $\mgame_1,\mgame_2 \in \Qset^{l \times m}$ where $A_1 = \{a_1,\dots,a_l\}$, $A_2 = \{b_1,\dots,b_m\}$ if and only if for the column vectors $x \in \Qset^l$ and $y \in \Qset^m$ where $x_i = \sigma_1(a_i)$ and $y_j = \sigma_2(b_j)$ for $1 \leq i \leq l$ and $1 \leq j \leq m$, we have:
\begin{eqnarray}
x^T(\mathbf{1} u - \mgame_1 y) & = & 0 \label{lc1-eqn} \\
y^T(\mathbf{1} v - \mgame_2^T x) & = & 0 \label{lc2-eqn} \\
\mathbf{1} u -\mgame_1 y & \geq & \mathbf{0} \label{lc3-eqn} \\
\mathbf{1} v-\mgame_2^T x & \geq & \mathbf{0} \label{lc4-eqn}
\end{eqnarray}
and $\mathbf{0}$ and $\mathbf{1}$ are vectors or matrices with all components 0 and 1, respectively.

\begin{examp}\label{stag1-eg} We consider a nonzero-sum stag hunt game~\cite{PSS+11} where, if players decide to cooperate, this can yield a large utility, but if the others do not, then the cooperating player gets nothing while the remaining players get a small utility. A scenario with 3 players, where two form a coalition (assuming the role of player 2), yields a bimatrix game: 
\vspace*{-0.2cm}
\[
\mgame_1 = 
\kbordermatrix{
 & \mathit{nc}_2 & \mathit{hc}_2 & c_2 \cr 
 \mathit{nc}_1 & 2 & 2 & 2 \cr 
 c_1 & 0 & 4 & 6 
 }
\qquad
\mgame_2 = 
\kbordermatrix{
 & \mathit{nc}_2 & \mathit{hc}_2 & c_2 \cr 
 \mathit{nc}_1 & 4 & 2 & 0 \cr 
 c_1 & 4 & 6 & 9
 }
\]
where $\mathit{nc}_i$ and $\mathit{c}_i$ represent player 1 and coalition 2 not cooperating and cooperating, respectively, and $\mathit{hc}_2$ represents half the players in the coalition cooperating.
A strategy profile $\sigma^* = ((x_1,x_2,x_3),(y_1,y_2))$ is an NE and $(u,v)$ the corresponding NE values of the game if and only if, from Eqn \eqnref{lc1-eqn} and \eqnref{lc2-eqn}:
\begin{eqnarray*}
u {\cdot}x_1 - 2{\cdot}x_1 {\cdot} y_1 - 2{\cdot}x_1 {\cdot} y_2 - 2{\cdot}x_1 {\cdot} y_3 + u {\cdot}x_2 - 4{\cdot}x_2 {\cdot} y_2 - 6{\cdot}x_2 {\cdot} y_3 & = & 0 \\
v{\cdot}y_1 - 4{\cdot}x_1{\cdot}y_1 - 4{\cdot}x_2{\cdot}y_1 + v{\cdot}y_2 - 2{\cdot}x_1{\cdot}y_2 - 6{\cdot}x2{\cdot}y_2 + v{\cdot}y_3 -9{\cdot}x_2{\cdot}y_3 & = & 0
\end{eqnarray*}
and, from Eqn \eqnref{lc3-eqn} and \eqnref{lc4-eqn}:
\begin{eqnarray*}
u - 2{\cdot}y_1 - 2{\cdot}y_2 - 2{\cdot}y_3 & \geq & 0 \\
u - 4{\cdot}y_2 - 6{\cdot}y_3 & \geq & 0 \\
v - 4{\cdot}x_1 - 4{\cdot}x_2 & \geq & 0 \\
v - 2{\cdot}x_1 - 6{\cdot}x_2 & \geq & 0 \\
v - 9{\cdot}x_2 & \geq & 0 \, .
\end{eqnarray*}
There are three solutions to this LCP problem which correspond to the following NE profiles:
\begin{itemize}
\item player 1 and the coalition pick $\mathit{nc}_1$ and $\mathit{nc}_2$, respectively, with NE values $(2,4)$;
\item player 1 selects $\mathit{nc}_1$ and $c_1$ with probabilities $5/9$ and $4/9$ and the coalition selects $\mathit{nc}_2$ and $c_2$ with probabilities $2/3$ and $1/3$, with NE values $(2,4)$;
\item player 1 and the coalition select $c_1$ and $c_2$, respectively, with NE values $(6,9)$.
\end{itemize}
For instance, in the first case, neither player 1 nor the coalition believes the other will cooperate: the best they can do is act alone. 
The third maximises the joint utility and is the only SWNE profile,
with corresponding SWNE values $(6,9)$.

To find SCNE profiles and SCNE values for the same set of utility functions, using \lemref{duality-lem}
we can negate all the utilities of the players in the game and look for NE profiles in the resulting bimatrix game;
again, there are three:
\begin{itemize}
\item player 1 and the coalition select $c_1$ and $\mathit{nc}_2$, respectively, with NE values $(0,-4)$;
\item player 1 selects $\mathit{nc}_1$ and $c_1$ with probabilities $1/2$ and $1/2$ and the coalition selects $\mathit{nc}_2$ and $\mathit{hc}_2$ with probabilities $1/2$ and $1/2$, with NE values $(-2,-4)$;
\item player 1 and the coalition select $\mathit{nc}_1$ and $c_2$, respectively, with NE values $(-2,0)$.
\end{itemize}
The third is the only SCNE profile, %
with corresponding SCNE values $(2,0)$.
\end{examp}

In this work, we compute the SWNE values for a bimatrix game
(or, via \lemref{duality-lem}, the SCNE values)
by first identifying all the NE values of the game.
For this, we use the Lemke-Howson algorithm~\cite{LH64}, which is based on the method of \emph{labelled polytopes}~\cite{NRTV07}.
Other well-known methods include those based on \emph{support enumeration}~\cite{PNS04} and \emph{regret minimisation} \cite{SGC05}.
Given a bimatrix game $\mgame_1,\mgame_2 \in \Qset^{l \times m}$, we denote the sets of deterministic strategies of players 1 and 2 by $I = \{1,\dots,l\}$ and $M = \{1,\dots,m\}$ and define $J = \{l{+}1,\dots,l{+}m\}$ by mapping $j \in M$
to $l{+}j \in J$. A \emph{label} is then defined as an element of $I \cup J$. The sets of strategies for players 1 and 2 can be represented by:
\[
X = \{x \in \Qset^l \mid (\mathbf{1} x = 1) \wedge (x \geq \mathbf{0}) \} \quad \mbox{and} \quad
Y = \{y \in \Qset^m \mid (\mathbf{1} y = 1) \wedge (y \geq \mathbf{0}) \} \, .
\]
The strategy set $Y$ is then divided into regions $Y(i)$ and $Y(j)$ (polytopes) for $i \in I$ and $j \in J$ such that $Y(i)$ contains strategies for which the deterministic strategy $i$ of player 1 is a best response and $Y(j)$ contain strategies which choose action $j$ with probability $0$:
\[
Y(i) = \{y \in Y \mid \forall k \in I . \; \mgame_1(i,:)y \geq \mgame_1(k,:)y \} \;\; \mbox{and} \; \; Y(j) = \{ y \in Y \mid y_{j-l} = 0\}
\]
where $\mgame_1(i,:)$ is the $i$th row vector of $\mgame_1$. A vector $y$ is then said to have label $k$ if $y \in Y(k)$, for $k \in I \cup J$.
The strategy set $X$ is divided analogously into regions $X(j)$ and $X(i)$ for $j \in J$ and $i\in I$ and a vector $x$ has label $k$ if $x \in X(k)$, for $k \in I \cup J$.
A pair of vectors $(x,y) \in X {\times} Y$ is \emph{completely labelled} if the union of the labels of $x$ and $y$ equals $I \cup J$.

The NE profiles of the game are the vector pairs that are completely labelled~\cite{LH64,SHP74}. 
The corresponding NE values can be computed through matrix-vector multiplication.
A SWNE profile and corresponding SWNE values can then be found through an NE profile with NE values that maximise the sum.

\subsection{Concurrent Stochastic Games}\label{csgs-sec}

We now define concurrent stochastic games~\cite{Sha53}, where players repeatedly make simultaneous choices over actions
that update the game state probabilistically.
\begin{definition}[Concurrent stochastic game] A \emph{concurrent stochastic multi-player game} (CSG) is a tuple
$\game = (N, S, \bar{S}, A, \Delta, \delta, \AP, \lab)$ where:
\begin{itemize}
\item $N=\{1,\dots,n\}$ is a finite set of players;
\item $S$ is a finite set of states and $\bar{S} \subseteq S$ is a set of initial states;
\item $A = (A_1\cup\{\bot\}) {\times} \cdots {\times} (A_n\cup\{\bot\})$ where $A_i$ is a finite set of actions available to player $i \in N$ and $\bot$ is an idle action disjoint from the set $\cup_{i=1}^n A_i$;
\item $\Delta \colon S \rightarrow 2^{\cup_{i=1}^n A_i}$ is an action assignment function;
\item $\delta \colon S {\times} A \rightarrow \dist(S)$ is a probabilistic transition function;
\item $\AP$ is a set of atomic propositions and $\lab \colon S \rightarrow 2^{\AP}$ is a labelling function.
\end{itemize}
\end{definition}
A CSG $\game$ starts in an initial state $\sinit \in \bar{S}$ and, when in state $s$, each player $i \in N$ selects an action from its available actions $A_i(s) \rmdef \Delta(s) \cap A_i$ if this set is non-empty, and from $\{ \bot \}$ otherwise. Supposing each player $i$ selects action $a_i$, the state of the game is updated according to the distribution $\delta(s,(a_1,\dots,a_n))$. A CSG is a \emph{turn-based stochastic multi-player game} (TSG) if for any state $s$ there is precisely one player $i$ for which $A_i(s) \neq \{ \bot \}$. Furthermore, a CSG is a \emph{Markov decision process} (MDP) if there is precisely one player $i$ such that $A_i(s) \neq \{ \bot \}$ for all states $s$.

A \emph{path} $\pi$ of $\game$ is a sequence $\pi = s_0 \xrightarrow{\alpha_0}s_1 \xrightarrow{\alpha_1} \cdots$
where $s_i \in S$, $\alpha_i\in A$ %
and
$\delta(s_i,\alpha_i)(s_{i+1})>0$ for all $i \geq 0$.
We denote by $\pi(i)$ the $(i{+}1)$th state of $\pi$, $\pi[i]$ the action associated with the $(i{+}1)$th transition and, if $\pi$ is finite, $\last(\pi)$ the final state. The length of a path $\pi$, denoted $|\pi|$, is the number of transitions appearing in the path. Let $\fpaths_\game$ and $\ipaths_\game$ ($\fpaths_{\game,s}$ and $\ipaths_{\game,s}$) be the sets of finite and infinite paths (starting in state $s$).

We augment CSGs with \emph{reward structures} of the form $r=(r_A,r_S)$, where $r_A \colon S{\times}A \ra \Qset$ is an action reward function (which maps each state and action tuple pair to a rational value that is accumulated when the action tuple is selected in the state) and $r_S \colon S \ra \Qset$ is a state reward function (which maps each state to a rational value that is incurred when the state is reached). 
We allow both positive and negative rewards; however, we will later impose certain restrictions to ensure the correctness of our model checking algorithms.

A \emph{strategy} for a player in a CSG resolves the player's choices in each state. These choices can depend on the history of the CSG's execution and can be randomised. Formally, we have the following definition.
\begin{definition}[Strategy]
A \emph{strategy} for player $i$ in a CSG $\game$ is a function of the form $\sigma_i \colon \fpaths_{\game} \ra \dist(A_i \cup \{ \bot \})$ such that, if $\sigma_i(\pi)(a_i)>0$, then $a_i \in A_i(\last(\pi))$. We denote by $\Sigma^i_\game$ the set of all strategies for player $i$.
\end{definition}

As for NFGs, a \emph{strategy profile} for $\game$ is a tuple $\sigma = (\sigma_1,\dots,\sigma_{n})$ of strategies for all players and, for player $i$ and strategy $\sigma_i'$, we define the sequence $\sigma_{-i}$ and profile $\sigma_{-i}[\sigma_i']$ in the same way.
For strategy profile $\sigma=(\sigma_1,\dots,\sigma_{n})$ and state $s$, we let
$\fpaths^\sigma_{\game,s}$ and
$\ipaths^\sigma_{\game,s}$ denote the
finite and
infinite paths from $s$ under the choices of $\sigma$.
We can define a probability measure $\Prob^{\sigma}_{\game,s}$ over the infinite paths $\ipaths^{\sigma}_{\game,s}$~\cite{KSK76}. This construction is based on first defining the probabilities for finite paths from the probabilistic transition function and choices of the strategies in the profile. More precisely, for a finite path $\pi = s_0 \xrightarrow{\alpha_0}s_1 \xrightarrow{\alpha_1} \cdots \xrightarrow{\alpha_{m-1}} s_m$ where $s_0=s$, the probability of $\pi$ under the profile $\sigma$ is defined by:
\[ \begin{array}{c}
\mathbf{P}^\sigma(\pi) \ \rmdef \ \prod_{j=0}^{m-1} \Big(  \Big( \prod_{i=1}^n \sigma_i (s_0 \xrightarrow{\alpha_0} \cdots \xrightarrow{\alpha_{j-1}} s_j)(\alpha_j(i))  \Big) \cdot \delta(s_j,\alpha_j)(s_{j+1}) \Big) \, .
\end{array} \]
Next, for each finite path $\pi$, we define the basic cylinder $C^\sigma(\pi)$ that consists of all infinite paths in $\ipaths^\sigma_{\game,s}$ that have $\pi$ as a prefix. Finally, using properties of cylinders, we can then construct the probability space $(\ipaths^{\sigma}_{\game,s}, \cF^\sigma_s, \Prob^{\sigma}_{\game,s})$, where $\cF^\sigma_s$ is the smallest $\sigma$-algebra generated by the set of basic cylinders $\{ C^\sigma(\pi) \mid \pi \in \ipaths^{\sigma}_{\game,s} \}$ and $Prob^{\sigma}_{\game,s}$ is the unique measure such that $Prob^{\sigma}_{\game,s}(C^\sigma(\pi)) = \mathbf{P}^\sigma(\pi)$ for all $\pi \in \fpaths^\sigma_{\game,s}$.

For random variable $X \colon \ipaths_{\game} \rightarrow \Qset$, we can then define for any profile $\sigma$ and state $s$ the expected value $\Eset^{\sigma}_{\game,s}(X)$ of $X$ in $s$ with respect to $\sigma$.
These random variables $X$ %
represent an \emph{objective} (or utility function) for a player, %
which includes both \emph{finite-horizon} and \emph{infinite-horizon} properties. Examples of finite-horizon properties include the probability of reaching a set of target states $T$ within $k$ steps or the expected reward accumulated over $k$ steps. These properties can be expressed by the random variables:
\begin{eqnarray*}
X(\pi) &=& 
\begin{cases}
1 & \mbox{if $\pi(j) \in T$ for some $j \leq k$} \\
0 & \mbox{otherwise}
\end{cases}
\\
Y(\pi) &=& \mbox{$\sum_{i=0}^{k-1}$} \big( r_A(\pi(i),\pi[i])+r_S(\pi(i)) \big) 
\end{eqnarray*}
respectively. Examples of infinite-horizon properties include the probability of reaching a target set $T$ and the expected cumulative reward until reaching a target set $T$ (where paths that never reach the target have infinite reward), which can be expressed by the random variables:
\begin{eqnarray*}
X(\pi) &=& \begin{cases}
1 & \mbox{if $\pi(j) \in T$ for some $j \in \Nset$} \\
0 & \mbox{otherwise}
\end{cases}
\\
Y(\pi) &=&
\begin{cases}
\sum_{i=0}^{k_{\min}-1} \big( r_A(\pi(i),\pi[i])+r_S(\pi(i)) \big) & \mbox{if $\pi(j) \in T$ for some $j \in \Nset$} \\
\infty & \mbox{otherwise}
\end{cases}
\end{eqnarray*}
where $k_{\min} = \min \{ j \in \Nset \mid \pi(j) \in T \}$, respectively.

Let us first focus on \emph{zero-sum} games, which are by definition two-player games.
As for NFGs (see \defref{minimax-thm}), for a two-player CSG $\game$ and a given objective $X$, we can consider the case where player 1 tries to maximise the expected value of $X$, while player 2 tries to minimise it. The above definition yields the \emph{value} of $\game$ with respect to $X$ if it is determined, i.e., if the maximum value that player 1 can ensure equals the minimum value that player 2 can ensure. Since the CSGs we consider are finite state and finitely-branching, it follows that they are determined for all the objectives we consider~\cite{Mar98}.
Formally we have the following.
\begin{definition}[Determinacy and optimality]\label{determined-def}
For a two-player CSG $\game$ and objective $X$, we say that $\game$ is \emph{determined} with respect to $X$ if, for any state $s$:
\[ \begin{array}{rcl}
\sup_{\sigma_1 \in \Sigma^1} \inf_{\sigma_2 \in \Sigma^2} \Eset^{\sigma_1,\sigma_2}_{\game,s}(X) & \; = \; & \inf_{\sigma_2 \in \Sigma^2} \sup_{\sigma_1 \in \Sigma^1} \Eset^{\sigma_1,\sigma_2}_{\game,s}(X) \, 
\end{array} \]
and call this the \emph{value} of $\game$ in state $s$ with respect to $X$, denoted $\val_\game(s,X)$.
Furthermore, a strategy $\sigma_1^\star$ of player 1 is \emph{optimal} with respect to $X$ if we have\/
$\smash{\Eset^{\sigma_1^\star,\sigma_2}_{\game,s}(X) \geq \val_\game(s,X)}$ for all $s\in S$ and $\sigma_2 \in \Sigma^2$
and a strategy of player 2 is \emph{optimal} with respect to $X$ if\/ $\smash{\Eset^{\sigma_1,\sigma_2^\star}_{\game,s}(X) \leq \val_\game(s,X)}$ for all $s\in S$ and $\sigma_1 \in \Sigma^1$.
\end{definition}

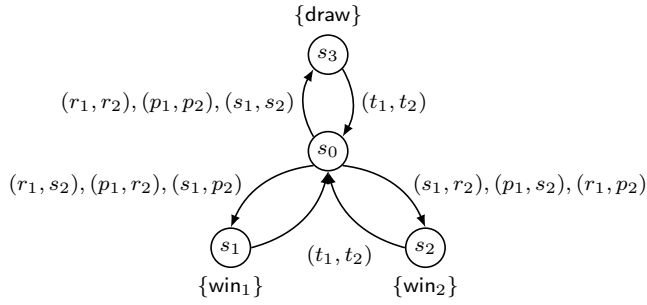
\begin{figure}[t]
\centering
\begin{tikzpicture}[->,>=latex,auto,node distance=2.8cm, semithick, scale=.40]
 \tikzstyle{every state}=[draw=black,text=black, initial text=]

\small
\node[state, inner sep=2pt,minimum size=0pt] (S)at(0,3.2) (s0) {$s_0$}; 
\node[state, inner sep=2pt,minimum size=0pt, label={below:$\{\mathsf{win}_1 \}$}] (S)at(-3.2,0) (s1) {$s_1$};
\node[state, inner sep=2pt,minimum size=0pt, label={below:$\{\mathsf{win}_2 \}$}] (S)at(3.2,0) (s2) {$s_2$}; 
\node[state, inner sep=2pt,minimum size=0pt, label={above:$\{\mathsf{draw}\}$}] (S)at(0,6.4) (s3) {$s_3$};
\path [->] (s0.north west) [bend left]
edge node [left] {$(r_1,r_2), (p_1,p_2),(s_1,s_2)$} (s3.south west);
\path [->] (s3.south east) [bend left]
edge node [right] {$(t_1,t_2)$} (s0.north east);
\path [->] (s0.south west) [bend right]
edge node[left,xshift=-1mm] {$(r_1,s_2),(p_1,r_2),(s_1,p_2)\,$} (s1.north);
\path [->] (s0.south east) [bend left]
edge node[right,xshift=1mm] {$\,(s_1,r_2),(p_1,s_2),(r_1,p_2)$} (s2.north); 

\path [->] (s1.east) [bend right]
edge[pos=0.5,swap] node[near start,xshift=2.5mm] {$(t_1,t_2)$} (s0.south);

\path [->] (s2.west) [bend left]
edge[pos=0.5,swap] (s0.south);

\end{tikzpicture}
\vspace*{-0.0cm}
\caption{Rock-paper-scissors CSG.}\label{fig:rps-csg}
\vspace*{-0.2cm}
\end{figure}

\begin{examp}\label{rps-eg-csg}
Consider the (non-probabilistic) CSG shown in \figref{fig:rps-csg} corresponding to two players repeatedly playing the rock-paper-scissors game (see \egref{rps-eg}). Transitions are labelled with action pairs, where $A_i = \{r_i,p_i,s_i,t_i\}$ for $1 \leq i \leq 2$, with $r_i$, $p_i$ and $s_i$ representing playing rock, paper and scissors, respectively, and $t_i$ restarting the game. The CSG starts in state $s_0$ and states $s_1$, $s_2$ and $s_3$ are labelled with atomic propositions corresponding to when a player wins or there is a draw in a round of the rock-paper-scissors game.

For the zero-sum objective to maximise the probability of reaching $s_1$ before $s_2$, i.e. player 1 winning a round of the game before player 2, the value of the game is $1/2$ and the optimal strategy of each player $i$ is to choose $r_i$, $p_i$ and $s_i$, each with probability $1/3$ in state $s_0$ and $t_i$ otherwise.
\end{examp}

For \emph{nonzero-sum} CSGs, with an objective $X_i$ for each player $i$, we will %
use NE,
which can be defined as for NFGs (see \defref{nash-def}).
In line with the definition of zero-sum optimality above
(and because the model checking algorithms we will later introduce are based on backward induction~\cite{SW+01,NMK+44}),
we restrict our attention to \emph{subgame-perfect} NE~\cite{OR04}, which are NE in \emph{every state} of the CSG.
\begin{definition}[Subgame-perfect NE] For CSG $\game$, a strategy profile $\sigma^\star$ is a \emph{subgame-perfect Nash equilibrium}
for objectives $\langle X_i \rangle_{i \in N}$ if and only if\/ $\Eset^{\sigma^\star}_{\game,s}(X_i) \geq \sup_{\sigma_i \in \Sigma_i} \Eset^{\sigma^\star_{-i}[\sigma_i]}_{\game,s}(X_i)$ for all $i \in N$ and $s \in S$.
\end{definition}
Furthermore, because we use a variety of objectives, including infinite-horizon objectives,
where the existence of NE is an open problem~\cite{BMS14},
we will in some cases use $\varepsilon$-NE, which do exist for any $\varepsilon>0$ for all the properties we consider.
\begin{definition}[Subgame-perfect $\varepsilon$-NE] For CSG $\game$ and $\varepsilon>0$, a strategy profile $\sigma^\star$ is a \emph{subgame-perfect $\varepsilon$-Nash equilibrium}
for objectives $\langle X_i \rangle_{i \in N}$ if and only if\/ $\Eset^{\sigma^\star}_{\game,s}(X_i) \geq \sup_{\sigma_i \in \Sigma_i} \Eset^{\sigma^\star_{-i}[\sigma_i]}_{\game,s}(X_i) - \varepsilon$ for all $i \in N$ and $s \in S$.
\end{definition}

\begin{examp}\label{mac-eg}
In \cite{BRE13} a non-probabilistic concurrent game is used to model medium access control. Two users with limited energy share a wireless channel and choose to transmit ($t_i$) or wait ($w_i$) and, if both transmit, the transmissions fail due to interference. We extend this to a CSG by assuming that transmissions succeed with probability $q$ if both transmit. \figref{mac-fig} presents a CSG model of the protocol where each user has enough energy for one transmission. The states are labelled with the status of each user, where the first value represents if the user $i$ has transmitted or not transmitted their message ($\mathit{tr}_i$ and $\mathit{nt}_i$ respectively) and the second if there is sufficient energy  to transmit or not ($1$ and $0$ respectively).

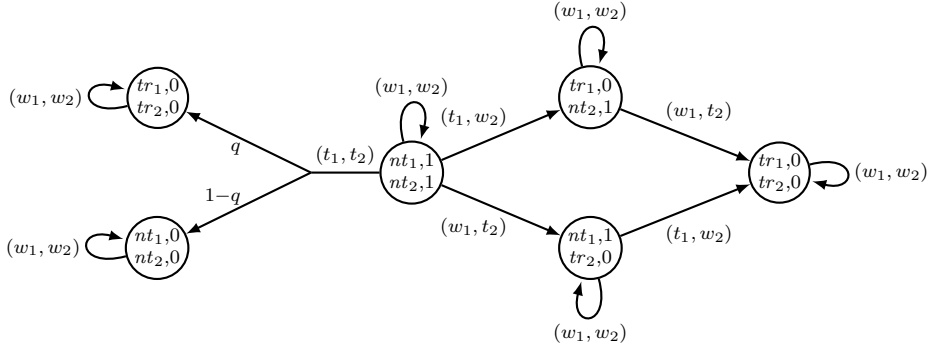
\begin{figure}[t]
\centering
\begin{tikzpicture}[auto,every node/.style={scale=0.9},initial text={},->,>=latex,thick]
\node[state, inner sep=-1pt, minimum size=0pt] (s0) {$\begin{array}{c}\mathit{nt}_1{,}1 \\ \mathit{nt}_2{,}1 \end{array}$};

\node[scale=0.02, left = 0.9 of s0] (nod1) {};
\draw[-] (s0) edge node [above]{$(t_1,t_2)$}(nod1);

\node[state, inner sep=-1pt, minimum size=0pt, above left=0.4 and 2.75 of s0] (s1) {$\begin{array}{c}\mathit{tr}_1{,}0 \\ \mathit{tr}_2{,}0 \end{array}$};
\node[state, inner sep=-1pt, minimum size=0pt, below left=0.4 and 2.75 of s0] (s2) {$\begin{array}{c}\mathit{nt}_1{,}0 \\ \mathit{nt}_2{,}0 \end{array}$};
\draw[->] (nod1) edge node[left,xshift=-0mm,yshift=-1mm] {$\;\;\;q$} (s1);
\draw[->] (nod1) edge node[left,xshift=0mm,yshift=1mm] {$\;\;\;1{-}q$} (s2);

\node[state, inner sep=-1pt, minimum size=0pt, above right=0.4 and 1.75 of s0] (s3) {$\begin{array}{c}\mathit{tr}_1{,}0 \\ \mathit{nt}_2{,}1 \end{array}$};

\node[state, inner sep=-1pt, minimum size=0pt, below right=0.4 and 1.75 of s0] (s4) {$\begin{array}{c}\mathit{nt}_1{,}1 \\ \mathit{tr}_2{,}0 \end{array}$};

\node[state, inner sep=-1pt, minimum size=0pt, right=4
of s0] (s5) {$\begin{array}{c}\mathit{tr}_1{,}0 \\ \mathit{tr}_2{,}0 \end{array}$};

\draw[->] (s0) edge node [above,xshift=-4mm]{$(t_1,w_2)$}(s3);
\draw[->] (s0) edge node [below,xshift=-4mm]{$(w_1,t_2)$}(s4);

\draw[->] (s4) edge node [below,xshift=2mm,yshift=-1mm]{$(t_1,w_2)$}(s5);
\draw[->] (s3) edge node [above,xshift=2mm,yshift=1mm]{$(w_1,t_2)$}(s5);

\draw[->] (s0) edge [loop above] node [above] {$(w_1,w_2)$} ();
\draw[->] (s3) edge [loop above] node [above,xshift=0mm,yshift=-0mm] {$(w_1,w_2)$} ();
\draw[->] (s4) edge [loop below] node [below,xshift=0mm,yshift=0mm] {$(w_1,w_2)$} ();

\draw[->] (s1) edge [loop left] node [left] {$(w_1,w_2)$} ();
\draw[->] (s2) edge [loop left] node [left] {$(w_1,w_2)$} ();
\draw[->] (s5) edge [loop right] node [right] {$(w_1,w_2)$} ();
\end{tikzpicture}
\vspace*{-0.4cm}
\caption{CSG model of a medium access control problem.}\label{mac-fig}
\vspace*{-0.4cm}
\end{figure}

If the objectives are to maximise the probability of a successful transmission, there are two subgame-perfect SWNE profiles, one when user 1 waits for user 2 to transmit before transmitting and another when user 2 waits for user 1 to transmit before transmitting. Under both profiles, both users successfully transmit with probability 1. If the objectives are to maximise the probability of being one of the first to transmit, then there is only one SWNE profile corresponding to both users immediately trying to transmit. In this case the probability of each user successfully transmitting is $q$.
\end{examp}

\section{Property Specification: Extending the Logic rPATL}\label{logic-sect}

In order to formalise properties of CSGs,
we propose an extension of the logic rPATL,
previously defined for zero-sum properties of TSGs~\cite{CFK+13b}.
In particular, we add operators to specify
nonzero-sum properties, using (social welfare or social cost) Nash equilibria,
and provide a semantics for this extended logic on CSGs.

\begin{definition}[Extended rPATL syntax]
The syntax of our extended version of {\rm rPATL} is given by the grammar:
\begin{eqnarray*}
\phi & \; \coloneqq \; & \mathtt{true} \; \mid \; \ap \; \mid \; \neg \phi \; \mid \; \phi \wedge \phi \; \mid \; \coalition{C}\probop{\sim q}{\psi} \; \mid \; \coalition{C}\rewop{r}{\sim x}{\rho} \; \mid \;   \nashop{C{:}C'}{\opt \sim x}{\theta} \\
\psi & \; \coloneqq \; & \next \phi \ \mid \ \phi \buntil \phi \ \mid \ \phi \until \phi \\
\rho & \; \coloneqq \; & \sinstant{=k} \ \mid \ \scumul{\leq k} 
\ \mid \  \future \phi \\
\theta & \; \coloneqq \; & \probop{}{\psi}{+}\probop{}{\psi} \ \mid \  \rewop{r}{}{\rho}{+}\rewop{r}{}{\rho}  
\end{eqnarray*}
where $\ap$ is an atomic proposition, $C$ and $C'$ are coalitions of players such that $C'  =  N {\setminus} C$, $\opt \in \{ \min,\max\}$, $\sim \,\in \{<, \leq, \geq, >\}$, $q \in\Qset\cap[0, 1]$, $x \in \Qset$, $r$ is a reward structure and $k \in \Nset$. %
\end{definition}
rPATL is a branching-time temporal logic for stochastic games,
which combines the probabilistic operator $\probopP$ of PCTL~\cite{HJ94},
PRISM's reward operator $\rewopR$~\cite{KNP11},
and the coalition operator $\coalition{C}$ of ATL~\cite{AHK02}.
The syntax distinguishes between state ($\phi$), path ($\psi$) and reward ($\rho$) formulae. State formulae are evaluated over states of a CSG, while path and reward formulae are both evaluated over paths.

The core operators from the existing version of rPATL~\cite{CFK+13b} are
$\coalition{C} \probop{\sim q}{\psi}$ and $\coalition{C} \rewop{r}{\sim x}{\rho}$.
A state satisfies a formula $\coalition{C} \probop{\sim q}{\psi}$ if the \emph{coalition} of players $C$ can ensure that the probability of the path formula $\psi$ being satisfied is ${\sim} q$, regardless of the actions of the other players ($N{\setminus}C$) in the game.
A state satisfies a formula $\coalition{C} \rewop{r}{\sim x}{\rho}$ if the players in $C$ can ensure that the expected value of the reward formula $\rho$ for reward structure $r$ is ${\sim} x$, whatever the other players do.
Such properties are inherently \emph{zero-sum} in nature as one coalition tries to maximise an objective (e.g., the probability of $\psi$) and the other tries to minimise it; hence, we call these \emph{zero-sum formulae}.

The most significant extension we make to the rPATL logic is the addition of
\emph{nonzero-sum formulae}.
These take the form $\nashop{C{:}C'}{\opt\sim x}{\theta}$, where
$C$ and $C'$ are two coalitions that represent a partition of the set of players $N$,
and $\theta$ is the sum
of either two probabilistic or two reward objectives.
Their meaning is as follows:
\begin{itemize}
\item
$\nashop{C{:}C'}{\max\sim x}{\theta}$ is satisfied
if there exists a subgame-perfect SWNE profile between coalitions $C$ and $C'$ under which the \emph{sum} of the objectives of $C$ and $C'$ in $\theta$ is ${\sim} x$;
\item
$\nashop{C{:}C'}{\min\sim x}{\theta}$ is satisfied
if there exists a subgame-perfect SCNE profile between coalitions $C$ and $C'$ under which the \emph{sum} of the objectives of $C$ and $C'$ in $\theta$ is ${\sim} x$.
\end{itemize}
\noindent
Like the existing zero-sum formulae, the new nonzero-sum formulae still 
split the players into just two coalitions, $C$ and $C' = N{\setminus}C$.
This means that the model checking algorithm (see \sectref{mc-sect})
reduces to finding equilibria in two-player CSGs,
which is more tractable than for larger numbers of players.
Technically, therefore, we could remove the second coalition $C'$ from the syntax.
However, we retain it for clarity about which coalition corresponds to each of the
two objectives, and to allow a later extension to more than two coalitions~\cite{KNPS20b}.

Both types of formula, zero-sum and nonzero-sum,
are composed of \emph{path} ($\psi$) and \emph{reward} ($\rho$) formulae,
used in probabilistic and reward objectives
included within $\probopP$ and $\rewopR$ operators, respectively.
For path formulae, we follow the existing rPATL syntax from~\cite{CFK+13b} and
allow \emph{next} ($\next \phi$), \emph{bounded until} ($\phi \buntil \phi$) and \emph{unbounded until} ($\phi \until \phi$).
We also allow the usual equivalences such as $\future\phi \equiv \true\until\phi$
(i.e., \emph{probabilistic reachability}) and $\bfuturep{k}\phi \equiv \true\buntil\phi$
(i.e., \emph{bounded probabilistic reachability}).

For reward formulae, we introduce some differences with respect to~\cite{CFK+13b}.
We allow instantaneous (state) reward at the $k$th step (\emph{instantaneous reward} $\sinstant{=k}$),
reward accumulated over $k$ steps (\emph{bounded cumulative reward} $\scumul{\leq k}$),
and reward accumulated until a formula $\phi$ is satisfied (\emph{expected reachability} $\future \phi$).
The first two, adapted from the property specification language of PRISM~\cite{KNP11},
were not previously included in rPATL, 
but proved to be useful for the case studies we present later in \sectref{case-sect}.
For the third ($\future \phi$), \cite{CFK+13b} defines several variants,
which differ in the treatment of paths that never reach a state satisfying $\phi$.
We restrict our attention to the most commonly used one,
which is the default in PRISM, where paths that never satisfy $\phi$ have infinite reward. In the case of zero-sum formulae, adding the additional variants is straightforward based on the algorithm of~\cite{CFK+13b}. On the other hand, for nonzero-sum formulae, currently no algorithms exist for these variants.

As for other probabilistic temporal logics, it is useful to consider \emph{numerical} queries, which represent the value of an objective, rather than checking whether it is above or below some threshold. In the case of zero-sum formulae, these take the form $\coalition{C} \probop{\min=?}{\psi}$, $\coalition{C} \probop{\max=?}{\psi}$, $\coalition{C} \rewop{r}{\min=?}{\rho}$ and $\coalition{C} \rewop{r}{\max=?}{\rho}$.
For nonzero-sum formulae, numerical queries are of the form $\coalition{C{:}C'}_{\max=?}[\theta]$ and $\coalition{C{:}C'}_{\min=?}[\theta]$ which return the SWNE and SCNE values, respectively.

\begin{examp}\label{logic-eg}
Consider a scenario in which two robots ($\mathit{rbt}_1$ and $\mathit{rbt}_2$) move concurrently over a square grid of cells, where each is trying to reach their individual goal location. Each step of the robot involves transitioning to an adjacent cell, possibly stochastically. Examples of \emph{zero-sum} formulae, where $\mathsf{crash}, \mathsf{goal}_1, \mathsf{goal}_2$ denote the obvious atomic propositions labelling states, include:
\begin{itemize}
\item
$\coalition{\mathit{rbt}_1}\probop{\max=?}{\neg \mathsf{crash}\ \untilop^{\leq 10}\ \mathsf{goal}_1}$ asks what is the maximum probability with which the first robot can ensure that it reaches its goal location within 10 steps and without crashing, no matter how the second robot behaves;
\item
$\coalition{\mathit{rbt}_2}\rewop{r_\mathit{crash}}{\leq 1.5}{\future \mathsf{goal}_2}$ states that, no matter the behaviour of the first robot, the second robot can ensure the expected number of times it crashes before reaching its goal is less than or equal to $1.5$ ($r_\mathit{crash}$ is a reward structure that assigns 1 to states labelled $\mathsf{crash}$ and 0 to all other states).
\end{itemize}
Examples of \emph{nonzero-sum} formulae include:
\begin{itemize}
\item
$\nashop{\mathit{rbt}_1{:}\mathit{rbt}_2}{\max \geq 2}{\probop{}{\future
\mathsf{goal}_1}{+}\probop{}{\neg \mathit{crash} \ \untilop^{\leq 10} \mathsf{goal}_2}}$ states the robots can collaborate so that both reach their goal with probability 1, with the additional condition that the second has to reach its goal within 10 steps without crashing;
\item
$\nashop{\mathit{rbt}_1{:}\mathit{rbt}_2}{\min=?}{\rewop{r_\mathit{steps}}{}{\future \mathsf{goal}_1}{+}\rewop{r_\mathit{steps}}{}{\future \mathsf{goal}_2}}$ asks what is the sum of expected reachability values when the robots collaborate and each wants to minimise their expected steps to reach their goal ($r_\mathit{steps}$ is a reward structure that assigns 1 to all state and action tuple pairs).
\end{itemize}
Examples of more complex \emph{nested} formulae for this scenario include the following, where $r_\mathit{steps}$ is as above:
\begin{itemize}
\item
$\coalition{\mathit{rbt}_1}\probop{\max=?}{ \future \coalition{\mathit{rbt}_2}\rewop{r_\mathit{steps}}{\geq 10}{\futureop \, \mathsf{goal}_2}}$ asks what is the maximum probability with which the first robot can get to a state where the expected time for the second robot to reach their goal is at least 10 steps;
\item
$\coalition{\mathit{rbt}_1,\mathit{rbt}_2}\probop{\geq 0.75}{ \future \nashop{\mathit{rbt}_1{:}\mathit{rbt}_2}{\min \leq 5}{\rewop{r_\mathit{steps}}{}{\futureop \, \mathsf{goal}_1}{+}\rewop{r_\mathit{steps}}{}{\futureop \, \mathsf{goal}_2}}}$ states the robots can collaborate to reach, with probability at least 0.75, a state where the sum of the expected time for the robots to reach their goals is at most 5.
\end{itemize}
\end{examp}
\noindent
Before giving the semantics of the logic, we define \emph{coalition games} which, for a CSG $\game$ and coalition (set of players) $C\subseteq N$, reduce $\game$ to a two-player CSG $\game^C$,
with one player representing $C$ and the other $N{\setminus}C$.
Without loss of generality we assume the coalition of players is of the form $C = \{1,\dots,n'\}$.
\begin{definition}[Coalition game] For CSG $\game = (N, S, \bar{s}, A, \Delta, \delta, \AP, \lab)$ and coalition $C =  \{1,\dots,n'\} \subseteq N$, the \emph{coalition game} $\game^C  =  ( \{1,2\}, S, \bar{s}, A^C, \Delta^C, \delta^C, \AP, \lab )$ is a two-player CSG where:
\begin{itemize}
\item
$A^C = (A^C_1\cup \{ \bot\}) {\times} (A^C_2\cup \{ \bot\})$;
\item
$A^C_1 = (A_1\cup\{\bot\}) {\times} \cdots {\times} (A_{n'}\cup\{\bot\}) \setminus \{(\bot,\dots,\bot)\}$;
\item
$A^C_2 = (A_{n'+1}\cup\{\bot\}) {\times} \cdots {\times} (A_n\cup\{\bot\}) \setminus \{(\bot,\dots,\bot)\}$;
\item $a_1^C  =  (a_1,\dots,a_m) \in \Delta^C(s)$ if and only if either $\Delta(s) \cap A_j =\emptyset$ and $a_j=\bot$ or $a_j \in \Delta(s)$  for all $1  \leq  j  \leq  m$ and $a_2^C  =  (a_{m+1},\dots,a_n) \in \Delta^C(s)$ if and only if either $\Delta(s) \cap A_j =\emptyset$ and $a_j=\bot$ or $a_j \in \Delta(s)$  for all $m+1  \leq  j  \leq  n$ for $s \in S$;
\item
for any $s \in S$, $a^C_1 \in A^C_1$ and $a^C_2 \in A^C_2$ we have $\delta^C(s,(a^C_1,a^C_2))=\delta(s,(a_1,a_2))$ where $a_i=(\bot,\dots,\bot)$ if $a^C_i=\bot$ and $a_i=a^C_i$ otherwise for $1  \leq i \leq 2$.
\end{itemize}
Furthermore, for a reward structure $r=(r_A,r_S)$ of $\game$, by abuse of notation we also use $r$ for the corresponding reward structure $r=(r^C_A,r^C_S)$ of $\game^C$ where:
\begin{itemize}
\item
for any $s \in S$, $a^C_1 \in A^C_1$ and $a^C_2 \in A^C_2$ we have $r^C_A(s,(a^C_1,a^C_2))=r_A(s,(a_1,a_2))$ where $a_i=(\bot,\dots,\bot)$ if $a^C_i=\bot$ and $a_i=a^C_i$ otherwise for $1  \leq i \leq 2$;
\item
for any $s \in S$ we have $r^C_S(s) =r_S(s)$.
\end{itemize}
\end{definition}
Our logic includes both \emph{finite-horizon} ($\next$, $\buntilop$, $\sinstant{=k}$, $\scumul{\leq k}$) and \emph{infinite-horizon} ($\untilop$, $\futureop$) temporal operators.
For the latter, the existence of \swne or \scne profiles is an open problem~\cite{BMS14},
but we can check for \eswne or \escne profiles for any $\varepsilon$.
Hence, we define the semantics of the logic in the context of a particular $\varepsilon$.

\begin{definition}[Extended rPATL semantics]\label{sem-def}
For a CSG $\game$, $\varepsilon>0$ and a formula $\phi$ in our rPATL extension,
we define the satisfaction relation $\sateps$ inductively over the structure of $\phi$.
The propositional logic fragment $(\mathtt{true}$, $\ap$, $\neg$, $\wedge)$
is defined in the usual way.
For a \emph{zero-sum} formula and state $s \in S$ of CSG $\game$, we have:
\begin{eqnarray*}
s \sateps \coalition{C} \probop{\sim q}{\psi}  
& \ \Leftrightarrow \ &
\exists \sigma_1 \in \Sigma^1 . \, \forall \sigma_2 \in \Sigma^2 . \,
\Eset^{\sigma_1,\sigma_2}_{\game^{C},s}(X^\psi) \sim q \\
s \sateps \coalition{C} \rewop{r}{\sim x}{\rho}
& \ \Leftrightarrow \ &
\exists \sigma_1 \in \Sigma^1 . \, \forall \sigma_2 \in \Sigma^2  . \, \Eset^{\sigma_1,\sigma_2}_{\game^C,s}(X^{r,\rho}) \sim x
\end{eqnarray*}
For a \emph{nonzero-sum} formula and state $s \in S$ of CSG $\game$, we have:
\begin{eqnarray*}
s \sateps \nashop{C{:}C'}{\opt \sim x}{\theta}
& \ \Leftrightarrow \ &
\exists \sigma_1^\star \in \Sigma^1, \sigma_2^\star \in \Sigma^2 . \, \left( \Eset^{\sigma_1^\star,\sigma_2^\star}_{\game^{C},s}(X^\theta_1)+\Eset^{\sigma_1^\star,\sigma_2^\star}_{\game^{C},s}(X^\theta_2) \right) \sim x %
\end{eqnarray*}
where $(\sigma_1^\star,\sigma_2^\star)$ is a subgame-perfect \eswne profile if $\opt = \max$, or a subgame-perfect \escne profile if $\opt = \min$, for the objectives $(X^\theta_1,X^\theta_2)$ in $\game^{C}$.
For an objective $X^{\psi}$, $X^{r,\rho}$
or $X^\theta_i$ ($1 \leq i \leq 2$),
and path $\pi \in \ipaths_{\game^C,s}$:
\begin{eqnarray*}
X^{\psi}(\pi) & \ = \ & 1 \;\mbox{if $\pi \sateps \psi$ and 0 otherwise} \\
X^{r,\rho}(\pi) & \ = \ & \rew(r,\rho)(\pi) \\
X^{\probop{}{\psi^1}{+}\probop{}{\psi^2}}_i(\pi) & \ = \ & 1 \;\mbox{if $\pi \sateps \psi^i$ and 0 otherwise} \\
X^{\rewop{r_{\scale{.75}{1}}}{}{\rho^1}{+}\rewop{r_{\scale{.75}{2}}}{}{\rho^2}}_i(\pi) & \ = \ &  \rew(r_i,\rho^i)(\pi)
\end{eqnarray*}
For a temporal formula and path $\pi \in \ipaths_{\game^C,s}$:
\begin{eqnarray*}
\pi \sateps \next \phi 
& \ \Leftrightarrow \ & 
\pi(1) \sateps \phi \\
\pi \sateps \phi_1 \buntil \phi_2 
& \ \Leftrightarrow \ &
\exists i \leq k . \, (\pi(i) \sateps \phi_2 \wedge \forall j < i . \, \pi(j) \sateps \phi_1 )
\\
\pi \sateps \phi_1 \until \phi_2 
& \ \Leftrightarrow \ & 
\exists i \in \Nset . \, ( \pi(i) \sateps \phi_2 \wedge \forall j < i  . \, \pi(j) \sateps \phi_1 )
\end{eqnarray*}
For a reward structure $r$, reward formula and path $\pi \in \ipaths_{\game^C,s}$:
\begin{eqnarray*}
\rew(r,\sinstant{=k})(\pi) & \ = \ & r_S(\pi(k)) \\
\rew(r,\scumul{\leq k})(\pi) & \ = \ & \mbox{$\sum_{i=0}^{k-1}$} \big( r_A(\pi(i),\pi[i])+r_S(\pi(i)) \big) \\
\rew(r,\future  \phi)(\pi) & \ = \ & \begin{cases}
\infty
& \mbox{if} \; \forall j \in \Nset . \, \pi(j) \notsateps \phi \\
\mbox{$\sum_{i=0}^{k_\phi-1}$} \big( r_A(\pi(i),\pi[i])+r_S(\pi(i)) \big) & \mbox{otherwise}
\end{cases}
\end{eqnarray*}
where $k_\phi = \min \{ k \mid \pi(k) \sateps \phi \}$.
\end{definition}
Using the notation above, we can also define the numerical queries mentioned previously.
For example, for state $s$ we have:
\[
\begin{array}{rcl}
\coalition{C} \probop{\min=?}{\psi}
& \ \rmdef & \ 
\inf_{\sigma_1 \in \Sigma^1_{\game^{\scale{.75}{C}}}} \sup_{\sigma_2 \in \Sigma^2_{\game^{\scale{.75}{C}}}}
\Eset^{\sigma_1,\sigma_2}_{\game^{C},s}(X^\psi) \\
\coalition{C} \probop{\max=?}{\psi}
& \ \rmdef & \ 
\sup_{\sigma_1 \in \Sigma^1_{\game^{\scale{.75}{C}}}} \inf_{\sigma_2 \in \Sigma^2_{\game^{\scale{.75}{C}}}}
\Eset^{\sigma_1,\sigma_2}_{\game^{C},s}(X^\psi) \, .
\end{array}
\]
As the zero-sum objectives appearing in the logic are either finite-horizon or infinite-horizon and correspond to either probabilistic until or expected reachability formulae, we have that CSGs are \emph{determined} (see \defref{determined-def}) with respect to these objectives~\cite{Mar98}, and therefore values exist. More precisely, for any CSG $\game$, coalition $C$, state $s$, path formula $\psi$, reward structure $r$ and reward formula $\rho$, the values $\val_{\game^C}(s,X^\psi)$ and $\val_{\game^C}(s,X^{r,\rho})$ of the game $\game^C$ in state $s$ with respect to the objectives $X^\psi$ and $X^{r,\rho}$ are well defined. 
This determinacy result also yields the following equivalences:
\[
\coalition{C} \probop{\max=?}{\psi} \ \equiv \ \coalition{N {\setminus} C} \probop{\min=?}{\psi} \;\;\; \mbox{and} \;\;\;
\coalition{C} \rewop{r}{\max=?}{\rho} \ \equiv \ \coalition{N {\setminus} C} \rewop{r}{\min=?}{\rho} \, .
\]
Also, as for other probabilistic temporal logics, we can represent negated path formulae by inverting the probability threshold, e.g.:
$\coalition{C} \probop{\geq q}{\neg\psi} \equiv  \coalition{C} \probop{\leq 1-q}{\psi}$ and $\nashop{C{:}C'}{\max\geq q }{\probop{}{\psi_1}{+}\probop{}{\psi_2}} \equiv \nashop{C{:}C'}{\min\leq 2-q }{\probop{}{\neg\psi_1}{+}\probop{}{\neg\psi_2}}$,
notably allowing the `globally' operator $\globally \phi \equiv \neg (\future \neg \phi)$ to be defined.
\section{Model Checking for Extended rPATL against CSGs}\label{mc-sect}

We now present model checking algorithms for the extended rPATL logic, introduced in the previous section, on a CSG $\game$.
Since rPATL is a branching-time logic, this works by recursively
computing the set $\Sat(\phi)$ of states satisfying formula $\phi$
over the structure of $\phi$, as is done for rPATL on TSGs~\cite{CFK+13b}.

If $\phi$ is a zero-sum formula of the form $\coalition{C} \probop{\sim q}{\psi}$ or $\coalition{C} \rewop{r}{\sim x}{\rho}$, this reduces to computing values for a two-player CSG (either $\game^C$ or $\game^{N \setminus C}$) with respect to $X^\psi$ or $X^{r,\rho}$. In particular, for $\sim \, \in \{ \geq , > \}$ and $s \in S$ we have:
\begin{eqnarray*}
s \sat \coalition{C} \probop{\sim q}{\psi} 
& \ \Leftrightarrow \ &
\val_{\game^C}(s,X^\psi) \sim q \\
s \sat \coalition{C} \rewop{r}{\sim x}{\rho} 
& \ \Leftrightarrow \ &
\val_{\game^C}(s,X^{r,\rho}) \sim x \, .
\end{eqnarray*}
and, since CSGs are determined for the zero-sum properties we consider, for $\sim \, \in \{ < , \leq \}$ we have:
\begin{eqnarray*}
s \sat \coalition{C} \probop{\sim q}{\psi}
& \ \Leftrightarrow \ &
\val_{\game^{N\setminus C}}(s,X^\psi) \sim q \\
s \sat \coalition{C} \rewop{r}{\sim x}{\rho}
& \ \Leftrightarrow \ &
\val_{\game^{N \setminus C}}(s,X^{r,\rho}) \sim x \, .
\end{eqnarray*}
Without loss of generality, for such formulae we focus on computing
$\val_{\game^C}(s,X^\psi)$ and $\val_{\game^C}(s,X^{r,\rho})$ and, to simplify the presentation, we denote these values by $\V_{\game^C}(s,\psi)$ and $\V_{\game^C}(s,r,\rho)$ respectively.

If, on the other hand, $\phi$ is a nonzero-sum formula of the form $\nashop{C{:}C'}{\opt \sim x}{\theta}$ then, from the semantics for $\nashop{C{:}C'}{\opt \sim x}{\theta}$ (see \defref{sem-def}),
computing $\Sat(\phi)$ reduces to the computation of subgame-perfect SWNE or SCNE values for the objectives $(X^\theta_1,X^\theta_2)$ and a comparison of their sum to the threshold $x$. Again, to simplify the presentation, will use the notation $\V_{\game^C}(s,\theta)$ for the SWNE values of the objectives $(X^\theta_1,X^\theta_2)$ in state $s$ of $\game^C$.

For the remainder of this section, we fix a CSG $\game= (N, S, \bar{S}, A, \Delta, \delta, \AP, \lab)$ and coalition $C$ of players and assume that the available actions of players 1 and 2 of the (two-player) CSG $\game^C$ in a state $s$ are $\{a_1,\dots,a_l\}$ and $\{b_1,\dots,b_m\}$, respectively.
We also fix a value $\varepsilon>0$ which, as discussed in \sectref{logic-sect},
is needed to define the semantics of our logic,
in particular for infinite-horizon objectives
where we need to consider $\varepsilon$-SWNE profiles.

\startpara{Assumptions}
Our model checking algorithms require several assumptions on CSGs,
depending on the operators that appear in the formula $\phi$.
These can all be checked using standard graph algorithms~\cite{dAH00}.
In the diverse set of model checking case studies
that we later present in \sectref{case-sect},
these assumptions have not limited the practical applicability of our model checking algorithms.

For zero-sum formulae, the only restriction is for infinite-horizon reward properties
on CSGs with both positive and negative reward values.
\begin{assumption}\label{game1-assum}
For a zero-sum formula of the form $\coalition{C}\rewop{r}{\sim x}{\future \phi}$,
from any state $s$ where $r_S(s)<0$ or $r_A(s,a)<0$ for some action $a$,
under all profiles of $\game$, with probability~1
we reach either a state satisfying $\phi$ or a state where all rewards are zero and which cannot be left with probability 1 under all profiles.
\end{assumption}
Without this assumption, the values computed during value iteration can oscillate, and therefore fail to converge (see~\appref{2-app}).
This restriction is not applied in the existing rPATL model checking algorithms for TSGs~\cite{CFK+13b} since that work assumes that all rewards are non-negative.

The remaining two assumptions concern nonzero-sum formulae
that contain infinite-horizon objectives.
We restrict our attention to a class of CSGs
that can be seen as a variant of \emph{stopping games}~\cite{CFKSW13},
as used for multi-objective TSGs.
Compared to~\cite{CFKSW13}, we use a weaker, objective-dependent assumption,
which ensures that, under all profiles,
with probability 1, eventually the outcome of each player's objective does not change by continuing.
\begin{assumption}\label{game2-assum}
For nonzero-sum formulae, if $\probop{}{\phi_1 \until \phi_2}$ is a probabilistic objective, then $\Sat(\neg \phi_1 \vee \phi_2)$ is reached with probability 1 from all states under all profiles of $\game$.
\end{assumption}
\begin{assumption}\label{game3-assum}
For nonzero-sum formulae, if $\rewop{r}{}{\future \phi}$ is a reward objective, then $\Sat(\phi)$ is reached with probability 1 from all states under all profiles of $\game$.
\end{assumption}
Like for Assumption~\ref{game1-assum}, without this restriction,
value iteration may not converge since values can oscillate
(see \appappref{3-app}{4-app}).
Notice that Assumption~\ref{game1-assum} is not required for nonzero-sum
properties containing negative rewards since Assumption~\ref{game3-assum} is itself a stronger restriction.

\subsection{Model Checking Zero-Sum Properties}\label{zero-sect}

In this section, we present algorithms for \emph{zero-sum} properties,
i.e., for computing the values $\V_{\game^C}(s,\psi)$ or $\V_{\game^C}(s,r,\rho)$
for path formulae $\psi$ or reward formulae $\rho$ in all states $s$ of $\game^C$.
We split the presentation into
\emph{finite-horizon properties}, which can be solved exactly using backward induction~\cite{SW+01,NMK+44},
and \emph{infinite-horizon properties}, for which we approximate values using value iteration~\cite{RF91,CH08}.
Both cases require the solution of matrix games, for which we rely on
the linear programming approach presented in \sectref{matrix-sect}.

\subsubsection{Computing the Values of Zero-Sum Finite-Horizon Formulae}

Finite-horizon properties are defined over a bounded number of steps:
the next or bounded until operators for probabilistic formulae,
and the instantaneous or bounded cumulative reward operators.
Computation of the values $\V_{\game^C}(s,\psi)$ or $\V_{\game^C}(s,r,\rho)$
for these is done recursively, based on the step bound, using backward induction and
solving matrix games in each state at each iteration.
The actions of each matrix game correspond to the actions available in that state;
the utilities are constructed from the transition probabilities $\delta^C$ of the game $\game^C$, the reward structure $r$ (in the case of reward formulae) and the values already computed recursively for successor states.

\startpara{Next} This is the simplest operator, over just one step,
and so in fact requires no recursion, just solution of a matrix game for each state.
If $\psi = \next \phi$, then for any state $s$ we have that $\V_{\game^C}(s,\psi) = \val(\mgame)$ where $\mgame \in \Qset^{l \times m}$ is the matrix game with:
\[ \begin{array}{c}
z_{i,j} = \sum_{s' \in \Sat(\phi)} \delta^C(s,(a_i,b_j))(s') \, .
\end{array} \]
\startpara{Bounded Until} If $\psi = \phi_1 \buntil \phi_2$, we compute the values for the path formulae $\psi_{n} = \phi_1 \ \untilop^{\leq n}\ \phi_2$ for $0 \leq n \leq k$ recursively. For any state $s$:
\[
\V_{\game^C}(s,\psi_n) = \begin{cases}
1 & \mbox{if $s \in \Sat(\phi_2)$} \\
0 & \mbox{else if $s \not\in \Sat(\phi_1)$} \\
0 & \mbox{else if $n = 0$} \\
\val(\mgame) & \mbox{otherwise}
\end{cases}
\]
where $\val(\mgame)$ equals the value of the matrix game $\mgame \in \Qset^{l \times m}$ with:
\begin{align*}
z_{i,j} & = \; \mbox{$\sum_{s' \in S}$} \, \delta^C(s,(a_i,b_j))(s') \cdot v^{s'}_{n-1}
\end{align*}
and $v^{s'}_{n-1} = \V_{\game^C}(s',\psi_{n-1})$ for all $s' \in S$.

\startpara{Instantaneous Rewards} If $\rho = \sinstant{=k}$, then for the reward structure $r$ we compute the values for the reward formulae $\rho_{n} = \sinstant{=n}$ for $0 \leq n \leq k$ recursively. For any state $s$:
\[
\V_{\game^C}(s,r,\rho_n) = \begin{cases}
r_S(s) & \mbox{if $n=0$} \\
\val(\mgame) & \mbox{otherwise}
\end{cases} 
\]
where $\val(\mgame)$ equals the value of the matrix game $\mgame \in \Qset^{l \times m}$ with:
\begin{align*}
z_{i,j} = \mbox{$\sum_{s' \in S}$} \, \delta^C(s,(a_i,b_j))(s') \cdot v^{s'}_{n-1}
\end{align*}
and $v^{s'}_{n-1} = \V_{\game^C}(s',r,\rho_{n-1})$ for all $s' \in S$.

\startpara{Bounded Cumulative Rewards} If $\rho = \scumul{\leq k}$, then for the reward structure $r$ we compute the values for the reward formulae $\rho_{n} = \scumul{\leq n}$ for $0 \leq n \leq k$ recursively. For any state $s$:
\[
\V_{\game^C}(s,r,\rho_n) = \begin{cases}
0 & \mbox{if $n=0$} \\
\val(\mgame) & \mbox{otherwise}
\end{cases}
\]
where $\val(\mgame)$ equals the value of the matrix game $\mgame\in \Qset^{l \times m}$ with:
\begin{align*}
z_{i,j} = r_A(s,(a_i,b_j)) + r_S(s) + \mbox{$\sum_{s' \in S}$} \, \delta^C(s,(a_i,b_j))(s') \cdot v^{s'}_{n-1}
\end{align*}
and $v^{s'}_{n-1} = \V_{\game^C}(s',r,\rho_{n-1})$ for all $s' \in S$.

\subsubsection{Computing the Values of Zero-Sum Infinite Horizon Formulae}\label{inf-zero-sect}

We now discuss how to compute the values $\V_{\game^C}(s,\psi)$ and $\V_{\game^C}(s,r,\rho)$ for infinite-horizon properties,
i.e., when the path formula $\psi$ is an until operator,
or for the expected reachability variant of the reward formulae $\rho$.
In both cases, we approximate these values using value iteration,
adopting a similar recursive computation to the finite-horizon cases above,
solving matrix games in each state and at each iteration,
which converges in the limit to the desired values.

Following the approach typically taken in probabilistic model checking tools
to implement value iteration, we estimate convergence of the iterative computation
by checking the maximum relative difference between successive iterations.
However, it is known~\cite{HM18} that, even for simpler probabilistic models
such as MDPs, this convergence criterion cannot be used to guarantee that
the final computed values are accurate to within a specified error bound.
Alternative approaches that resolve this by computing lower and upper bounds
for each state have been proposed for MDPs (e.g.~\cite{HM18,BCC+14})
and extended to both single- and multi-objective solution of TSGs~\cite{KKKW18,ACK+20};
extensions could be investigated for CSGs.
Another possibility is to use \emph{policy iteration} (see, e.g., \cite{CAH13}).

\startpara{Until} If $\psi = \phi_1 \until \phi_2$, the probability values can be approximated through value iteration using the fact that $\langle \V_{\game^C}(s,\phi_1 \buntil \phi_2) \rangle_{k \in \Nset}$ is a non-decreasing sequence converging to $\V_{\game^C}(s,\phi_1 \until \phi_2)$.
We compute $\V_{\game^C}(s,\phi_1 \buntil \phi_2)$ for increasingly large $k$
and estimate convergence as described above,
based on the difference between values in successive iterations.
However, we can potentially speed up convergence by first precomputing the set of states $S^{\psi}_0$ for which the value of the zero-sum objective $X^\psi$ is 0 and the set of states $S^{\psi}_1$ for which the value is 1 using standard graph algorithms~\cite{dAH00}. We can then apply value iteration to approximate $\V_{\game^C}(s,\phi_1 \until \phi_2) = \lim_{k \ra \infty} \V_{\game^C}(s,\phi_1 \until \phi_2,k)$ where:
\[
\V_{\game^C}(s,\phi_1 \until \phi_2,n) = \begin{cases}
1 & \mbox{if $s \in S^{\psi}_1$} \\
0 & \mbox{else if $s \in S^{\psi}_0$} \\
0 & \mbox{else if $n = 0$} \\
\val(\mgame) & \mbox{otherwise}
\end{cases}
\]
where $\val(\mgame)$ equals the value of the matrix game $\mgame \in \Qset^{l \times m}$ with:
\begin{align*}
z_{i,j} & = \; \mbox{$\sum_{s' \in S}$} \, \delta^C(s,(a_i,b_j))(s') \cdot v^{s'}_{n-1}
\end{align*}
and $v^{s'}_{n-1} = \V_{\game^C}(s',\phi_1 \until \phi_2,n-1)$ for all $s' \in S$.
\startpara{Expected Reachability} If $\rho = \future \phi$ and the reward structure is $r$, then we first make all states of $\game^C$ satisfying $\phi$ absorbing, {i.e., we remove all outgoing transitions from such states}.
Second, we find the set of states $S^\rho_\infty$ for which the reward is infinite; as in \cite{CFK+13b}, this involves finding the set of states satisfying the formula $\coalition{C}\probop{<1}{\future \phi}$ and we can use the graph algorithms of \cite{dAH00} to find these states. Again following \cite{CFK+13b}, to deal with zero-reward cycles we need to use value iteration to compute a greatest fixed point.
This involves first computing upper bounds on the actual values, by changing all zero reward values to some value $\gamma>0$ to construct the reward structure $r_\gamma=(r_A^\gamma,r_A^\gamma)$ and then applying value iteration to approximate $\V_{\game^C}(s,r_\gamma,\rho) = \lim_{k \ra \infty} \V_{\game^C}(s,r_\gamma,\rho_k)$ where:
\[
\V_{\game^C}(s,r_\gamma,\rho_n) = \begin{cases}
0 & \mbox{if $s \in Sat(\phi)$} \\
\infty & \mbox{if $s \in S^\rho_\infty$} \\
\val(\mgame) & \mbox{otherwise}
\end{cases}
\]
where $\val(\mgame)$ equals the value of the matrix game $\mgame \in \Qset^{l \times m}$ with:
\begin{align*}
z_{i,j} = r_A^\gamma(s,(a_i,b_j)) + r_S^\gamma(s) + \mbox{$\sum_{s' \in S}$} \, \delta^C(s,(a_i,b_j))(s') \cdot v^{s'}_{n-1}
\end{align*}
and $v^{s'}_{n-1} = \V_{\game^C}(s',r_\gamma,\rho_{n{-}1})$ for all $s' \in S$.
Finally, using these upper bounds as the initial values we again perform value iteration as above, except now using the original reward structure $r$,
i.e., to approximate $\V_{\game^C}(s,r_,\rho) = \lim_{k \ra \infty} \V_{\game^C}(s,r,\rho_k)$.
The choice of $\gamma$ can influence value iteration computations in opposing ways: increasing $\gamma$ can speed up convergence when computing over-approximations, while potentially slowing it down when computing the actual values.

\subsection{Model Checking Nonzero-Sum Properties}\label{nonzero-sect}

Next, we show how to compute subgame-perfect SWNE and SCNE values
for the two objectives corresponding to a \emph{nonzero-sum} formula.
As for the zero-sum case, the approach taken depends on whether the formula
contains finite-horizon or infinite-horizon objectives.
We now have three cases:
\begin{enumerate}
\item when both objectives are finite-horizon, we use backward induction~\cite{SW+01,NMK+44} to compute (precise) subgame-perfect SWNE and SCNE values;
\item
when both objectives are infinite-horizon, we use value iteration~\cite{RF91,CH08}
to approximate the values;
\item
when there is a mix of the two types of objectives,
we convert the problem to two infinite-horizon objectives on an augmented model.
\end{enumerate}
We describe these three cases separately in Sections~\ref{nonzero-bounded-sect}, \ref{nonzero-unbounded-sect} and \ref{nonzero-mixed-sect},  respectively,
focusing on the computation of SWNE values.
Then, in \sectref{scne-sect}, we explain how to adapt this for SCNE values.
 
In a similar style to the algorithms for zero-sum properties,
in all three cases the computation is an iterative process
that analyses a two-player game for each state at each step.
However, this now requires finding SWNE or SCNE values of a bimatrix game,
rather than solving a matrix game as in the zero-sum case.
We solve bimatrix games using the approach presented in \sectref{bimatrix-sect}
(see also the more detailed discussion of its implementation in \sectref{impl-sect}).

Another important aspect of our algorithms is that, for efficiency, 
if we reach a state where the value of one player's objective cannot change
(e.g., the goal of that player is reached or can no longer be reached),
then we switch to the simpler problem of solving an MDP
to find the optimal value for the other player in that state.
This is possible since the only SWNE profile in that state corresponds to
maximising the objective of the other player. More precisely:
\begin{itemize}
\item
the first player (whose objective cannot change) is \emph{indifferent},
since its value will not be affected by the choices of either player;
\item
the second player cannot do better than the optimal value of its objective in the corresponding MDP where both players collaborate;
\item
for any NE profile, the value of the first player is fixed and the value of the second is less than or equal to the optimal value of its objective in the MDP.
\end{itemize}
We use the notation $\probopP^{\max}_{\game,s}(\psi)$ and $\rewopR^{\max}_{\game,s}(r,\rho)$ for the maximum probability of satisfying the path formula $\psi$ and the maximum expected reward for the random variable $\rew(r,\rho)$, respectively, when the players collaborate in state $s$.
These values can be computed through standard MDP model checking~\cite{BdA95,dA99}.

\subsubsection{Computing SWNE Values of Finite-Horizon Nonzero-Sum Formulae}\label{nonzero-bounded-sect}

As for the zero-sum case, for a \emph{finite-horizon} nonzero-sum formula $\theta$,
we compute the SWNE values $\V_{\game^C}(s,\theta)$ for all states $s$ of $\game^C$
in a recursive fashion based on the step bound.
We now solve bimatrix games at each step, which are defined in a 
similar manner to the matrix games for zero-sum properties:
the actions of each bimatrix game correspond to the actions available in that state and
the utilities are constructed from the transition probabilities $\delta^C$ of the game $\game^C$, the reward structure (in the case of reward formulae) and the values already computed recursively for successor states.

For any state formula $\phi$ and state $s$ we let $\eta_{\phi}(s)$ equal $1$ if $s \in \Sat(\phi)$ and $0$ otherwise.
Recall that probability and reward values of the form $\probopP^{\max}_{\game,s}(\psi)$ and $\rewopR^{\max}_{\game,s}(r,\rho)$, respectively, are computed through standard MDP verification.
Below, we explain the computation for both types of finite-horizon probabilistic objectives (next and bounded until) and reward objectives (instantaneous and bounded cumulative), as well as combinations of each type.

\startpara{Next} If $\theta = \probop{}{\next \phi^1}{+}\probop{}{\next \phi^2}$, then 
$\V_{\game^C}(s,\theta)$ equals SWNE values of the bimatrix game $(\mgame_1, \mgame_2) \in \Qset^{l \times m}$ where: 
\begin{align*}
z^1_{i,j} & = \; \mbox{$\sum_{s' \in \Sat(\phi^1)}$} \, \delta^C(s,(a_i,b_j))(s') \\
z^2_{i,j} & = \; \mbox{$\sum_{s' \in \Sat(\phi^2)}$} \, \delta^C(s,(a_i,b_j))(s') \, .
\end{align*}
Again, since next is a 1-step property, no recursion is required.
\startpara{Bounded Until} If $\theta = \probop{}{\phi_1^1 \buntilp{k_1} \phi_2^1}\,+\,\probop{}{\phi_1^2 \buntilp{k_2} \phi_2^2}$, we compute SWNE values for the objectives for the nonzero-sum formulae $\theta_{n+n_1,n+n_2}=\probop{}{\phi_1^1 \buntilp{n+n_1} \phi_2^1}\,+\,\probop{}{\phi_1^2 \buntilp{n+n_2} \phi_2^2}$ for $0 \leq n \leq k$ recursively, where $k = \min\{k_1,k_2\}$, $n_1 = k_1{-}k$ and $n_2 = k_2{-}k$. In this case, there are three situations in which the value of the objective of one of the players cannot change, and hence we can switch to MDP verification. The first is when the step bound is zero for only one of the corresponding objectives, the second is when a state satisfying $\phi_2^i$ is reached by only one player $i$ (and therefore the objective is satisfied by that state) and the third is when a state satisfying $\neg \phi_1^i \wedge \neg \phi_2^i$ is reached by only one player $i$ (and therefore the objective is not satisfied by that state). For any state $s$, if $n = 0$, then:
\[
\V_{\game^C}(s,\theta_{n_1,n_2}) = \begin{cases}
(\eta_{\phi^1_2}(s),\eta_{\phi^2_2}(s)) & \mbox{if $n_1 = n_2 = 0$} \\
(\eta_{\phi^1_2}(s),\probopP^{\max}_{\game,s}(\phi_1^2 \buntilp{n_2} \phi_2^2)) & \mbox{else if $n_1 = 0$} \\
(\probopP^{\max}_{\game,s}(\phi_1^1 \buntilp{n_1} \phi_2^1),\eta_{\phi^2_2}(s)) & \mbox{otherwise.}
\end{cases}
\]
On the other hand, if $n>0$, then:
\[
\V_{\game^C}(s,\theta_{n+n_1,n+n_2}) = 
\begin{cases}
(1,1) & \mbox{if $s \in \Sat(\phi_2^1) \cap \Sat(\phi_2^2)$} \\
(1,\probopP^{\max}_{\game,s}(\phi_1^2 \buntilp{n+n_2} \phi_2^2)) & \mbox{else if $s \in \Sat(\phi_2^1)$} \\
(\probopP^{\max}_{\game,s}(\phi_1^1 \buntilp{n+n_1} \phi_2^1),1) & \mbox{else if $s \in \Sat(\phi_2^2)$} \\
(\probopP^{\max}_{\game,s}(\phi_1^1 \buntilp{n+n_1} \phi_2^1),0) & \mbox{else if $s \in \Sat(\phi_1^1) \setminus \Sat(\phi_1^2)$} \\
(0,\probopP^{\max}_{\game,s}(\phi_1^2 \buntilp{n+n_2} \phi_2^2)) & \mbox{else if $s \in \Sat(\phi_1^2) \setminus \Sat(\phi_1^1)$} \\
(0,0) & \mbox{else if $s \not\in \Sat(\phi_1^1) \cap \Sat(\phi_1^2)$} \\
\val(\mgame_1, \mgame_2) & \mbox{otherwise}
\end{cases}
\]
where $\val(\mgame_1, \mgame_2)$ equals SWNE values of the bimatrix game $(\mgame_1,\mgame_2)\in \Qset^{l \times m}$:
\begin{align*}
z^1_{i,j} & = \; \mbox{$\sum_{s' \in S}$} \, \delta^C(s,(a_i,b_j))(s') \cdot v^{s',1}_{(n-1)+n_1} \\ 
z^2_{i,j} & = \; \mbox{$\sum_{s' \in S}$} \, \delta^C(s,(a_i,b_j))(s') \cdot v^{s',2}_{(n-1)+n_2}
\end{align*}
and $(v^{s',1}_{(n-1)+n_1},v^{s',2}_{(n-1)+n_2}) = \V_{\game^C}(s',\theta_{(n-1)+n_1,(n-1)+n_2})$ for all $s' \in S$.
\startpara{Next and Bounded Until} If $\theta = \probop{}{\next \phi^1}{+}\probop{}{\phi_1^2 \buntilp{k_2} \phi_2^2}$, then 
$\V_{\game^C}(s,\theta)$ equals SWNE values of the bimatrix game $(\mgame_1, \mgame_2) \in \Qset^{l \times m}$ where: 
\begin{align*}
z^1_{i,j} & = \; \mbox{$\sum_{s' \in S}$} \, \delta^C(s,(a_i,b_j))(s') \cdot \eta_{\phi^1}(s') \\
z^2_{i,j} & = \; \begin{cases}
1 & \mbox{if $s \in \Sat(\phi_2^2)$} \\
0 & \mbox{else if $k_2 = 0$} \\
\mbox{$\sum_{s' \in S}$} \, \delta^C(s,(a_i,b_j))(s') \cdot \probopP^{\max}_{\game,s}(\phi_1^2 \buntilp{k_2-1} \phi_2^2) & \mbox{else if $\Sat(\phi_1^2)$} \\
0 & \mbox{otherwise.}
\end{cases}
\end{align*}
In this case, since the value for objectives corresponding to next formulae cannot change after the first step, we can always switch to MDP verification after this step. The symmetric case is similar.
\startpara{Instantaneous Rewards} If $\theta = \rewop{r_1}{}{\sinstant{=k_1}}{+}\rewop{r_2}{}{\sinstant{=k_2}}$, we compute SWNE values of the objectives for the nonzero-sum formulae $\theta_{n+n_1,n+n_2}=\rewop{r_1}{}{\sinstant{=n+n_1}}+\rewop{r_2}{}{\sinstant{=n+n_2}}$ for $0 \leq n \leq k$ recursively, where $k = \min\{k_1,k_2\}$, $n_1 = k_1{-}k$ and $n_2 = k_2{-}k$. Here, there is only one situation in which the value of the objective of one of the players cannot change: when one of the step bounds equals zero. Hence, this is the only time we switch to MDP verification. For any state $s$, if $n = 0$, then:
\[
\V_{\game^C}(s,\theta_{n_1,n_2}) = 
\begin{cases}
(r^1_S(s),r^2_S(s)) & \mbox{if $n_1 = n_2 = 0$} \\
(r^1_S(s),\rewopR^{\max}_{\game,s}(r_2,\sinstant{=n_2})) & \mbox{else if $n_1 = 0$} \\
(\rewopR^{\max}_{\game,s}(r_1,\sinstant{= n_1}),r^2_S(s)) & \mbox{otherwise.}
\end{cases}
\]
On the other hand, if $n>0$, then
$\V_{\game^C}(s,\theta_{n+n_1,n+n_2})$ equals SWNE values of the bimatrix game $(\mgame_1, \mgame_2) \in \Qset^{l \times m}$ where: 
\begin{align*}
z^1_{i,j} & = \; \mbox{$\sum_{s' \in S}$} \, \delta^C(s,(a_i,b_j))(s') \cdot v^{s',1}_{(n-1)+n_1} \\
z^2_{i,j} & = \; \mbox{$\sum_{s' \in S}$} \, \delta^C(s,(a_i,b_j))(s') \cdot v^{s',2}_{(n-1)+n_2}
\end{align*}
and $(v^{s',1}_{(n-1)+n_1},v^{s',2}_{(n-1)+n_2}) = \V_{\game^C}(s',\theta_{(n-1)+n_1,(n-1)+n_2})$ for all $s' \in S$.
\startpara{Bounded Cumulative Rewards} If $\theta = \rewop{r_1}{}{\scumul{\leq k_1}}{+}\rewop{r_2}{}{\scumul{\leq k_2}}$, we compute values of the objectives for the formulae $\theta_{n+n_1,n+n_2}=\rewop{r_1}{}{\scumul{\leq n+n_1}}+\rewop{r_2}{}{\scumul{\leq n+n_2}}$ for $0 \leq n \leq k$ recursively, where $k = \min\{k_1,k_2\}$, $n_1 = k_1{-}k$ and $n_2 = k_2{-}k$. As for instantaneous rewards, the only time we can switch to MDP verification is when one of the step bounds equals zero.
For state $s$, if $n = 0$:
\[
\V_{\game^C}(s,\theta_{n_1,n_2}) = 
\begin{cases}
(0,0) & \mbox{if $n_1 = n_2 = 0$} \\
(0,\rewopR^{\max}_{\game,s}(r_2,\scumul{\leq n_2})) & \mbox{else if $n_1 = 0$} \\
(\rewopR^{\max}_{\game,s}(r_1,\scumul{\leq n_1}),0) & \mbox{otherwise}
\end{cases}
\]
and if $n>0$, then $\V_{\game^C}(s,\theta_{n+n_1,n+n_2})$ equals SWNE values of the bimatrix game $(\mgame_1, \mgame_2) \in \Qset^{l \times m}$: 
\begin{align*}
z^1_{i,j} &= \; r^1_S(s) + r^1_A(s,(a_i,b_j)) + \mbox{$\sum_{s' \in S}$} \, \delta^C(s,(a_i,b_j))(s') \cdot v^{s',1}_{(n-1)+n_1} \\
z^2_{i,j} &= \; r^2_S(s) + r^2_A(s,(a_i,b_j)) + \mbox{$\sum_{s' \in S}$} \, \delta^C(s,(a_i,b_j))(s') \cdot v^{s',l}_{(n-1)+n_2}
\end{align*}
and $(v^{s',1}_{(n-1)+n_1},v^{s',2}_{(n-1)+n_2}) = \V_{\game^C}(s',\theta_{(n-1)+n_1,(n-1)+n_2})$ for all $s' \in S$.
\startpara{Bounded Instantaneous and Cumulative Rewards} If $\theta = \rewop{r_1}{}{\sinstant{\leq k_1}}{+}\rewop{r_2}{}{\scumul{\leq k_2}}$, we compute values of the objectives for the formulae $\theta_{n+n_1,n+n_2}=\rewop{r_1}{}{\sinstant{\leq n+n_1}}+\rewop{r_2}{}{\scumul{\leq n+n_2}}$ for $0 \leq n \leq k$ recursively, where $k = \min\{k_1,k_2\}$, $n_1 = k_1{-}k$ and $n_2 = k_2{-}k$. Again,  here we can only switch to MDP verification when one of the step bounds equals zero. For state $s$, if $n = 0$:
\[
\V_{\game^C}(s,\theta_{n_1,n_2}) = 
\begin{cases}
(r^1_S(s),0) & \mbox{if $n_1 = n_2 = 0$} \\
(r^1_S(s),\rewopR^{\max}_{\game,s}(r_2,\scumul{\leq n_2})) & \mbox{else if $n_1 = 0$} \\
(\rewopR^{\max}_{\game,s}(r_1,\sinstant{\leq n_1}),0) & \mbox{otherwise}
\end{cases}
\]
and if $n>0$, then $\V_{\game^C}(s,\theta_{n+n_1,n+n_2})$ equals SWNE values of the bimatrix game $(\mgame_1, \mgame_2) \in \Qset^{l \times m}$: 
\begin{align*}
z^1_{i,j} &= \; \mbox{$\sum_{s' \in S}$} \, \delta^C(s,(a_i,b_j))(s') \cdot v^{s',1}_{(n-1)+n_1} \\
z^2_{i,j} &= \; r^2_S(s) + r^2_A(s,(a_i,b_j)) + \mbox{$\sum_{s' \in S}$} \, \delta^C(s,(a_i,b_j))(s') \cdot v^{s',l}_{(n-1)+n_2}
\end{align*}
and $(v^{s',1}_{(n-1)+n_1},v^{s',2}_{(n-1)+n_2}) = \V_{\game^C}(s',\theta_{(n-1)+n_1,(n-1)+n_2})$ for all $s' \in S$. 
The symmetric case follows similarly.

\subsubsection{Computing SWNE Values of Infinite-Horizon Nonzero-Sum Formulae}\label{nonzero-unbounded-sect}

We next show how to compute SWNE values $\V_{\game^C}(s,\theta)$ for \emph{infinite-horizon} nonzero-sum formulae $\theta$ in all states $s$ of $\game^C$.
As for the zero-sum case, we approximate these using a value iteration approach.
Each step of this computation is similar in nature to the algorithms
in the previous section, where a bimatrix game is solved for each state,
and a reduction to solving an MDP is used
after one of the player's objective can no longer change.

A key aspect of the value iteration algorithm is that,
while the SWNE (or SCNE) values take the form of a pair, with one value for each player,
convergence is defined over the \emph{sum} of the two values.
This is because there is not necessarily a unique pair of such values,
but the maximum (or minimum) of the sum of NE values \emph{is} uniquely defined.
Convergence of value iteration is estimated in the same way as for the zero-sum
computation (see \sectref{inf-zero-sect}), by comparing values in successive iterations.
As previously, this means that we are not able to guarantee that the computed
values are within a particular error bound of the exact values.

Below, we give the algorithms for the cases of two infinite-horizon objectives.
The notation used is as in the previous section:
for any state formula $\phi$ and state $s$ we let $\eta_{\phi}(s)$ equal $1$ if $s \in \Sat(\phi)$ and $0$ otherwise; and values of the form $\probopP^{\max}_{\game,s}(\psi)$ and $\rewopR^{\max}_{\game,s}(r,\rho)$ are computed through standard MDP verification.
\startpara{Until}
If $\theta = \probop{}{\phi_1^1 \until \phi_2^1}{+}\probop{}{\phi_1^2 \until \phi_2^2}$, values for any state $s$ can be computed through value iteration as the limit $\V_{\game^C}(s,\theta) = \lim_{n \ra \infty} \V_{\game^C}(s,\theta,n)$ where:
\[
\V_{\game^C}(s,\theta,n) = \begin{cases}
(1,1) & \mbox{if $s \in \Sat(\phi_2^1) \cap \Sat(\phi_2^2)$} \\
(1,\probopP^{\max}_{\game,s}(\phi_1^2 \until \phi_2^2)) & \mbox{else if $s \in \Sat(\phi_2^1)$} \\
(\probopP^{\max}_{\game,s}(\phi_1^1 \until \phi_2^1),1) & \mbox{else if $s \in \Sat(\phi_2^2)$} \\
(\probopP^{\max}_{\game,s}(\phi_1^1 \until \phi_2^1),0) & \mbox{else if $s \in \Sat(\phi_1^1) \setminus \Sat(\phi_1^2)$} \\
(0,\probopP^{\max}_{\game,s}(\phi_1^2 \until \phi_2^2)) & \mbox{else if $s \in \Sat(\phi_1^2) \setminus \Sat(\phi_1^1)$} \\
(0,0) & \mbox{else if $n = 0$ or $s \not\in \Sat(\phi_1^1) \cap \Sat(\phi_1^2)$} \\
\val(\mgame_1, \mgame_2) & \mbox{otherwise}
\end{cases}
\]
where $\val(\mgame_1, \mgame_2)$ equals SWNE values of the bimatrix game $(\mgame_1,\mgame_2)\in \Qset^{l \times m}$:
\begin{align*}
z^1_{i,j} & = \; \mbox{$\sum_{s' \in S}$} \, \delta^C(s,(a_i,b_j))(s') \cdot v^{s',1}_{n-1} \\
z^2_{i,j} & = \; \mbox{$\sum_{s' \in S}$} \, \delta^C(s,(a_i,b_j))(s') \cdot v^{s',2}_{n-1} 
\end{align*}
and $(v^{s',1}_{n-1},v^{s',2}_{n-1}) = \V_{\game^C}(s',\theta,n{-}1)$ for all $s' \in S$.

As can be seen, there are two situations in which we switch to MDP verification. These correspond to the two cases where the value of the objective of one of the players cannot change: when a state satisfying $\phi_2^i$ is reached for only one player $i$ (and therefore the objective is satisfied by that state) and when a state satisfying $\neg \phi_1^i \wedge \neg \phi_2^i$ is reached for only one player $i$ (and therefore the objective is not satisfied by that state).
\startpara{Expected Reachability} If $\theta = \rewop{r_1}{}{\future \phi^1}{+}\rewop{r_2}{}{\future \phi^2}$, values can be computed through value iteration as the limit 
$\V_{\game^C}(s,\theta) = \lim_{n \ra \infty} \V_{\game^C}(s,\theta,n)$ where:
\[
\V_{\game^C}(s,\theta,n) = \begin{cases}
(0,0) & \mbox{if $s \in \Sat(\phi^1) \cap \Sat(\phi^2)$} \\
(0,0) & \mbox{else if $n = 0$} \\
(0,\rewopR^{\max}_{\game,s}(r_2,\future \phi^2)) & \mbox{else if $s \in \Sat(\phi^1)$} \\
(\rewopR^{\max}_{\game,s}(r_1,\future \phi^1),0) & \mbox{else if $s \in \Sat(\phi^2)$} \\
\val(\mgame_1, \mgame_2) & \mbox{otherwise}
\end{cases}
\]
where $\val(\mgame_1, \mgame_2)$ equals SWNE values of the bimatrix game $(\mgame_1,\mgame_2)\in \Qset^{l \times m}$:
\begin{align*}
z^1_{i,j} &= \; r^1_S(s) + r^1_A(s,(a_i,b_j)) + \mbox{$\sum_{s' \in S}$} \, \delta^C(s,(a_i,b_j))(s') \cdot v^{s',1}_{n-1} \\
z^2_{i,j} &= \; r^2_S(s) + r^2_A(s,(a_i,b_j)) + \mbox{$\sum_{s' \in S}$}\, \delta^C(s,(a_i,b_j))(s') \cdot v^{s',2}_{n-1} 
\end{align*}
and $(v^{s',1}_{n-1},v^{s',2}_{n-1}) = \V_{\game^C}(s',\theta,n{-}1)$ for all $s' \in S$. 

In this case, the only situation in which the value of the objective of one of the players cannot change is when only one of their goals is reached, i.e., when a state satisfying $\phi^i$ is reached for only one player $i$. This is therefore the only time we switch to MDP verification.

\subsubsection{Computing SWNE Values of Mixed Nonzero-Sum Formulae}\label{nonzero-mixed-sect}

We now present the algorithms for computing SWNE values of nonzero-sum formula containing a \emph{mixture} of both finite- and infinite-horizon objectives.
This is achieved by finding values for a sum of two modified (infinite-horizon) objectives $\theta'$ on a modified game $\game'$ using the algorithms presented in \sectref{nonzero-unbounded-sect}. This approach is based on the standard construction for converting the verification of finite-horizon properties to infinite-horizon properties~\cite{Put94}. We consider the cases when the first objective is finite-horizon and second infinite-horizon; the symmetric cases follow similarly. In each case, the modified game has states of the form $(s,n)$, where $s$ is a state of $\game^C$, $n \in \Nset$ and the SWNE values $\V_{\game^C}(s,\theta)$ are given by the SWNE values $\V_{\game'}((s,0),\theta')$.
\startpara{Next and Unbounded Until} If $\theta = \probop{}{\next \phi^1}{+}\probop{}{\phi_1^2 \until \phi_2^2}$, then we construct the game $\game' = (\{1,2\}, S', \bar{S}', A^C, \Delta', \delta', \{ \ap_{\phi^1}, \ap_{\phi_1^2}, \ap_{\phi_2^2} \} , \lab')$ where:
\begin{itemize}
\item
$S' = \{ (s,n) \mid s \in S \wedge 0 \leq n \leq 2 \}$ and $\bar{S}' = \{ (s,0) \mid s \in S \}$;
\item $\Delta'((s,n)) = \Delta^C(s)$ for all $(s,n) \in S'$;
\item for any $(s,n),(s',n') \in S'$ and $a \in A^C$:
\[
\delta'((s,n),a)((s',n')) = \begin{cases}
\delta^C(s,a)(s') & \mbox{if $0 \leq n \leq 1$ and $n' = n{+}1$} \\
\delta^C(s,a)(s') & \mbox{else if $n = n' = 2$} \\
0 & \mbox{otherwise;}
\end{cases}
\]
\item
for any $(s,n) \in S'$ and $1 \leq j \leq 2$:
\begin{itemize}
\item
$\ap_{\phi^1} \in \lab'((s,n))$ if and only if $s \in \Sat(\phi^1)$ and $n = 1$;
\item
$\ap_{\phi_j^2} \in \lab'((s,n))$ if and only if $s \in \Sat(\phi_j^2)$.
\end{itemize}
\end{itemize}
and compute the SWNE values of $\theta'=\probop{}{\true \until \ap_{\phi^1}}{+}\probop{}{\ap_{\phi_1^2} \until \ap_{\phi_2^2}}$ for $\game'$.
\startpara{Bounded and Unbounded Until} If $\theta = \probop{}{\phi_1^1 \buntilp{k_1} \phi_2^1} + \probop{}{\phi_1^2 \until \phi_2^2}$, then we construct the game
$\game' = (\{1,2\}, S', \bar{S}', A^C, \Delta', \delta', \{ \ap_{\phi_1^1} , \ap_{\phi_2^1}, \ap_{\phi_1^2}, \ap_{\phi_2^2} \}, \lab')$
where:
\begin{itemize}
\item
$S' = \{ (s,n) \mid s \in S \wedge 0 \leq n \leq k_1{+}1 \}$ and $\bar{S}' = \{ (s,0) \mid s \in S \}$;
\item $\Delta'((s,n)) = \Delta^C(s)$ for all $(s,n) \in S'$;
\item for any $(s,n),(s',n') \in S'$ and $a \in A^C$:
\[
\delta'((s,n),a)((s',n')) = \begin{cases}
\delta^C(s,a)(s') & \mbox{if $0 \leq n \leq k_1$ and $n' = n{+}1$} \\
\delta^C(s,a)(s') & \mbox{else if $n = n' = k_1{+}1$} \\
0 & \mbox{otherwise;}
\end{cases}
\]
\item
for any $(s,n) \in S'$ and $1 \leq j \leq 2$:
\begin{itemize}
\item
$\ap_{\phi_1^1} \in \lab'((s,n))$ if and only if $s \in \Sat(\phi_1^1)$ and $0 \leq n \leq k_1$;
\item
$\ap_{\phi_2^1} \in \lab'((s,n))$ if and only if $s \in \Sat(\phi_2^1)$ and $0 \leq n \leq k_1$;
\item
$\ap_{\phi_j^2} \in \lab'((s,n))$ if and only if $s \in \Sat(\phi_j^2)$.
\end{itemize}
\end{itemize}
and compute the SWNE values of $\theta'=\probop{}{\ap_{\phi_1^1} \until \ap_{\phi_2^1}}{+}\probop{}{\ap_{\phi_1^2} \until \ap_{\phi_2^2}}$ for $\game'$.
\startpara{Bounded Instantaneous and Expected Rewards} If $\theta = \rewop{r_1}{}{\sinstant{=k_1}}+\rewop{r_2}{}{\future \phi^2}$, then we construct the game $\game' = (\{1,2\}, S', \bar{S}', A^C, \Delta', \delta', \{ \ap_{k_1+1} , \ap_{\phi^2} \}, \lab')$ and reward structures $r_1'$ and $r_2'$ where:
\begin{itemize}
\item
$S' = \{ (s,n) \mid s \in S \wedge 0 \leq n \leq k_1{+}1 \}$ and $\bar{S}' = \{ (s,0) \mid s \in S \}$;
\item $\Delta'((s,n)) = \Delta^C(s)$ for all $(s,n) \in S'$;
\item for any $(s,n),(s',n') \in S'$ and $a \in A^C$:
\[
\delta'((s,n),a)((s',n')) = \begin{cases}
\delta^C(s,a)(s') & \mbox{if $0 \leq n \leq k_1$ and $n' = n{+}1$} \\
\delta^C(s,a)(s') & \mbox{else if $n = n' = k_1{+}1$} \\
0 & \mbox{otherwise;}
\end{cases}
\]
\item
for any $(s,n) \in S'$:
\begin{itemize}
\item
$\ap_{k_1+1} \in \lab'((s,n))$ if and only if $n = k_1{+}1$;
\item
$\ap_{\phi^2} \in \lab'((s,n))$ if and only if $s \in \Sat(\phi^2)$;
\end{itemize}
\item
for any $(s,n) \in S'$ and $a \in A^C$:
\begin{itemize}
\item
$r^{1'}_A((s,n),a)=0$ and $r^{1'}_S((s,n))=r_S^{1^C}(s)$ if $n = k_1$ and $r^{1'}_A((s,n),a)=0$ and $r^{1'}_S((s,n))=0$ otherwise;
\item
$r^{2'}_A((s,n),a)=r^{2^C}_A(s)(a)$ and $r^{2'}_S((s,n))=r_S^{2^C}(s)$.
\end{itemize}
\end{itemize}
and compute the SWNE values of $\theta'=\rewop{r_1'}{}{\future \ap_{k_1+1}}{+}\rewop{r_2'}{}{\future \ap_{\phi^2}}$ for $\game'$.
\startpara{Bounded Cumulative and Expected Rewards} If $\theta = \rewop{r_1}{}{\scumul{\leq k_1}}+\rewop{r_2}{}{\future \phi^2}$, then we construct the game $\game' = (\{1,2\}, S', \bar{S}', A^C, \Delta', \delta', \{ \ap_{k_1} , \ap_{\phi^2} \}, \lab')$ and reward structures $r_1'$ and $r_2'$ where:
\begin{itemize}
\item
$S' = \{ (s,n) \mid s \in S \wedge 0 \leq i \leq k_1 \}$ and $\bar{S}' = \{ (s,0) \mid s \in S \}$;
\item $\Delta'((s,n)) = \Delta^C(s)$ for all $(s,n) \in S'$;
\item for any $(s,n),(s',n') \in S'$ and $a \in A^C$:
\[
\delta'((s,n),a)((s',n')) = \begin{cases}
\delta^C(s,a)(s') & \mbox{if $0 \leq n \leq k_1{-}1$ and $n' = n{+}1$} \\
\delta^C(s,a)(s') & \mbox{else if $n = n' = k_1$} \\
0 & \mbox{otherwise;}
\end{cases}
\]
\item
for any $(s,n) \in S'$:
\begin{itemize}
\item
$\ap_{k_1} \in \lab'((s,n))$ if and only if $n = k_1$;
\item
$\ap_{\phi^2} \in \lab'((s,n))$ if and only if $s \in \Sat(\phi^2)$;
\end{itemize}
\item
for any $(s,n) \in S'$ and $a \in A^C$:
\begin{itemize}
\item
$r_A^{1'}((s,n))(a)=r_A^{1^C}(s)$ if $0 \leq n \leq k_1{-}1$ and equals 0 otherwise;
\item
$r_S^{1'}((s,n))=r_S^{1^C}(s)$ if $0 \leq n \leq k_1{-}1$ and equals 0 otherwise;
\item
$r^{2'}_A((s,n),a)=r^{2^C}_A(s)(a)$ and $r^{2'}_S((s,n))=r_S^{2^C}(s)$.
\end{itemize}
\end{itemize}
and compute the SWNE values of $\theta'=\rewop{r_1'}{}{\future \ap_{k_1}}{+}\rewop{r_2'}{}{\future \ap_{\phi^2}}$ for $\game'$.
\subsubsection{Computing SCNE Values of Nonzero-Sum Formulae}\label{scne-sect}

The case for SCNE values follows similarly to the SWNE case using backward induction for finite-horizon properties and value iteration for infinite-horizon properties. There are two differences in the computation. First, when solving MDPs, we find the minimum probability of satisfying path formulae and the minimum expected reward for reward formulae. Second, when solving the bimatrix games constructed during backward induction and value iteration, we find SCNE rather than SWNE values; this is achieved through \lemref{duality-lem}. More precisely, we negate all the utilities in the game, find the SWNE values of this modified game, then negate these values to obtain SCNE values of the original bimatrix game.

\subsection{Strategy Synthesis}

In addition to verifying formulae in our extension of rPATL, it is typically also very useful to perform
\emph{strategy synthesis}, i.e., to construct a witness to the satisfaction of a property.
For each zero-sum formula $\coalition{C}\probop{\sim q}{\psi}$ or $\coalition{C}\rewop{r}{\sim x}{\rho}$ appearing as a sub-formula,
this comprises optimal strategies for the players in coalition $C$
(or, equivalently, for player 1 in the coalition game $\game^C$)
for the objective $X^\psi$ or $X^{r,\rho}$.
For each nonzero-sum formula $\nashop{C{:}C'}{\opt \sim x}{\theta}$  appearing as a sub-formula, this is a subgame-perfect SWNE/SCNE profile for the objectives $(X^\theta_1,X^\theta_2)$ in the coalition game $\game^C$.

We can perform strategy synthesis by adapting the model checking algorithms described in the
previous sections which computes the values of zero-sum objectives and SWNE or SCNE values of nonzero-sum objectives.
The type of strategy needed (deterministic or randomised; memoryless or
finite-memory) depends on the types of objectives.
As discussed previously (in \sectsectref{nonzero-unbounded-sect}{inf-zero-sect}),
for infinite-horizon objectives our use of value iteration means
we cannot guarantee that the values computed are within a particular error
bound of the actual values; so, the same will be true of the optimal strategy
that we synthesise for such a formula.

\startpara{Zero-sum properties}
For zero-sum formulae, all strategies synthesised are randomised;
this is in contrast to checking the equivalent properties against TSGs~\cite{CFK+13b},
where deterministic strategies are sufficient.
For infinite-horizon objectives, we synthesise memoryless strategies,
i.e., a distribution over actions for each state of the game.
For finite-horizon objectives, strategies are finite-memory,
with a separate distribution required for each state and each time step.

For both types of objectives, we synthesise the strategies
whilst computing values using the approach presented in \sectref{zero-sect}:
from the matrix game solved for each state, we extract not just the value of the game,
but also an optimal (randomised) strategy for player 1 of $\game^C$ in that state.
It is also possible to extract the optimal strategy for player 2 in the state by solving the dual LP problem for the matrix game (see \sectref{matrix-sect}).
For finite-horizon objectives, we retain the choices for all steps;
for infinite-horizon objectives, just those from the final step of value iteration are needed.
\startpara{Nonzero-sum properties}
In the case of a nonzero-sum formula,
randomisation is again needed for all types of objectives.
Similarly to zero-sum formulae above,
strategies are generated whilst computing SWNE or SCNE values,
using the algorithms presented in \sectref{nonzero-sect}.
Now, we do this in two distinct ways:
\begin{itemize}
\item
when solving bimatrix games in each state, we also extract an SWNE/SCNE profile,
comprising the distributions over actions for each player of $\game^C$ in that state;
\item
when solving MDPs, we also synthesise an optimal strategy for the MDP~\cite{KP13},
which is equivalent to a strategy profile for $\game^C$ (in fact, randomisation is not needed for this part). 
\end{itemize}
The final synthesised profile is then constructed by initially following the ones generated when solving bimatrix games, and then switching to the MDP strategies if we reach a state where the value of one player's objective cannot change.
This means that all strategies synthesised for nonzero-sum formulae
may need memory. As for the zero-sum case, finite-horizon strategies
are finite-memory since separate player choices are stored for each state and each time step. But, in addition, for both finite- and infinite-horizon objectives,
one bit of memory is required to record that a switch is made to the strategy
extracted when solving the MDP.
\subsection{Complexity}

Due to its overall recursive nature, the complexity of our model checking algorithms are linear in the size of the formula $\phi$. In terms of the problems solved for each subformula, finding zero-sum values of a 2-player CSG is PSPACE~\cite{CH12} and finding subgame-perfect NE for reachability objectives of a 2-player CSG is PSPACE-complete~\cite{BBG+19}.
In practice, our algorithms are iterative, so the complexity depends on the
number of iterations required, the number of states in the CSG
and the problems solved for each state and in each step.

For finite-horizon objectives, the number of iterations is equal to the step-bound
in the formula. For infinite-horizon objectives, the number of iterations depends on the convergence criterion used. For zero-sum properties, an exponential lower bound has been shown for the worst-case number of iterations required for a non-trivial approximation~\cite{HIM11}.
We report on efficiency in practice in \sectref{expr-sect}.

In the case of zero-sum properties, for each state, at each iteration, we need to solve an LP problem of size $|A|$. Such problems can be solved using the simplex algorithm, which is PSPACE-complete~\cite{FS15}, but performs well on average~\cite{Todd02}. Alternatively, Karmarkar's algorithm \cite{Kar84} could be used, which is PTIME.

For nonzero-sum properties, in each state, at each iteration, 
we need to find all solutions to an LCP problem of size $|A|$. Papadimitriou established the complexity of solving the class of LCPs we encounter to be in PPAD (\emph{polynomial parity argument in a directed graph}) \cite{Pap94} and, to the best of our knowledge, there is still no polynomial algorithm for solving such problems. More closely related to finding all solutions, it has been shown that determining if there exists an equilibrium in a bimatrix game for which each player obtains a utility of a given bound is NP-complete~\cite{GZ89}. Also, it is demonstrated in \cite{ARSS10} that bimatrix games may have a number of NE that is exponential with respect to the size of the game, and thus any method that relies on finding all NE in the worst case cannot be expected to perform in a running time that is polynomial with respect to the size of the game.
\section{Correctness of the Model Checking Algorithms}

The overall (recursive) approach and the reduction to solution of a two-player game is essentially the same as for TSGs~\cite{CFK+13b}, and therefore the same correctness arguments apply. In the case of \emph{zero-sum} formulae, the correctness of value iteration for infinite-horizon properties follows from~\cite{RF91} and for finite-horizon properties from \defref{sem-def} and the solution of matrix games (see \sectref{prelim-sect}). Below, we show the correctness of the model checking algorithms for \emph{nonzero-sum} formulae.

\subsection{Nonzero-Sum Formulae}

We fix a game $\game$ and a nonzero-sum formula $\nashop{C{:}C'}{\opt \sim x}{\theta}$. For the case of \emph{finite-horizon} nonzero-sum formulae, the correctness of the model checking algorithms follows from the fact that we use backward induction~\cite{SW+01,NMK+44}. For \emph{infinite-horizon} nonzero-sum formulae, the proof is based on showing that the values computed during value iteration correspond to subgame-perfect SWNE values of finite game trees, and the values of these game trees converge uniformly and are bounded from above by the actual values of $\game^C$.

The fact that we use MDP model checking when the goal of one of the players is reached means that the values computed during value iteration are not finite approximations for the values of $\game^C$. Therefore we must also show that the values computed during value iteration are bounded from below by finite approximations for the values of $\game^C$. We first consider the case when both the objectives in the sum $\theta$ are infinite-horizon objectives. Below we assume $\opt = \max$ and the case when $\opt = \min$ follow similarly. For any $(v_1,v_2),(v_1',v_2') \in \Qset^2$, let $(v_1,v_2)\leq(v_1',v_2')$ if and only if $v_1 \leq v_1'$ and $v_2 \leq v_2'$. The following lemma follows by definition of subgame-perfect SWNE values.
\begin{lemma}\label{probrew-lem}
Consider any strategy profile $\sigma$ and state $s$ of $\game^C$ and let $(v_1^{\sigma,s},v_2^{\sigma,s})$ be the corresponding values of the players in $s$ for the objectives $(X^{\theta_1},X^{\theta_2})$. Considering subgame-perfect SWNE values of the objectives $(X^{\theta_1},X^{\theta_2})$ in state $s$, in the case that $\theta$ is of the form $\probop{}{\phi^1_1 \until \phi^1_2}{+}\probop{}{\phi^2_1 \until \phi^2_2}:$
\begin{itemize}
\item
if $s \sat \phi^1_2 \wedge \phi^2_2$, then $(1,1)$ are the unique subgame-perfect SWNE values for state $s$ and $(v_1^{\sigma,s},v_2^{\sigma,s}) \leq (1,1)$;
\item
if $s \sat \phi^1_2 \wedge \phi^2_1 \wedge \neg \phi^2_2$, then $(1,\probopP^{\max}_{\game,s}(\phi^2_1 \until \phi^2_2))$ are the unique subgame-perfect SWNE values for state $s$ and $(v_1^{\sigma,s},v_2^{\sigma,s}) \leq (1,\probopP^{\max}_{\game,s}(\phi^2_1 \until \phi^2_2))$;
\item
if $s \sat \phi^1_1 \wedge \neg \phi^1_2 \wedge \phi^2_2$, then $(\probopP^{\max}_{\game,s}(\phi^1_1 \until \phi^1_2),1)$ are the unique subgame-perfect SWNE values for state $s$ and $(v_1^{\sigma,s},v_2^{\sigma,s}) \leq (\probopP^{\max}_{\game,s}(\phi^1_1 \until \phi^1_2),1)$;
\item
if $s \sat \phi^1_2 \wedge \neg \phi^2_1 \wedge \neg \phi^2_2$, then $(1,0)$ are the unique subgame-perfect SWNE values for state $s$ and $(v_1^{\sigma,s},v_2^{\sigma,s}) \leq (1,0)$;
\item
if $s \sat \neg \phi^1_1 \wedge \neg \phi^1_2 \wedge \phi^2_2$, then $(0,1)$ are the unique subgame-perfect SWNE values for state $s$ and $(v_1^{\sigma,s},v_2^{\sigma,s}) \leq (0,1)$;
\item
if $s \sat \neg \phi^1_1 \wedge \neg \phi^1_2 \wedge \phi^2_1 \wedge \neg \phi^2_2$, then $(0,\probopP^{\max}_{\game,s}(\phi^2_1 \until \phi^2_2))$ are the unique subgame-perfect SWNE values for state $s$ and $(v_1^{\sigma,s},v_2^{\sigma,s}) \leq (0,\probopP^{\max}_{\game,s}(\phi^2_1 \until \phi^2_2))$;
\item
if $s \sat \phi^1_1 \wedge \neg \phi^1_2 \wedge \neg \phi^2_1 \wedge \neg \phi^2_2$, then $(\probopP^{\max}_{\game,s}(\phi^1_1 \until \phi^1_2),0)$ are the unique subgame-perfect SWNE values for state $s$ and $(v_1^{\sigma,s},v_2^{\sigma,s}) \leq (\probopP^{\max}_{\game,s}(\phi^1_1 \until \phi^1_2),0)$;
\item
if $s \sat \neg \phi^1_1 \wedge \neg \phi^1_2 \wedge \neg \phi^2_1 \wedge \neg \phi^2_2$, then $(0,0)$ are the unique subgame-perfect SWNE values for state $s$ and $(v_1^{\sigma,s},v_2^{\sigma,s}) \leq (0,0)$.
\end{itemize}
\noindent
On the other hand, in the case that $\theta$ is of the form $\rewop{r_1}{}{\future \phi^1}{+}\rewop{r_2}{}{\future \phi^2}:$
\begin{itemize}
\item
if $s \sat \phi^1 \wedge \phi^2$, then $(0,0)$ are the unique subgame-perfect SWNE values for state $s$ and $(v_1^{\sigma,s},v_2^{\sigma,s}) \leq (0,0)$;
\item
if $s \sat \phi^1 \wedge \neg \phi^2$, then $(0,\rewopR^{\max}_{\game,s}(r_2,\future \phi^2))$ are the unique subgame-perfect SWNE values for state $s$ and $(v_1^{\sigma,s},v_2^{\sigma,s}) \leq (0,\rewopR^{\max}_{\game,s}(r_2,\future \phi^2))$;
\item
if $s \sat \neg \phi^1 \wedge \phi^2$, then $(\rewopR^{\max}_{\game,s}(r_1,\future \phi^1),0)$ are the unique subgame-perfect SWNE values for state $s$ and $(v_1^{\sigma,s},v_2^{\sigma,s}) \leq (\rewopR^{\max}_{\game,s}(r_1,\future \phi^1),0)$.
\end{itemize}
\end{lemma}
Next we require the following objectives of $\game^C$.
\begin{definition}\label{bounded-objective-def}
For any sum of two probabilistic or reward objectives $\theta$, $1 \leq i \leq 2$ and $n \in \Nset$, let $X^\theta_{i,n}$ be the objective where for any path $\pi$ of $\game^C:$
\begin{eqnarray*}
X^{\probop{}{\phi^1_1 \! \until \! \phi^1_2}{+}\probop{}{\phi^2_1 \! \until \! \phi^2_2}}_{i,n}(\pi) & = & \begin{cases}
1 & \mbox{if $\exists k \leq n. \, ( \pi(k) \sat \phi^i_2 \wedge \forall j < k. \, \pi(j) \sat \phi^i_1)$} \\
0 & \mbox{otherwise}
\end{cases} \\
X^{\rewop{r_{\scale{.75}{1}}}{}{\future \! \phi^1}{+}\rewop{r_{\scale{.75}{2}}}{}{\future \! \phi^2}}_{i,n}(\pi) & = & \begin{cases}
\infty
& \mbox{if} \; \forall k \in \Nset . \, \pi(k) \notsat \phi^i \\
\mbox{$\sum_{k=0}^{k_{\phi_i}-1}$} \big( r_A(\pi(k),\pi[k])+r_S(\pi(k)) \big) & \mbox{if $k_{\phi^i} \leq n{-}1$} \\
0 & \mbox{otherwise}
\end{cases}
\end{eqnarray*}
\noindent
and $k_{\phi_i} = \min \{ k \mid k \in \Nset \wedge \pi(k) \sat \phi^i \}$.
\end{definition}
The following lemma demonstrates that, for a fixed strategy profile and state, the values of these objectives are non-decreasing and converge uniformly to the values of $\theta$.
\begin{lemma}\label{epsilon-lem}
For any sum of two probabilistic or reward objectives $\theta$ and $\varepsilon>0$, there exists $N \in \Nset$ such that, for any $n \geq N$, $s \in S$, $\sigma \in \Sigma^1_{\game^C} {\times} \Sigma^2_{\game^C}$ and $1 \leq i \leq 2:$
\[
0 \ \leq \ \Eset^{\sigma}_{\game^C,s}(X^\theta_i) - \Eset^{\sigma}_{\game^C,s}(X^\theta_{i,n}) \ \leq \ \varepsilon \, .
\]
\end{lemma}
\begin{proof}
Consider any sum of two probabilistic or reward objectives $\theta$, state $s$ and $1 \leq i \leq 2$. 
Using \assumref{game3-assum} we have that,  for subformulae $\rewop{r}{}{\future \phi^i}$, the set $\Sat(\phi^i)$ is reached with probability 1 from all states of $\game$ under all profiles, and therefore $\Eset^{\sigma}_{\game^\cC,s}(X^\theta_i)$ is finite. Furthermore, for any $n \geq N$, by \defdefref{sem-def}{bounded-objective-def} we have that $\Eset^{\sigma}_{\game^\cC,s}(X^\theta_{i,n})$ is the value of state $s$ for the $n$th iteration of value iteration~\cite{CH08} when computing $\Eset^{\sigma}_{\game^\cC,s}(X^\theta_i)$ in the DTMC obtained from $\game^\cC$ by following the strategy $\sigma$, and the sequence is both non-decreasing and converges. The fact that we can choose an $N$ independent of the strategy profile for uniform convergence follows from \assumassumref{game2-assum}{game3-assum}. \qed
\end{proof}
In the proof of correctness we will use the fact that $n$ iterations of value iteration is equivalent to performing backward induction on the following game trees.
\begin{definition}\label{trees-def}
For any state $s$ and $n \in \Nset$, let $\game^C_{n,s}$ be the game tree corresponding to playing $\game^C$ for $n$ steps when starting from state $s$ and then terminating. 
\end{definition}
We can map any strategy profile $\sigma$ of $\game^C$ to a strategy profile of $\game^C_{n,s}$ by only considering the choices of the profile over the first $n$ steps when starting from state $s$. This mapping is clearly surjective, i.e., we can generate all profiles of $\game^C_{n,s}$, but is not injective.
We also need the following objectives corresponding to the values computed during value iteration for the game trees of \defref{trees-def}.
\begin{definition}\label{value-objective-def}
For any sum of two probabilistic or reward objectives $\theta$, $s \in S$, $n \in \Nset$, $1 \leq i \leq 2$ and $j= i{+}1 \bmod 2$, let $Y^\theta_i$ be the objective where, for any path $\pi$ of $\game^C_{n,s}:$
\begin{align*}
\lefteqn{Y^{\probop{}{\phi^1_1 \! \until \! \phi^1_2}{+}\probop{}{\phi^2_1 \! \until \! \phi^2_2}}_i(\pi) =} \\
 & \begin{cases}
1 & \mbox{if $\exists m \leq n . \, (\pi(m) \sat \phi^i \wedge \forall k < m . \, \pi(k) \sat \phi^1_1 {\wedge} \neg \phi^1_2 {\wedge} \phi^2_1 {\wedge} \neg \phi^2_2)$} \\
\probopP^{\max}_{\game,\pi(m)}(\phi^i_1 \until \phi^i_2) & \mbox{else if $\exists m \leq n . \, (\pi(m) \sat \phi^j \wedge \forall k < m . \, \pi(k) \sat \phi^1_1 {\wedge} \neg \phi^1_2 {\wedge} \phi^2_1 {\wedge} \neg \phi^2_2)$}\\
0 & \mbox{otherwise} 
\end{cases}
\\
\lefteqn{Y^{\rewop{r_{\scale{.75}{1}}}{}{\future \! \phi^1}{+}\rewop{r_{\scale{.75}{2}}}{}{\future \! \phi^2}}_i(\pi) =} \\
& \begin{cases}
\infty
& \mbox{if} \; \forall k \leq n . \, \pi(k) \notsat \phi^i \\
\mbox{$\sum_{k=0}^{k_{\phi^{\scale{.75}{1}}\vee\phi^{\scale{.75}{2}}}-1}$} \big( r_A(\pi(k),\pi[k]) + r_S(\pi(k)) \big) + r^i_S(\pi(k)) & \mbox{otherwise}
\end{cases}
\end{align*}
\noindent
where
\[
r^i_S(s') \; = \; \begin{cases}
\rewopR^{\max}_{\game,s'}(r_i,\future \phi^{i}) & \mbox{if $s \sat \neg \phi^i \wedge \phi^j$}\\
0 & \mbox{otherwise} 
\end{cases}
\]
for $s' \in S$ and $k_{\phi^1\vee\phi^2} = \min \{ k \mid k \leq n \wedge \pi(k) \sat \phi^1 \vee \phi^2 \}$. %
\end{definition}
Similarly to \lemref{epsilon-lem}, the lemma below demonstrates, for a fixed strategy profile and state $s$ of $\game^C$, that the values for the objectives given in \defref{value-objective-def} when played on the game trees $\game^C_{n,s}$ are non-decreasing and converge uniformly. As with \lemref{epsilon-lem} the result follows from \assumassumref{game2-assum}{game3-assum}.
\begin{lemma}\label{epsilon2-lem}
For any sum of two probabilistic or reward objectives $\theta$ and $\varepsilon>0$, there exists $N \in \Nset$ such that for any $m \geq n \geq N$, $\sigma \in \Sigma^1_{\game^C} {\times} \Sigma^2_{\game^C}$, $s \in S$ and $1 \leq i \leq 2:$
\[
0 \ \leq \ \Eset^{\sigma}_{\game^C_{m,s}}(Y^\theta_i) - \Eset^{\sigma}_{\game^C_{n,s}}(Y^\theta_i) \ \leq \ \varepsilon \, .
\]
\end{lemma}
We require the following lemma relating the values of the objectives $X^\theta_{i,n}$, $Y^\theta_i$ and $X^\theta_i$ for $1 \leq i \leq 2$.
\begin{lemma}\label{precomp-lem}
For any sum of two probabilistic or reward objectives $\theta$, state $s$ of $\game^C$, strategy profile $\sigma$ such that when one of the targets of the objectives of $\theta$ is reached, the profile then collaborates to maximise the value of the other objective, $n \in \Nset$ and $1 \leq i \leq 2:$
\[
\sup\nolimits_{\sigma_i \in \Sigma_i^{\game^{\scale{.75}{C}}_{\scale{.75}{n,s}}}} \Eset^{\sigma_{-i}[\sigma_i]}_{\game^C,s}(X^\theta_{i,n}) \ \leq \ \sup\nolimits_{\sigma_i \in \Sigma_i^{\game^{\scale{.75}{C}}_{\scale{.75}{n,s}}}} \Eset^{\sigma_{-i}[\sigma_i]}_{\game^C_{n,s}}(Y^\theta_i)
\ \leq \
\sup\nolimits_{\sigma_i \in \Sigma_i^{\game^{\scale{.75}{C}}}} \Eset^{\sigma_{-i}[\sigma_i]}_{\game^C,s}(X^\theta_i) \, .
\]
\end{lemma}
\begin{proof}
Consider any strategy profile $\sigma$, $n \in \Nset$ and $1 \leq i \leq 2$. By \defdefref{bounded-objective-def}{value-objective-def} it follows that:
\[
\Eset^{\sigma}_{\game^C,s}(X^\theta_{i,n}) \ \leq \ \Eset^{\sigma}_{\game^C_{n,s}}(Y^\theta_i) .
\]
Furthermore, if we restrict the profile $\sigma$ such that, when one of the targets of the objectives of $\theta$ is reached, the profile then collaborates to maximise the value of the other objective, then by \defdefref{value-objective-def}{sem-def}:
\[
\Eset^{\sigma}_{\game^C_{n,s}}(Y^\theta_i) \ \leq \ 
\Eset^{\sigma}_{\game^C,s}(X^\theta_i) .
\]
Combining these results with \lemref{probrew-lem}, we have:
\[
\sup\nolimits_{\sigma_i \in \Sigma_i^{\game^{\scale{.75}{C}}_{\scale{.75}{n,s}}}} \Eset^{\sigma_{-i}[\sigma_i]}_{\game^C,s}(X^\theta_{i,n}) \ \leq \ \sup\nolimits_{\sigma_i \in \Sigma_i^{\game^{\scale{.75}{C}}_{\scale{.75}{n,s}}}} \Eset^{\sigma_{-i}[\sigma_i]}_{\game^C_{n,s}}(Y^\theta_i)
\ \leq \
\sup\nolimits_{\sigma_i \in \Sigma_i^{\game^{\scale{.75}{C}}}} \Eset^{\sigma_{-i}[\sigma_i]}_{\game^C,s}(X^\theta_i) 
\]
as required. \qed
\end{proof}
We now define the strategy profiles synthesised during value iteration.
\begin{definition}
For any $n \in \Nset$ and $s \in S$, let $\sigma^{n,s}$ be the strategy profile generated for the game tree $\game^C_{n,s}$ (when considering value iteration as backward induction) and $\sigma^{n,\star}$ be the synthesised strategy profile for $\game^C$ after $n$ iterations.
\end{definition}
Before giving the proof of correctness we require the following results.
\begin{lemma}\label{backwards-lem}
For any state $s$ of $\game^C$, sum of two probabilistic or reward objectives $\theta$ and $n \in \Nset$ we have that $\sigma^{n,s}$ is a subgame-perfect SWNE profile of the CSG $\game^C_{n,s}$ for the objectives $(Y^{\theta_1},Y^{\theta_2})$.
\end{lemma}
\begin{proof}
The result follows from the fact that value iteration selects SWNE profiles, value iteration corresponds to performing backward induction for the objectives $(Y^{\theta_1},Y^{\theta_2})$ and backward induction returns a subgame-perfect NE~\cite{SW+01,NMK+44}. \qed
\end{proof}
The following proposition demonstrates that value iteration converges and depends on \assumassumref{game2-assum}{game3-assum}. Without this assumption convergence cannot be guaranteed as demonstrated by the counterexamples in \appref{3-app} and \appref{4-app}. Although value iteration converges, unlike value iteration for MDPs or zero-sum games, the generated sequence of values is not necessarily non-decreasing. 
\begin{proposition}\label{convergence-prop}
For any sum of two probabilistic or reward objectives $\theta$ and state $s$, the sequence $\langle \V_{\game^C}(s,\theta,n) \rangle_{n \in \Nset}$ converges.
\end{proposition}
\begin{proof}
For any state $s$ and $n \in \Nset$ we can consider $\game^C_{n,s}$ as two-player infinite-action NFGs $\nfgame_{n,s}$ where for $1 \leq i \leq 2$:
\begin{itemize}
\item
the set of actions of player $i$ equals the set of strategies of player $i$ in $\game^C$;
\item for the action pair $(\sigma_1,\sigma_2)$, the utility function for player $i$ returns $\Eset^{\sigma}_{\game^C_{n,s}}(Y^\theta_i)$.
\end{itemize}
The correctness of this construction relies on the mapping of strategy profiles from the game $\game^C$ to $\game^C_{n,s}$ being surjective. Using \lemref{epsilon2-lem}, we have that the sequence $\langle \nfgame_{n,s} \rangle_{n \in N}$ of NFGs converges uniformly, and therefore, since $\V_{\game^C}(s,\theta,n)$ are subgame-perfect SWNE values of $\game^C_{n,s}$ (see \lemref{backwards-lem}), the sequence $\langle \V_{\game^C}(s,\theta,n) \rangle_{n \in \Nset}$ also converges. \qed
\end{proof}
A similar convergence result to \propref{convergence-prop} has been shown for the simpler case of discounted properties in~\cite{FL83}.
\begin{lemma}\label{strats-lem}
For any $\varepsilon>0$, there exists $N \in \Nset$ such that for any $s \in S$ and $1 \leq i \leq 2$:
\[
\big| \, \Eset^{\sigma^{n,\star}}_{\game^C,s}(X^\theta_i) - \Eset^{\sigma^{n,s}}_{\game^C_{n,s}}(Y^\theta_i) \, \big|
\ \leq \ \varepsilon \, .
\]
\end{lemma}
\begin{proof}
Using \lemref{epsilon2-lem} and \propref{convergence-prop}, we can choose $N$ such that the choices of the profile $\sigma^{n,s}$ agree with those of $\sigma^{n,\star}$ for a sufficient number of steps such that the inequality holds. \qed
\end{proof}
\begin{theorem}
For a given sum of two probabilistic or reward objectives $\theta$ and $\varepsilon>0$, there exists $N \in \Nset$ such that for any $n \geq N$ the strategy profile $\sigma^{n,\star}$ is a subgame-perfect $\varepsilon$-SWNE profile of $\game^C$ and the objectives $(X^{\theta_1},X^{\theta_2})$.
\end{theorem}
\begin{proof}
Consider any $\varepsilon>0$. From \lemref{strats-lem} there exists $N_1 \in \Nset$ such that for any $s\in S$ and $n \geq N_1$:
\begin{equation}\label{1-eqn}
\big| \, \Eset^{\sigma^{n,\star}}_{\game^C,s}(X^\theta_i) - \Eset^{\sigma^{n,s}}_{\game^C_{n,s}}(Y^\theta_i) \, \big|
\ \leq \ \mbox{$\frac{\varepsilon}{2}$} \, .
\end{equation}
For any $m \in \Nset$ and $s \in S$, using \lemref{backwards-lem} we have that $\sigma^{m,s}$ is a NE of $\game^C_{m,s}$, and therefore for any $m \in \Nset$, $s\in S$ and $1 \leq i \leq 2$:
\begin{equation}\label{nash-eqn}
\Eset^{\sigma^{m,s}}_{\game^{C}_{m,s}}(Y^\theta_i) 
\ \geq \
\sup\nolimits_{\sigma_i \in \Sigma_i^{\game^{\scale{.75}{C}}_{\scale{.75}{m,s}}}} \Eset^{\sigma^{m,s}_{-i}[\sigma_i]}_{\game^C_{m,s}}(Y^\theta_i) \, .
\end{equation}
From \lemref{epsilon-lem} there exists $N_2 \in \Nset$ such that for any $n \geq N_2$, $s \in S$ and $1 \leq i \leq 2$:
\begin{equation}\label{2-eqn}
\sup\nolimits_{\sigma_i \in \Sigma_i^{\game^{\scale{.75}{C}}}} \Eset^{\sigma^{n,\star}_{-i}[\sigma_i]}_{\game^C,s}(X^\theta_i) - \sup\nolimits_{\sigma_i \in \Sigma_i^{\game^{\scale{.75}{C}}}} \Eset^{\sigma^{n,\star}_{-i}[\sigma_i]}_{\game^C,s}(X^\theta_{i,n}) 
\ \leq \ \mbox{$\frac{\varepsilon}{2}$} \, .
\end{equation}
By construction, $\sigma^{n,\star}$ is a profile for which, if one of the targets of the objectives of $\theta$ is reached, the profile maximises the value of the objective. We can thus rearrange \eqnref{2-eqn} and apply \lemref{precomp-lem} to yield for any $n \geq N_2$, $s \in S$ and $1 \leq i \leq 2$:
\begin{equation}\label{3-eqn}
\sup\nolimits_{\sigma_i \in \Sigma_i^{\game^{\scale{.75}{C}}_{\scale{.75}{n,s}}}} \Eset^{\sigma^{n,s}_{-i}[\sigma_i]}_{\game^C_{n,s}}(Y^\theta_i)
\ \geq \
\sup\nolimits_{\sigma_i \in \Sigma_i^{\game^{\scale{.75}{C}}}} \Eset^{\sigma^{n,\star}_{-i}[\sigma_i]}_{\game^C,s}(X^\theta_i) - \mbox{$\frac{\varepsilon}{2}$} \, .
\end{equation}
Letting $N = \max \{ N_1 , N_2 \}$, for any $n \geq N$, $s \in S$ and $1 \leq i \leq 2$:
\begin{align*}
\Eset^{\sigma^{n,\star}}_{\game^C,s}(X^\theta_i) \ \ & \geq \ \ \Eset^{\sigma^{n,s}}_{\game^C_{n,s}}(Y^\theta_i) - \mbox{$\frac{\varepsilon}{2}$} & \mbox{by \eqnref{1-eqn} since $n \geq N_1$} \\
&\geq \ \ \sup\nolimits_{\sigma_i \in \Sigma_i^{\game^{\scale{.75}{C}}_{\scale{.75}{n,s}}}} \Eset^{\sigma^{n,s}_{-i}[\sigma_i]}_{\game^C_{n,s}}(Y^\theta_i) - \mbox{$\frac{\varepsilon}{2}$} & \mbox{by \eqnref{nash-eqn}} \\
&\geq \ \ \left( \sup\nolimits_{\sigma_i \in \Sigma_i^{\game^{\scale{.75}{C}}}} \Eset^{\sigma^{n,\star}_{-i}[\sigma_i]}_{\game^C,s}(X^\theta_i) - \mbox{$\frac{\varepsilon}{2}$} \right) - \mbox{$\frac{\varepsilon}{2}$} & \mbox{by \eqnref{3-eqn} since $n \geq N_2$} \\
& = \ \ \sup\nolimits_{\sigma_i \in \Sigma_i^{\game^{\scale{.75}{C}}}} \Eset^{\sigma^{n,\star}_{-i}[\sigma_i]}_{\game^C,s}(X^\theta_i) - \varepsilon
\end{align*}
and hence, since $\varepsilon>0$, $s \in S$ and $1 \leq i \leq 2$ were arbitrary, $\sigma^{n,\star}$ is a subgame-perfect $\varepsilon$-NE. It remains to show that the strategy profile is a subgame-perfect social welfare optimal $\varepsilon$-NE, which follows from the fact that when solving the bimatrix games during value iteration social welfare optimal NE are returned. \qed
\end{proof}
It remains to consider the model checking algorithms for nonzero-sum properties for which the sum of objectives contains both a finite-horizon and an infinite-horizon objective. In this case (see \sectref{nonzero-mixed-sect}), for a given game $\game^C$ and sum of objectives $\theta$, the algorithms first build a modified game $\game'$ with states $S' \subseteq S{\times}\Nset$ and sum of infinite-horizon objectives $\theta'$ and then computes SWNE/SCNE values of $\theta'$ in $\game'$. The correctness of these algorithms follows by first showing there exists a bijection between the profiles of $\game^C$ and $\game'$ and then that, for any profile $\sigma$ of $\game^C$ and $\sigma'$, the corresponding profile of $\game'$ under this bijection, we have:
\[
\Eset^{\sigma}_{\game^{C},s}(X_i^{\theta}) = \Eset^{\sigma'}_{\game',(s,0)}(X_i^{\theta'})
\]
for all states $s$ of $\game^C$ and $1 \leq i \leq 2$. This result follows from the fact that in \sectref{nonzero-mixed-sect} we used a standard construction for converting the verification of finite-horizon properties to infinite-horizon properties.

\section{Implementation and Tool Support}

We have implemented support for modelling and automated verification of CSGs in PRISM-games 3.0~\cite{KNPS20}, which previously only handled TSGs and zero-sum objectives~\cite{KPW18}. The PRISM-games tool is available from~\cite{pgwww} and the files for the case studies, described in the next section, are available from~\cite{files}.

\subsection{Modelling}

We extended the PRISM-games modelling language to support specification of CSGs.
The language allows multiple parallel components, called modules, operating both asynchronously and synchronously. Each module's state is defined by a number of finite-valued variables, and its behaviour is defined using probabilistic guarded commands of the form ${[}a{]}\ g \rightarrow u$, where $a$ is an action label, $g$ is a guard (a predicate over the variables of all modules) and $u$ is a probabilistic state update. If the guard is satisfied then the command is enabled, and the module can (probabilistically) update its variables according to $u$.
The language also allows for the specification of cost or reward structures.
These are defined in a similar fashion to the guarded commands,
taking the form ${[}a{]}\ g : v$ (for action rewards) and $g : v$ (for state rewards),
where $a$ is an action label, $g$ is a guard and $v$ is a real-valued expression over variables.

For CSGs, we assign modules to players and, in every state of the model,
each player can choose between the enabled commands of the corresponding modules
(or, if no command is enabled, the player idles).
In contrast to %
the usual behaviour of PRISM, where modules synchronise on common actions,
in CSGs action labels are distinct for each %
player and the players move concurrently.
To allow the updates of variables to depend on the choices of other players,
we extend the language by allowing commands to be labelled with lists of actions ${[}a_1,\dots,a_n{]}$.
Moreover, %
updates to variables can 
be dependent on the new values
of other variables being updated in the same concurrent transition,
provided there are no cyclic dependencies.
This %
ensures that variables of different players are %
updated according to a joint probability distribution. 
Another addition is the possibility of specifying ``independent'' modules, that is,
modules not associated with a specific player,
which do not feature nondeterminism and update their own variables
when synchronising with other players' actions.
Reward definitions are also extended to use action lists, similarly to commands,
so that an action reward can depend on the choices taken by multiple players. 
For further details of the new PRISM-games modelling language,
we refer the reader to the tool documentation~\cite{pgwww}.

\subsection{Implementation}\label{impl-sect}

PRISM-games constructs a CSG from a given model specification and implements the rPATL model checking and strategy synthesis algorithms from \sectref{mc-sect}.
We extend existing functionality within the tool,
such as modelling and property language parsers, the simulator
and basic model checking functionality.
We build, store and verify CSGs using an extension of PRISM's `explicit' model checking engine,
which is based on sparse matrices and implemented in Java.
For strategy synthesis we have included the option to export the generated strategies to a graphical representation using the Dot language~\cite{dot}.

Computing values (and optimal strategies) of matrix games (see~\sectref{matrix-sect}),
as required for zero-sum formulae,
is performed using the LPSolve library~\cite{LPS} via linear programming.
This library is based on the revised simplex and branch-and-bound methods. 
Computing SWNE or SCNE values (and SWNE or SCNE strategies) of bimatrix games (see~\sectref{bimatrix-sect})
for nonzero-sum formulae is performed via labelled polytopes
through a reduction to SMT. 
Currently, we implement this in both Z3~\cite{Z3} and Yices~\cite{Dut14}. As an optimised precomputation step, when possible we also search for and filter out \emph{dominated strategies}, which speeds up computation and reduces calls to the solver.

Since bimatrix games can have multiple SWNE values, when selecting SWNE values of such games
we choose %
the SWNE values for which the value of player 1 is %
maximal.
In case player 1 is indifferent, i.e., their utility is the same for all pairs,
we choose the SWNE values which maximise the value of player 2.
If both players are indifferent, an arbitrary pair of SWNE values is selected.

\begin{table}[t]
\centering
{\scriptsize
\begin{tabular}{|c|r||r|r|r|r|r|r|} \hline
\multirow{3}{*}{Bimatrix Game} & 
\multicolumn{1}{c||}{Actions}
& \multicolumn{3}{c|}{SWNE values} & \multicolumn{3}{c|}{SCNE values} \\ \cline{3-8} 
& 
\multicolumn{1}{c||}{of each}
& \multicolumn{2}{c|}{$\!\!$Solution time (s)$\!\!$} & \multicolumn{1}{c|}{$\!\!$Num. of$\!\!$} & \multicolumn{2}{c|}{$\!\!$Solution time (s)$\!\!$} & \multicolumn{1}{c|}{$\!\!$Num. of$\!\!$} \\ \cline{3-4} \cline{6-7}
& \multicolumn{1}{c||}{player} & \multicolumn{1}{c|}{Yices} & \multicolumn{1}{c|}{Z3} & \multicolumn{1}{c|}{$\!\!$ NE in $\nfgame$ $\!\!$} & \multicolumn{1}{c|}{Yices} & \multicolumn{1}{c|}{Z3} & \multicolumn{1}{c|}{$\!\!$ NE in $\nfgame^{-}$ $\!\!$}
 \\ \hline \hline
\multirow{5}{*}{\shortstack[c]{\emph{Covariant} \\ \emph{games}}}
&  2 &   0.04 &  0.06 &   1 & 0.2   & 1.1   &   1 \\
&  4 &   0.04 &  0.1  &   3 & 0.04  & 0.1   &   3 \\
&  8 &   0.3  &  2.7  &  13 & 0.3  & 1.8    &  21 \\
& 12 &   8.3  & 77.0  &  15 & 11.0  & 130.2 &  45 \\
& 16 & 764.8  & 4,238 & 103 & 318.6 & 2,627 & 109 \\
\hline\hline
\multirow{4}{*}{\shortstack[c]{\emph{Dispersion} \\ \emph{games}}}
&  2 & 0.04  & 0.08 &       3 & 0.04 & 0.08 &      3 \\
&  4 & 0.05  & 0.2  &      51 & 0.04 & 0.1  &     15 \\
&  8 & 1.8   & 16.4 &   6,051 & 0.2  & 1.1  &    255 \\
& 12 & 6,368 &  t/o & 523,251 & 6.0  & 38.3 & 4,095 \\
\hline\hline
\multirow{5}{*}{\shortstack[c]{\emph{Majority voting} \\ \emph{games}}}
&  2 &   0.03 &  0.08 &      5 &   0.03 &  0.07 &      2 \\
&  4 &   0.05 &   0.1 &     15 &   0.04 &  0.1 &     13 \\
&  8 &   0.5  &   1.0 &    433 &    0.2 &  0.7 &    186 \\
& 12 &   9.9  &  22.4 &  3,585 &    9.9 &  16.0 &  3,072 \\
& 16 & 532.5  & 2,386 & 61,441 &  465.9 & 2,791 & 49,153 \\
\hline\hline
\multirow{5}{*}{$\!\!\!$\shortstack[c]{\emph{Randomly generated} \\ \emph{games}}$\!\!\!$}
&  2 &  0.15 &   1.9  &  1 &  0.03 &   0.08 &  1 \\
&  4 &  0.04 &   0.1  &  3 &  0.03 &   0.09 &  3 \\
&  8 &   0.4 &   2.3  & 13 &  0.3  &    1.8 &  7 \\
& 12 & 45.60 & 422.0  & 27 & 68.0  &  343.6 & 31 \\
& 16 & 2,370 &   t/o  & 81 & 1,112 &    t/o & 69 \\
\hline
\end{tabular}}
\vspace*{-0.2cm}
\caption{Finding SWNE/SCNE values in bimatrix games: comparing SMT solvers.}\label{tab:smt-stats}
\vspace*{-0cm}
\end{table}
\tabref{tab:smt-stats} presents experimental results for the time to solve bimatrix games using the Yices and Z3 solvers, as the numbers of actions of the individual games vary.
The table also shows the number of NE in each game $\nfgame$, as found
when determining the SWNE values, and also the number of NE in $\nfgame^{-}$,
as found when determining the SCNE values (see \lemref{duality-lem}).
These games were generated using GAMUT (a suite of game generators)~\cite{NWSL04}
and a time-out of 2 hours was used for the experiments.
The results show Yices to be the faster implementation and that the difference in solution time grows as the number of actions increases. Therefore, in our experimental results in the next section, all verification runs use the Yices implementation. The results in \tabref{tab:smt-stats} also demonstrate that the solution time for either solver can vary widely
and depends on both the number of NE that need to be found and the structure of the game. For example, when solving the dispersion games, the differences in the solution times for SWNE and SCNE seem to correspond to the differences in the number of NE that need to found. On the other hand, there is no such correspondence between the difference in the solution times for the covariant games.

Regarding the complexity of solving bimatrix games, if each player has $n$ actions, then the number of possible assignments to the supports of the strategy profiles
(i.e., the action tuples that are chosen with nonzero probability)
is $(2^n{-}1)^2$, which therefore grows exponentially with the number of actions, surpassing 4.2 billion when each player has $16$ actions. This particularly affects performance in cases where one or both players are \emph{indifferent} with respect to a given support. More precisely, in such cases, if there is an equilibrium including pure strategies over these supports, then there are also equilibria including mixed strategies over these supports as the indifferent player would get the same utility for \emph{any affine combination} of pure strategies. 

\begin{examp}\label{indiff-eg}
Consider the following bimatrix game:
\[
\mgame_1 = 
\kbordermatrix{
 & b_1 & b_2 \cr 
 a_1 & 1 & 0 \cr 
 a_2 & 1 & 0
 }
\qquad
\mgame_2 = 
\kbordermatrix{
 & b_1 & b_2 \cr 
 a_1 & 2 & 2 \cr 
 a_2 & 4 & 4
 }
\]
Since the entries in the rows for the utility matrix for player 1 are the same and the columns are the same for player 2, it is easy to see that both players are indifferent with respect to their actions. As can be seen in \tabref{tab:ex-indiff}, all $(2^2{-}1)^2 = 9$ possible support assignments lead to an equilibrium. %
\end{examp}
\begin{table}[t]
\centering
{\scriptsize
\begin{tabular}{|r|r|r|r|c|} \hline
\multicolumn{2}{|c|}{Player 1 strategy} & 
\multicolumn{2}{c|}{Player 2 strategy} & 
\multicolumn{1}{c|}{Utilities} \\ \cline{1-4}	
prob.\ $a_1$ & prob.\ $a_2$ & prob.\ $b_1$ & prob.\ $b_2$ &
$(u_1,u_2)$ \\ \hline\hline
0.0 & 1.0 & 0.0 & 1.0 & (0.0,4.0) \\ \hline 
0.0 & 1.0 & 1.0 & 0.0 & (1.0,4.0) \\ \hline
0.0 & 1.0 & 0.5 & 0.5 & (0.5,4.0) \\ \hline
1.0 & 0.0 & 0.0 & 1.0 & (0.0,2.0) \\ \hline
1.0 & 0.0 & 1.0 & 0.0 & (1.0,2.0) \\ \hline
1.0 & 0.0 & 0.5 & 0.5 & (0.5,2.0) \\ \hline
0.5 & 0.5 & 0.0 & 1.0 & (0.0,3.0) \\ \hline
0.5 & 0.5 & 1.0 & 0.0 & (1.0,3.0) \\ \hline
0.5 & 0.5 & 0.5 & 0.5 & (0.5,3.0)  \\
\hline
\end{tabular}}
\vspace*{-0.2cm}
\caption{Possible NE strategies and utilities of the bimatrix game of \egref{indiff-eg}.}\label{tab:ex-indiff}
\vspace*{-0.4cm}
\end{table}
\noindent
For the task of computing non-optimal NE values, the large number of supports can be somewhat mitigated by eliminating \emph{weakly dominated} strategies~\cite{NRTV07}. However, removing such strategies is not a straightforward task when computing SWNE or SCNE values, since it can lead to the elimination of SWNE or SCNE profiles, and hence also SWNE or SCNE values. For example, if we removed the row corresponding to action $a_2$ or the column corresponding to action $b_1$ from the matrices in \egref{indiff-eg} above, then we eliminate a SWNE profile. As the number of actions for each player increases, the number of NE profiles also tends to increase and so does the likelihood of indifference. Naturally, the number of actions also affects the number of variables that have to be allocated, and the number and complexity of assertions passed to the SMT solver. As our method is based on the progressive elimination of support assignments that lead to NE, it takes longer to find SWNE and SCNE values as the number of possible supports grows and further constraints are added each time an equilibrium is found.
\section{Case Studies and Experimental Results}

To demonstrate the applicability and benefits of our techniques, and to evaluate their performance,
we now present results from a variety of case studies.
Supporting material for these examples (models and properties)
is available from~\cite{files}. These can be run with PRISM-games 3.0~\cite{KNPS20}. 

\begin{table}[t]
\centering
{\scriptsize
\begin{tabular}{|c|r||r|r|r|r|r|} \hline
\multicolumn{1}{|c|}{Case study} & 
\multicolumn{1}{c||}{Param.} & 
\multicolumn{1}{c|}{Players} & \multicolumn{1}{c|}{States} &  \multicolumn{1}{c|}{Transitions} 
& \multicolumn{1}{c|}{Constr.} \\ \multicolumn{1}{|c|}{[parameters]} &
\multicolumn{1}{c||}{values} & & & & \multicolumn{1}{c|}{time (s)}
 \\ \hline \hline	
\multirow{5}{*}{\shortstack[c]{\emph{Robot coordination} \\
$\mbox{[$l$]}$}}
&  4 & 2 &     226 &      6,610 & 0.1 \\
&  8 & 2 &   3,970 &    201,650 & 1.0 \\
& 12 & 2 &  20,450 &  1,221,074 & 3.9 \\
& 16 & 2 &  65,026 &  4,198,450 & 11.8 \\
& 24 & 2 & 330,626 & 23,049,650 & 62.5 \\
\hline\hline
\multirow{5}{*}{\shortstack[c]{\emph{Future markets} \\ \emph{investors} \\
$\mbox{[$\mathit{months}$]}$}}
&  6  & 3 &    34,247 &    144,396 &  1.5 \\
&  12 & 3 &   257,301 &  1,259,620 &  8.3 \\
&  24 & 3 &   829,125 &  4,318,372 & 25.4 \\
&  36 & 3 & 1,400,949 &  7,377,124 & 59.3 \\
&  48 & 3 & 1,972,773 & 10,435,876 & 70.3 \\
\hline \hline		 	
\multirow{4}{*}{\shortstack[c]{\emph{Future markets} \\ \emph{investors} \\
$\mbox{[$\mathit{months}$,$\mathit{pbar}$]}$}}
&   3,0.5  & 2 &   3,460 &    1,2012 &  0.3 \\
&   6,0.5  & 2 &  29,703 &   12,4748 &  2.5 \\
&  12,0.5  & 2 & 221,713 & 1,083,024 &  8.0 \\
&  18,0.5  & 2 & 467,497 & 2,396,784 & 16.2 \\
\hline \hline	
\multirow{4}{*}{\shortstack[c]{\emph{User-centric} \\ \emph{networks} \\
$\mbox{[$\mathit{td},K$]}$}} 
& 1,3 & 7 &  32,214 &   121,659 & 2.1 \\
& 1,4 & 7 & 104,897 &   433,764 & 6.1 \\
& 1,5 & 7 & 294,625 & 1,325,100 & 17.5 \\
& 1,6 & 7 & 714,849 & 3,465,558 & 42.8 \\
\hline\hline
\multirow{5}{*}{\shortstack[c]{\emph{Aloha (deadline)} \\
$\mbox{[$b_{\max}$,$D$]}$}}
& 1,8 & 3 &     3,519 &     5,839 & 0.2 \\
& 2,8 & 3 &    14,230 &    28,895 & 0.5 \\
& 3,8 & 3 &    72,566 &   181,438 & 2.1 \\
& 4,8 & 3 &   413,035 & 1,389,128 & 9.4 \\
& 5,8 & 3 & 2,237,981 & 9,561,201 & 58.4 \\
\hline\hline
\multirow{5}{*}{\shortstack[c]{\emph{Aloha} \\
$\mbox{[$b_{\max}$]}$}}
& 2 & 3 &     5,111 &     10,100 & 0.3 \\
& 3 & 3 &    22,812 &     56,693 & 0.8 \\
& 4 & 3 &   107,799 &    355,734 & 2.5 \\
& 5 & 3 &   556,168 &  2,401,113 & 13.4 \\
& 6 & 3 & 3,334,681 & 17,834,254 & 118.8 \\
\hline\hline
\multirow{4}{*}{\shortstack[c]{\emph{Intrusion detection} \\ \emph{system}    \\
$\mbox{[$\mathit{rounds}$]}$}}
& 25  & 2 &  75 &    483 & 0.07 \\
& 50  & 2 & 150 &    983 & 0.08 \\
& 100 & 2 & 300 &  1,983 & 0.1 \\
& 200 & 2 & 600 &  3,983 & 0.1 \\
\hline\hline
\multirow{4}{*}{\shortstack[c]{\emph{Jamming radio} \\ \emph{systems}    \\
$\mbox{[$\mathit{chans}$,$\mathit{slots}$]}$}}
& 4,6  & 2 &   531 &    45,004 & 0.5 \\
& 4,12 & 2 & 1,623 &   174,796 & 1.4 \\
& 6,6  & 2 & 1,061 &   318,392 & 2.0 \\
& 6,12 & 2 & 3,245 & 1,240,376 & 6.6 \\ 
\hline \hline
\multirow{4}{*}{\shortstack[c]{\emph{Medium access} \\ \emph{control} \\
$\mbox{[$e_{\max}$]}$}}
& 10 & 3 &  10,591 &   135,915 & 1.7 \\
& 15 & 3 &  33,886 &   457,680 & 5.0 \\
& 20 & 3 &  78,181 & 1,083,645 & 8.2 \\
& 25 & 3 & 150,226 & 2,115,060 & 14.0 \\
\hline \hline
\multirow{4}{*}{\shortstack[c]{\emph{Medium access} \\ \emph{control} \\
$\mbox{[$e_{\max}$,$s_{\max}$]}$}}
& 4,2 & 3 &  14,723 &   129,097 & 2.2 \\
& 4,4 & 3 &  18,751 &   147,441 & 4.3 \\
& 6,4 & 3 & 122,948 & 1,233,976 & 11.0 \\
& 6,6 & 3 & 138,916 & 1,315,860 & 12.5 \\
\hline \hline
\multirow{4}{*}{\shortstack[c]{\emph{Power control} \\
$\mbox{[$e_{\max}, pow_{\max}$]}$}}
& 40,8  & 2 &  32,812 &   260,924 & 1.7 \\
& 80,8  & 2 & 193,396 & 1,469,896 & 6.5 \\
& 40,16 & 2 &  34,590 &   291,766 & 1.6 \\
& 80,16 & 2 & 301,250 & 2,627,278 & 10.5 \\
\hline 
\end{tabular}}
\vspace*{-0.2cm}
\caption{Model statistics for the CSG case studies.}\label{tab:model-stats}
\vspace*{-0.4cm}
\end{table}\begin{table}[t]
\centering
{\scriptsize
\begin{tabular}{|c|r||c|r|r|r|} \hline
\multicolumn{1}{|c|}{Case study \& property} & 
\multicolumn{1}{c||}{$\!\!$Param.$\!\!$} & 
\multicolumn{1}{c|}{Actions} & 
\multicolumn{1}{c|}{Val.} & 
\multicolumn{2}{c|}{Verif. time (s)} \\ 
\cline{5-6}
\multicolumn{1}{|c|}{[parameters]} &
\multicolumn{1}{c||}{values} & 
\multicolumn{1}{c|}{max/avg} & 
\multicolumn{1}{c|}{iters} & 
\multicolumn{1}{c|}{Qual.} & \multicolumn{1}{c|}{Quant.} \\ \hline \hline	
\multirow{4}{*}{\shortstack[c]{\emph{Robot coordination} \\
$\coalition{\mathit{rbt}_1}\probop{\max=?}{\neg \mathsf{c} \, \buntilop \mathsf{g}_1}$ \\
$\mbox{[$l$,$k$]}$}}
&  4,4  & $\!\!\!$3,3/2.52,2.52$\!\!\!$ & 4  & 0.02 & 0.04 \\
&  8,8  & $\!\!\!$3,3/2.52,2.52$\!\!\!$ & 8  & 0.27 & 1.1 \\
& 12,12 & $\!\!\!$3,3/2.68,2.68$\!\!\!$ & 12 & 1.61 & 5.7 \\
& 16,16 & $\!\!\!$3,3/2.76,2.76$\!\!\!$ & 16 & 5.88 & 35.4 \\
& 24,24 & $\!\!\!$3,3/2.84,2.84$\!\!\!$ & 24 & 46.5 & 185.2 \\
\hline\hline
\multirow{5}{*}{\shortstack[c]{\emph{Robot coordination} \\
$\coalition{\mathit{rbt}_1}\rewop{}{\min=?}{\futureop \, \mathsf{g}_1}$ \\
$\mbox{[$l$]}$}}
&  4 & $\!\!\!$3,3/2.52,2.52$\!\!\!$ & 10;\;9  & 0.02 & 0.1 \\
&  8 & $\!\!\!$3,3/2.52,2.52$\!\!\!$ & 15;\;14 & 0.35 & 4.1 \\
& 12 & $\!\!\!$3,3/2.68,2.68$\!\!\!$ & 20;\;19 & 1.99 & 23.2 \\
& 16 & $\!\!\!$3,3/2.76,2.76$\!\!\!$ & 24;\;24 & 7.57 & 96.5 \\
& 24 & $\!\!\!$3,3/2.84,2.84$\!\!\!$ & 34;\;33 & 56.9 & 744.7 \\
\hline\hline
\multirow{5}{*}{\shortstack[c]{\emph{Future markets investors} \\ $\coalition{i_1} \rewop{}{\max=?}{\future \mathsf{c}_1}$ \\
$\mbox{[$\mathit{months}$]}$}}
&  6  &  $\!\!\!$8,2/1.54,1.21$\!\!\!$ & 14;\;14 &   1.3 &   6.0 \\
& 12  &  $\!\!\!$8,2/1.68,1.27$\!\!\!$ & 26;\;26 &  14.8 &  91.7 \\
& 24  &  $\!\!\!$8,2/1.73,1.29$\!\!\!$ & 50;\;50 &  94.3 & 616.0 \\
& 36  &  $\!\!\!$8,2/1.74,1.29$\!\!\!$ & 74;\;73 & 240.4 & 1,613 \\
& 48  &  $\!\!\!$8,2/1.74,1.29$\!\!\!$ & 98;\;97 & 395.4 & 2,770 \\
\hline\hline
\multirow{4}{*}{\shortstack[c]{\emph{User-centric networks} \\
$\coalition{\mathit{user}} \rewop{}{\min=?}{\future  \mathsf{f}}$ \\
$\mbox{[$\mathit{td},K$]}$}} 
& 1,3 & $\!\!$16,8/2.11,1.91$\!\!$ & 15;\;1 &  1.4 & 182.3 \\
& 1,4 & $\!\!$16,8/2.31,1.92$\!\!$ & 21;\;1 &  4.4 & 776.2 \\
& 1,5 & $\!\!$16,8/2.46,1.94$\!\!$ & 25;\;1 & 14.4 & 2,456 \\
& 1,6 & $\!\!$16,8/2.60,1.96$\!\!$ & 29;\;1 & 43.4 & 6,762 \\
\hline\hline
\multirow{4}{*}{\shortstack[c]{\emph{Aloha (deadline)} \\
$\!\!\coalition{\mathit{usr}_2,\mathit{usr}_3}\probop{\max=?}{\futureop \, \mathsf{s}_{1,2} {\wedge} t{\leq}D}\!\!$ \\
$\mbox{[$b_{\max}$,$D$}$]}}
& 2,8 & 4,2/1.01,1.00 & 24 & 0.5  & 0.8 \\
& 3,8 & 4,2/1.01,1.00 & 23 & 1.8  & 1.8 \\
& 4,8 & 4,2/1.01,1.00 & 23 & 6.5  & 4.5 \\
& 5,8 & 4,2/1.01,1.00 & 23 & 35.9 & 9.4 \\
\hline\hline
\multirow{5}{*}{\shortstack[c]{\emph{Aloha} \\
$\coalition{\mathit{usr}_2,\mathit{usr}_3}\rewop{}{\min=?}{\futureop \, \mathsf{sent}_{2,3}}$ \\
$\mbox{[$b_{\max}$}$]}}
& 2 & 4,2/1.01,1.00 &   58;\;47 &  0.2  &   2.1 \\
& 3 & 4,2/1.01,1.00 &   71;\;57 &  0.6  &   6.0 \\
& 4 & 4,2/1.01,1.00 &  109;\;86 &  3.7  &  33.7 \\
& 5 & 4,2/1.01,1.00 & 193;\;150 & 26.6  & 317.6 \\
& 6 & 4,2/1.01,1.00 & 362;\;279 & 314.5 & 3,836 \\
\hline\hline		 			%
\multirow{4}{*}{\shortstack[c]{\emph{Intrusion detection system} \\
$\;\coalition{\mathit{policy}} \rewop{}{\min=?}{\futureop \, r{=}\mathit{rounds}}\;$ \\
$\mbox{[$\mathit{rounds}$]}$}}
& 25  & 2,2/1.96,1.96 &   26;\;26 & 0.03 & 0.2 \\
& 50  & 2,2/1.98,1.98 &   51;\;51 & 0.06 & 0.7 \\
& 100 & 2,2/1.99,1.99 & 101;\;101 &  0.2 & 2.5 \\
& 200 & 2,2/2.00,2.00 & 201;\;201 &  0.5 & 7.0 \\
\hline\hline
\multirow{4}{*}{\shortstack[c]{\emph{Jamming radio systems} \\
$\coalition{\mathit{user}} \probop{\max=?}{\future \mathit{sent}{\geq}\mathit{slots}{/}2}$ \\
$\mbox{[$\mathit{slots}$]}$}}
& 4,6  & 3,3/2.17,2.17 &  7 & 0.03 & 0.1 \\
& 4,12 & 3,3/2.49,2.49 & 13 &  0.2 & 0.5 \\
& 6,6  & 4,4/2.76,2.76 &  7 &  0.1 & 0.3 \\
& 6,12 & 4,4/3.24,3.24 & 13 &  0.5 & 2.0 \\ \hline
\end{tabular}}
\vspace*{-0.2cm}
\caption{Statistics for CSG zero-sum verification instances.}\label{tab:zerosum-results}
\vspace*{-0.4cm}
\end{table}\begin{table}[t]
\centering
{\scriptsize
\begin{tabular}{|c|r||c|r|r|r|} \hline
\multicolumn{1}{|c|}{Case study \& property} & 
\multicolumn{1}{c||}{$\!\!\!$Param.$\!\!\!$} & 
\multicolumn{1}{c|}{Actions} & 
\multicolumn{1}{c|}{Val.} & 
\multicolumn{2}{c|}{Verif. time (s)} \\ 
 \cline{5-6}
\multicolumn{1}{|c|}{[parameters]} &
\multicolumn{1}{c||}{values} & 
\multicolumn{1}{c|}{max/avg} & 
\multicolumn{1}{c|}{$\!\!\!\!$iters.$\!\!\!\!$} & 
\multicolumn{1}{c|}{MDP} & \multicolumn{1}{c|}{CSG} \\ \hline \hline
\multirow{5}{*}{\shortstack[c]{\emph{Robot coordination} \\
$\!\!\!\nashop{\mathit{rbt}_1{:}\mathit{rbt}_2}{\max=?}{\probop{}{\neg \mathsf{c} \, \buntilop \mathsf{g}_1}{+}\probop{}{\neg \mathsf{c} \, \buntilop \mathsf{g}_2}}\!\!\!$ \\
$\mbox{[$l$,$k$]}$}}
&  4,4  & $\!\!\!$3,3/2.07,2.07$\!\!\!$ &  4 &  0.02 & 0.1 \\
&  8,8  & $\!\!\!$3,3/2.52,2.52$\!\!\!$ &  8 &  0.45 & 1.0 \\
& 12,12 & $\!\!\!$3,3/2.68,2.68$\!\!\!$ & 12 &  6.25 & 5.8 \\
& 16,16 & $\!\!\!$3,3/2.76,2.76$\!\!\!$ & 16 &  34.7 & 22.2 \\
& 24,24 & $\!\!\!$3,3/2.84,2.84$\!\!\!$ & 24 & 375.2 & 1,365 \\
\hline\hline
\multirow{4}{*}{\shortstack[c]{\emph{Robot coordination} \\
$\nashop{\mathit{rbt}_1{:}\mathit{rbt}_2}{\max=?}{\probop{}{\neg \mathsf{c} \, \buntilop \mathsf{g}_1}{+}\probop{}{\neg \mathsf{c} \, \untilop \, \mathsf{g}_2}}$ \\
$\mbox{[$l$,$k$]}$}}
& 4,8  & $\!\!\!$3,3/2.10,2.04$\!\!\!$ & 21 & 0.16 &  3.3 \\
& 4,16 & $\!\!\!$3,3/2.12,2.05$\!\!\!$ & 21 & 0.28 &  12.0 \\
& 8,8  & $\!\!\!$3,3/2.53,2.51$\!\!\!$ & 34 & 5.51 &  77.2 \\
& 8,16 & $\!\!\!$3,3/2.54,2.52$\!\!\!$ & 35 & 17.2 & 3,513 \\
\hline\hline
\multirow{4}{*}{\shortstack[c]{\emph{Robot coordination} \\
$\nashop{\mathit{rbt}_1{:}\mathit{rbt}_2}{\min=?}{\rewop{}{}{\futureop \, \mathsf{g}_1}{+}\rewop{}{}{\futureop \, \mathsf{g}_2}}$ \\
$\mbox{[$l$]}$}}
& 4  & $\!\!\!$3,3/2.07,2.07$\!\!\!$ &  8 & 0.05 & 0.2 \\
& 8  & $\!\!\!$3,3/2.52,2.52$\!\!\!$ & 15 & 0.44 & 10.3 \\
& 12 & $\!\!\!$3,3/2.68,2.68$\!\!\!$ & 23 & 2.12 & 227.7 \\
& 16 & $\!\!\!$3,3/2.84,2.84$\!\!\!$ & 28 & 12.1 & 3,272 \\
\hline\hline
\multirow{4}{*}{\shortstack[c]{\emph{Future markets investors} \\ $\nashop{i_1{:}i_2}{\max=?}{\rewop{}{}{\futureop \, \mathsf{c}_1}{+}\rewop{}{}{\futureop \, \mathsf{c}_2}}$ \\
$\mbox{[$\mathit{months}$]}$}}
& 3  & $\!\!\!$2,2/1.15,1.15$\!\!\!$ &  6 &  0.1 &   0.2  \\
& 6  & $\!\!\!$2,2/1.22,1.22$\!\!\!$ & 13 &  0.8 &   3.1 \\
& 12 & $\!\!\!$2,2/1.27,1.27$\!\!\!$ & 25 &  9.1 & 284.7 \\
& 18 & $\!\!\!$2,2/1.29,1.29$\!\!\!$ & 37 & 31.6 & 4,326 \\
\hline \hline	
\multirow{5}{*}{\shortstack[c]{\emph{Aloha (deadline)} \\
$\!\!\!\nashop{\mathit{usr}_1{:}\mathit{usr}_2{,}\mathit{usr}_3}{\max=?}{\probop{}{\futureop \, \mathsf{s}_1}{+}\probop{}{\futureop \, \mathsf{s}_{2,3}}}\!\!\!$ \\
$\mbox{[$b_{\max}$,$D$]}$}}
& 1,8 & $\!\!\!$2,4/1.00,1.01$\!\!\!$ & 23 &  0.1 &   0.3 \\
& 2,8 & $\!\!\!$2,4/1.00,1.00$\!\!\!$ & 23 &  0.4 &   1.2 \\
& 3,8 & $\!\!\!$2,4/1.00,1.00$\!\!\!$ & 22 &  2.0 &   3.4 \\
& 4,8 & $\!\!\!$2,4/1.00,1.00$\!\!\!$ & 22 &  7.1 &  18.5 \\
& 5,8 & $\!\!\!$2,4/1.00,1.00$\!\!\!$ & 22 & 39.9 & 103.0 \\
\hline\hline
\multirow{5}{*}{\shortstack[c]{\emph{Aloha} \\
$\nashop{\mathit{usr}_1{:}\mathit{usr}_2{,}\mathit{usr}_3}{\min=?}{\rewop{}{}{\futureop \, \mathsf{s}_1}{+}\rewop{}{}{\futureop \, \mathsf{s}_{2,3}}}$ \\
$\mbox{[$b_{\max}$]}$}}
& 2 & $\!\!\!$2,4/1.00,1.01$\!\!\!$ &  54 &   0.2  &   0.8 \\
& 3 & $\!\!\!$2,4/1.00,1.00$\!\!\!$ &  62 &   0.8  &   3.1 \\
& 4 & $\!\!\!$2,4/1.00,1.00$\!\!\!$ &  88 &   4.5  &  40.1 \\
& 5 & $\!\!\!$2,4/1.00,1.00$\!\!\!$ & 145 &  35.3  & 187.7 \\
& 6 & $\!\!\!$2,4/1.00,1.00$\!\!\!$ & 256 & 453.7  & 2,396 \\
\hline\hline
\multirow{4}{*}{\shortstack[c]{\emph{Medium access control} \\ $\nashop{p_1{:}p_2{,}p_3}{\max=?}{\rewop{}{}{\scumul{\leq k}}{+}\rewop{}{}{\scumul{\leq k}}}$ \\
$\mbox{[$e_{\max}$,$\mathit{k}$]}$}}
& 10,25 & $\!\!\!$2,4/1.91,3.63$\!\!\!$ & 25 & 0.0 &  65.8 \\
& 15,25 & $\!\!\!$2,4/1.94,3.75$\!\!\!$ & 25 & 0.0 & 193.3 \\
& 20,25 & $\!\!\!$2,4/1.95,3.81$\!\!\!$ & 25 & 0.0 & 394.8 \\
& 25,25 & $\!\!\!$2,4/1.96,3.85$\!\!\!$ & 25 & 0.0 & 565.4 \\
\hline \hline
\multirow{4}{*}{\shortstack[c]{\emph{Medium access control} \\ $\nashop{p_1{:}p_2{,}p_3}{\max=?}{\probop{}{\futureop \, \mathsf{m}_1}{+}\probop{}{\futureop \, \mathsf{m}_{2,3}}}$ \\
$\mbox{[$e_{\max}$,$s_{\max}$]}$}}
& 4,2 & $\!\!\!$2,4/1.70,2.88$\!\!\!$ & 10 & 0.61 & 96.2 \\
& 4,4 & $\!\!\!$2,4/1.64,2.70$\!\!\!$ & 12 & 0.61 & 22.8 \\
& 6,4 & $\!\!\!$2,4/1.77,3.12$\!\!\!$ & 17 & 7.40 & 1,639 \\
& 6,6 & $\!\!\!$2,4/1.74,3.02$\!\!\!$ & 18 & 4.09 &  93.6 \\
\hline \hline
\multirow{4}{*}{\shortstack[c]{\emph{Medium access control} \\ $\nashop{p_1{:}p_2{,}p_3}{\max=?}{\probop{}{\bfutureop \, \mathsf{m}_1}{+}\probop{}{\futureop \, \mathsf{m}_{2,3}}}$ \\
$\mbox{[$e_{\max}$,$s_{\max}$,$k$]}$}}
& 4,4,4  & $\!\!\!$2,4/1.67,2.70$\!\!\!$ & 12 & 1.64 &  39.1 \\
& 4,4,8  & $\!\!\!$2,4/1.68,2.70$\!\!\!$ & 12 & 3.30 & 106.3 \\
& 6,4,6  & $\!\!\!$2,4/1.76,3.02$\!\!\!$ & 18 & 23.4 & 341.9 \\
& 6,4,12 & $\!\!\!$2,4/1.74,3.02$\!\!\!$ & 18 & 60.0 & 961.8 \\
\hline \hline
\multirow{4}{*}{\shortstack[c]{\emph{Power control} \\
$\nashop{p_1{:}p_2}{\max=?}{\rewop{}{}{\futureop \, e_1{=}0}{+}\rewop{}{}{\futureop \, e_2{=}0}}$ \\
$\mbox{[$e_{\max}$,$\mathit{pow}_{\max}$]}$}}
& 40,8  & $\!\!\!$2,2/1.91,1.91$\!\!\!$ & 20 &  4.1 & 11.7 \\
& 80,8  & $\!\!\!$2,2/1.88,1.88$\!\!\!$ & 40 & 34.8 & 130.3 \\
& 40,16 & $\!\!\!$2,2/1.95,1.95$\!\!\!$ & 40 &  5.4 &  11.6 \\
& 80,16 & $\!\!\!$2,2/1.98,1.98$\!\!\!$ & 80 & 64.4 & 211.4 \\
\hline \hline
\multirow{4}{*}{\shortstack[c]{\emph{Power control} \\
$\nashop{p_1{:}p_2}{\max=?}{\rewop{}{}{\rewop{}{}{\futureop \, e_1{=}0}}{+}\rewop{}{}{\scumul{\leq k}}}$ \\
$\mbox{[$e_{\max}$,$\mathit{pow}_{\max}$,$\mathit{k}$]}$}}
& 40,4,20 & $\!\!\!$2,2/1.69,1.69$\!\!\!$ & 20 &  13.3 &  27.5 \\
& 80,4,20 & $\!\!\!$2,2/1.69,1.69$\!\!\!$ & 20 &  83.9 & 134.2 \\
& 40,8,20 & $\!\!\!$2,2/1.91,1.91$\!\!\!$ & 20 &  49.1 &  84.6 \\
& 80,8,20 & $\!\!\!$2,2/1.88,1.88$\!\!\!$ & 20 & 498.6 & 846.8 \\
\hline
\end{tabular}}
\vspace*{-0.2cm}
\caption{Statistics for CSG nonzero-sum verification instances.}\label{tab:nonzerosum-results}
\vspace*{-0.4cm}
\end{table}\begin{table}[t]
\centering
{\scriptsize
\begin{tabular}{|c|r||r|r|r|r|} \hline
\multicolumn{1}{|c|}{Case study} & 
\multicolumn{1}{c||}{Param.} & 
\multicolumn{1}{c|}{Players} & \multicolumn{1}{c|}{States} &  \multicolumn{1}{c|}{Transitions}  \\ \multicolumn{1}{|c|}{[parameters]} &
\multicolumn{1}{c||}{values} & & & 
 \\ \hline \hline	
\multirow{4}{*}{\shortstack[c]{\emph{Robot coordination} \\
$\nashop{\mathit{rbt}_1{:}\mathit{rbt}_2}{\max=?}{\probop{}{\neg \mathsf{c} \, \buntilop \mathsf{g}_1}{+}\probop{}{\neg \mathsf{c} \, \untilop \, \mathsf{g}_2}}$ \\
$\mbox{[$l$,$k$]}$}}
& 4,8  & 2 &  1,923 &    56,385 \\
& 4,16 & 2 &  3,491 &   104,545 \\
& 8,8  & 2 & 36,773 & 1,860,691 \\
& 8,16 & 2 & 67,525 & 3,443,443 \\
\hline \hline
\multirow{4}{*}{\shortstack[c]{\emph{Medium access control} \\ $\nashop{p_1{:}p_2{,}p_3}{\max=?}{\probop{}{\bfutureop \, \mathsf{m}_1}{+}\probop{}{\futureop \, \mathsf{m}_{2,3}}}$ \\
$\mbox{[$e_\mathit{max}$,$s_\mathit{max}$,$k$]}$}}
& 4,4,4  & 3 &    89,405 &    718,119 \\
& 4,4,8  & 3 &   158,609 &  1,282,435 \\
& 6,4,6  & 3 &   944,727 &  9,071,885 \\
& 6,4,12 & 3 & 1,745,001 & 16,800,083 \\
\hline \hline
\multirow{4}{*}{\shortstack[c]{\emph{Power control} \\
$\nashop{p_1{:}p_2}{\max=?}{\rewop{}{}{\futureop \, e_1{=}0}{+}\rewop{}{}{\scumul{\leq k}}}$ \\
$\mbox{[$e_\mathit{max}$,$\mathit{pow}_\mathit{max}$,$\mathit{k}$]}$}}
& 40,4,20 & 2 &   182,772 &  1,040,940 \\
& 80,4,20 & 2 &   917,988 &  5,153,700 \\
& 40,8,20 & 2 &   524,473 &  4,179,617 \\
& 80,8,20 & 2 & 3,806,240 & 28,950,948 \\
\hline
\end{tabular}}
\vspace*{-0.2cm}
\caption{Model statistics for CSGs built verifying mixed nonzero-sum properties.}\label{tab:mixed-stats}
\vspace*{-0.0cm}
\end{table}\begin{table}[t]
\centering
{\scriptsize
\begin{tabular}{|c|r||r|r|r|r|} \hline
\multirow{3}{*}{\shortstack[c]{Case study \\ $\mbox{[parameters]}$}} & 
\multirow{3}{*}{\shortstack[c]{$\!\!$Param.$\!\!$ \\ $\!\!$values$\!\!$}} & 
\multicolumn{2}{c|}{Inner Formula} &
\multicolumn{2}{c|}{Outer Formula}
\\  \cline{3-6}
&
& 
\multicolumn{1}{c|}{Val.} & 
\multicolumn{1}{c|}{Verif.} &
\multicolumn{1}{c|}{Val.} & 
\multicolumn{1}{c|}{Verif.}
\\ 
&
& 
\multicolumn{1}{c|}{$\!\!\!\!$iters.$\!\!\!\!$} & 
\multicolumn{1}{c|}{$\!\!\!\!$time (s)$\!\!\!\!$} & 
\multicolumn{1}{c|}{$\!\!\!\!$iters.$\!\!\!\!$} & 
\multicolumn{1}{c|}{$\!\!\!\!$time (s)$\!\!\!\!$}  \\ \hline \hline
\multirow{5}{*}{\shortstack[c]{\emph{Robot coordination} \\
$\coalition{\mathit{rbt}_1}\probop{\max=?}{ \future \phi}$ \\
$\phi{=}\coalition{\mathit{rbt}_2}\rewop{}{\geq 10}{\futureop \, \mathsf{g}_2}$
$\mbox{[$l$]}$}}
&  4 & 12;\;12 &   0.1 &  5 &  0.1 \\
&  8 & 22;\;22 &   5.0 & 10 &  0.8 \\
& 12 & 31;\;30 &  40.5 & 10 &  2.6 \\
& 16 & 40;\;40 & 187.7 & 11 &  6.6 \\
& 24 & 58;\;58 & 1,235 & 12 & 22.6 \\  \hline\hline
\multirow{4}{*}{\shortstack[c]{\emph{Robot coordination} \\
$\!\!\!\!$$\coalition{\mathit{rbt}_1,\mathit{rbt}_2}\probop{\min=?}{ \future \phi}$$\!\!\!\!$ \\
$\phi{=}\nashop{\mathit{rbt}_1{:}\mathit{rbt}_2}{\min \leq 5}{\rewop{}{}{\futureop \, \mathsf{g}_1}{+}\rewop{}{}{\futureop \, \mathsf{g}_2}}$ $\mbox{[$l$]}$}}
&  4 &  8 &   0.1 & 24 &  0.1 \\
&  8 & 17 &  20.8 & 18 &  1.4 \\
& 12 & 32 & 527.3 & 28 &  8.4 \\
& 16 & 39 & 4,664 & 37 & 24.9 \\  \hline\hline
\multirow{5}{*}{\shortstack[c]{\emph{Aloha} \\
$\coalition{\mathit{usr}_1,\mathit{usr}_2,\mathit{usr}_3}\probop{\max=?}{ \future \phi}$ \\
$\phi{=}\nashop{\mathit{usr}_1{:}\mathit{usr}_2,\mathit{usr}_3}{\min \geq 2}{\probop{}{\futureop \, \mathsf{s}_1}{+}\probop{}{\futureop \, \mathsf{s}_{2,3}}}$ \\
$\mbox{[$b_{\max}$,$D$]}$}}
& 1,8  & 23 &   0.6 & 24 &  0.3  \\
& 2,8  & 23 &   1.4 & 23 &  0.8  \\
& 3,8  & 22 &   7.9 & 23 &  2.8  \\
& 4,8  & 22 &  39.8 & 23 &  9.3 \\
& 5,8  & 22 & 172.6 & 23 & 37.1   \\  \hline
\end{tabular}}
\vspace*{-0.2cm}
\caption{Statistics for verification of nested properties for CSGs.}\label{tab:nest-stats}
\vspace*{-0.2cm}
\end{table}

\subsection{Efficiency and Scalability}\label{expr-sect}

We begin by presenting a selection of results illustrating the performance of our implementation. The experiments were run on a 2.10 GHz Intel Xeon with 16GB of JVM memory. In \tabref{tab:model-stats}, we present the model statistics for the examples used: the number of players, states, transitions and model construction times (details of the case studies themselves follow in the next section). Due to improvements in the modelling language and the model building procedure, some of the model statistics differ from those presented in~\cite{KNPS18,KNPS19}. %
The main reason is that the earlier version of the implementation did not allow for variables of different players to be updated following a joint probability distribution, which made it necessary to introduce intermediate states in order to specify some of the behaviour. Also, some model statistics differ from~\cite{KNPS18}
since models were modified to meet \assumassumref{game2-assum}{game3-assum} to enable %
the analysis of nonzero-sum properties. 

\tabtabref{tab:zerosum-results}{tab:nonzerosum-results} present the model checking statistics when analysing zero-sum and nonzero-sum properties, respectively. In both tables, this includes the maximum and average number of actions of each coalition in the matrix/bimatrix games solved at each step of value iteration and the number of iterations performed. In the case of zero-sum properties including reward formulae of the form $\future \phi$, value iteration is performed twice (see \sectref{inf-zero-sect}), and therefore the number of iterations for each stage are presented (and separated by a semi-colon). For zero-sum properties, the timing statistics are divided into the time for qualitative (column `Qual.') and quantitative verification, which includes solving matrix games (column `Quant.'). For nonzero-sum properties we divide the timing statistics into the time for CSG verification, which includes solving bimatrix games (column `CSG'), and the instances of MDP verification (column `MDP'). 
In the case of mixed nonzero-sum properties, i.e., properties including both finite and infinite horizon objectives, we must first build a new game (see \sectref{nonzero-mixed-sect}); the statistics for these CSGs (number of players, states and transitions) are presented in \tabref{tab:mixed-stats}.
Finally, \tabref{tab:nest-stats} presents the timing results for three nested properties. Here we give the time required for verifying the inner and outer formula separately, as well as the number of iterations for value iteration at each stage.

Our results demonstrate significant gains in efficiency with respect to 
those presented for zero-sum properties in~\cite{KNPS18}
and nonzero-sum properties in~\cite{KNPS19}
(for the latter, a direct comparison with the published results
is possible since it uses an identical experimental setup).
The gains are primarily due to faster SMT solving
and reductions in CSG size as a result of modelling improvements, and specifically the removal of intermediate states as discussed above.

The implementation can analyse models with over 3 million states and almost 18 million transitions; all are solved in under 2 hours and most are considerably quicker. The majority of the time is spent solving matrix or bimatrix games, so performance is affected by the number of choices available within each coalition, rather than the number of players, as well as the number of states.  %
For example, larger instances of the Aloha models are verified relatively quickly since the coalitions have only one choice in many states (the average number of choices is 1.00 for both coalitions).  However, for models where players have choices in almost all states, only models with up to hundreds of thousands of states for zero-sum properties and tens of thousands of states for nonzero-sum properties can be verified within 2 hours. 

\subsection{Case Studies}\label{case-sect}

Next, we present more information about our case studies,
to illustrate the applicability and benefits of our techniques.
We use some of these examples to illustrate the benefits of concurrent stochastic games,
in contrast to their turn-based counterpart; here, we build both TSG and CSG models
for the case study and compare the results.

To study the benefits of nonzero-sum properties,
we compare the results with corresponding zero-sum properties.
For example, for a nonzero-sum formula of the form $\nashop{C{:}C'}{\max=?}{\probop{}{\future \phi_1}{+}\probop{}{\future \phi_2}}$, we compute the value and an optimal strategy $\sigma_C^\star$ for coalition $C$ of the formula $\coalition{C}\probop{\max=?}{\future \phi_1}$, and then find the value of an optimal strategy for the coalition $C'$ for $\probop{\min=?}{\future \phi_2}$ and $\probop{\max=?}{\future \phi_2}$ in the MDP induced by CSG when $C$ follows $\sigma_C^\star$. 
The aim is to showcase the advantages of cooperation since, in many real-world applications, agents’ goals are not strictly opposed and adopting a strategy that assumes antagonistic behaviour can have a negative impact from both individual and collective standpoints.

As will be seen, our results demonstrate that, by using nonzero-sum properties,
at least one of the players gains and in almost all cases neither player loses 
(in the one case study where this is not the case, the gains far outweigh the losses).
The individual SWNE/SCNE values for players need not be unique and,
for all case studies (except Aloha and medium access in which the players are not symmetric),
the values can be swapped to give alternative SWNE/SCNE values.

Finally, we note that,
for infinite-horizon nonzero-sum properties, we compute the value of $\varepsilon$
for the synthesised $\varepsilon$-NE and find that $\varepsilon=0$ in all cases.

\newcommand{\xmin}{0}
\newcommand{\xmax}{3}
\newcommand{\ymin}{0}
\newcommand{\ymax}{3}
\newcommand{\cro}{red}
\newcommand{\crt}{cyan}
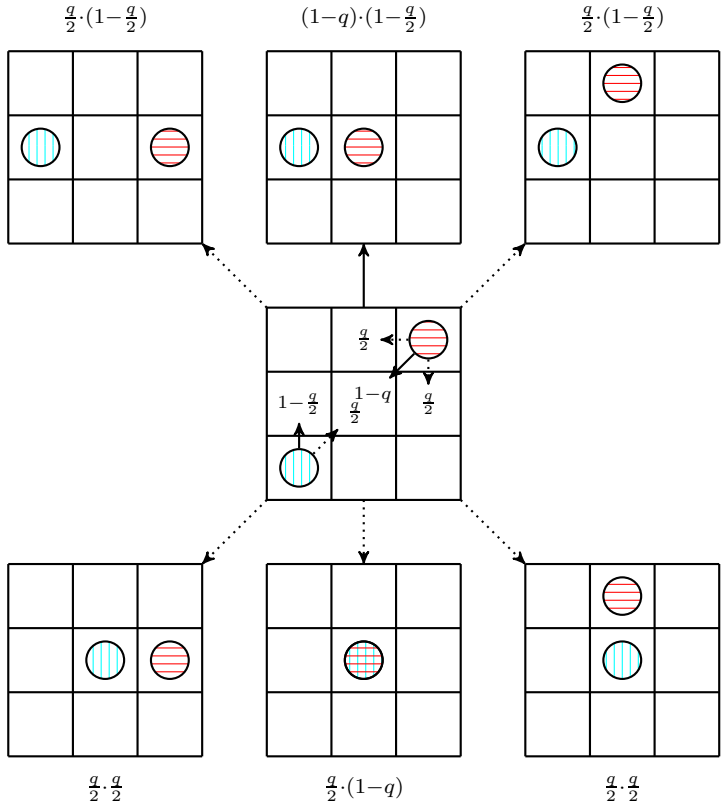
\begin{figure}[t]
\centering
\begin{tikzpicture}[scale=0.85,>=stealth',thick]
    \foreach \r in {0,4,8} {
        \foreach \c in {0,4,8} {
            \ifthenelse{\NOT\(\r=4\AND\c=0\) \AND \NOT\(\r=4\AND\c=8\)} {
                \foreach \x in {\xmin,...,\xmax} {
                    \draw (\c+\x,\r+\ymin) -- (\c+\x,\r+\ymax);
                }
                \foreach \y in {\ymin,...,\ymax} {
                    \draw (\c+\xmin,\r+\y) -- (\c+\xmax,\r+\y);
                }
            }
        }
    }

\node[pattern=vertical lines, circle, pattern color=\crt, draw, radius=0.25, inner sep=5pt] at (0.5,9.5) (r001) {};
\node[pattern=horizontal lines, circle, pattern color=\cro, draw, radius=0.25, inner sep=5pt] at (2.5,9.5) (r002) {};
\node[pattern=vertical lines, circle, pattern color=\crt, draw, radius=0.25, inner sep=5pt] at (4.5,9.5) (r041) {};
\node[pattern=horizontal lines, circle, pattern color=\cro, draw, radius=0.25, inner sep=5pt] at (5.5,9.5) (r042) {};
\node[pattern=vertical lines, circle, pattern color=\crt, draw, radius=0.25, inner sep=5pt] at (8.5,9.5) (r081) {};
\node[pattern=horizontal lines, circle, pattern color=\cro, draw, radius=0.25, inner sep=5pt] at (9.5,10.5) (r082) {};
    
\node[pattern=vertical lines, circle, pattern color=\crt, draw, radius=0.25, inner sep=5pt] at (4.5,4.5) (r441) {};
\node[pattern=horizontal lines, circle, pattern color=\cro, draw, radius=0.25, inner sep=5pt] at (6.5,6.5) (r442) {};
\node[] at (4.5,5.5) (mv1) {\scriptsize $1{-}\frac{q}{2}$};
\node[] at (5.36,5.36) (mv2) {\scriptsize $\frac{q}{2}$};
\node[] at (5.64,5.64) (mv3) {\scriptsize $1{-}q$};
\node[] at (5.5,6.5) (mv4) {\scriptsize $\frac{q}{2}$};
\node[] at (6.5,5.5) (mv5) {\scriptsize $\frac{q}{2}$};
\draw[->] (r441) -- (mv1);
\draw[dotted,->] (r441) -- (mv2);
\draw[->,] (r442) -- (mv3);
\draw[dotted,->] (r442) -- (mv4);
\draw[dotted,->] (r442) -- (mv5);

\node[pattern=vertical lines, circle, pattern color=\crt, draw, radius=0.25, inner sep=5pt] at (1.5,1.5) (r801) {};
\node[pattern=horizontal lines, circle, pattern color=\cro, draw, radius=0.25, inner sep=5pt] at (2.5,1.5) (r802) {};
\node[pattern=vertical lines, circle, pattern color=\crt, draw, radius=0.25, inner sep=5pt] at (5.5,1.5) (r841) {};
\node[pattern=horizontal lines, circle, pattern color=\cro, draw, radius=0.25, inner sep=5pt] at (5.5,1.5) (r842) {};
\node[pattern=vertical lines, circle, pattern color=\crt, draw, radius=0.25, inner sep=5pt] at (9.5,1.5) (r881) {};
\node[pattern=horizontal lines, circle, pattern color=\cro, draw, radius=0.25, inner sep=5pt] at (9.5,2.5) (r882) {};

\draw[dotted, ->] (4,7) -- (3,8);
\draw[solid,->] (5.5,7) -- (5.5,8);
\draw[dotted,->] (7,7) -- (8,8);

\node[] at (1.5,11.5) {$\frac{q}{2}{\cdot}(1{-}\frac{q}{2})$};
\node[] at (5.5,11.5) {$(1{-}q){\cdot}(1{-}\frac{q}{2})$};
\node[] at (9.5,11.5) {$\frac{q}{2}{\cdot}(1{-}\frac{q}{2})$};

\draw[dotted,->] (4,4) -- (3,3);
\draw[dotted,->] (5.5,4) -- (5.5,3);
\draw[dotted,->] (7,4) -- (8,3);

\node[] at (1.5,-0.5) {$\frac{q}{2}{\cdot}\frac{q}{2}$};
\node[] at (5.5,-0.5) {$\frac{q}{2}{\cdot}(1{-}q)$};
\node[] at (9.5,-0.5) {$\frac{q}{2}{\cdot}\frac{q}{2}$};

\end{tikzpicture}
\caption{Robot coordination on a $3{\times}3$ grid: probabilistic choices for one pair of action choices in the initial state.
Solid lines indicate movement in the intended direction, dotted lines where there is deviation due to obstacles.}\label{gridrobot-fig}
\vspace*{-0.4cm}
\end{figure}
\startpara{Robot Coordination} Our first case study concerns a scenario in which two robots move concurrently over a grid of size $l{\times}l$, briefly discussed in \egref{logic-eg}. The robots start in diagonally opposite corners and try to reach the corner from which the other starts. A robot can move either diagonally, horizontally or vertically towards its goal. Obstacles which hinder the robots as they move from location to location
are modelled stochastically according to a parameter $q$ (which we set to $0.25$): when a robot moves, there is a probability that it instead moves in an adjacent direction, e.g., if it tries to moves north west, then with probability $q/2$ it will instead move north and with the same probability west. 

We can model this scenario as a two-player CSG, where the players correspond to the robots ($\mathit{rbt}_1$ and $\mathit{rbt}_2$), the states of the game represent their positions on the grid. In states where a robot has not reached its goal, it can choose between actions that move either diagonally, horizontally or vertically towards its goal (under the restriction that it remains in the grid after this move). For $i \in \{1,2\}$, we let $\mathsf{goal}_i$ be the atomic proposition labelling those states of the game in which $\mathit{rbt}_i$ has reached its goal and $\mathsf{crash}$ the atomic proposition labelling the states in which the robots have crashed, i.e., are in the same grid location.
In \figref{gridrobot-fig}, we present the states that can be reached from the initial state of the game when $l=3$, when the robot in the south west corner tries to move north and the robot in the north east corner tries to move south west. As can be seen there are six different outcomes and the probability of the robots crashing is $\frac{q}{2}{\cdot}(1{-}q)$.

We first investigate the probability of the robots eventually reaching their goals without crashing for different size grids. In the zero-sum case, we find the values for the formula $\coalition{\mathit{rbt}_1}\probop{\max=?}{\neg \mathsf{crash} \until \mathsf{goal}_1}$ converge to $1$ as $l$ increases; for example, the values for this formula in the initial states of game when $l=5$, $10$ and $20$ are approximately $0.9116$, $0.9392$ and $0.9581$, respectively. On the other hand, in the nonzero-sum case, considering SWNE values for the formula $\nashop{\mathit{rbt}_1{:}\mathit{rbt}_2}{\max =?}{\probop{}{\neg \mathsf{crash} \until \mathsf{goal}_1}{+}\probop{}{\neg \mathsf{crash} \until \mathsf{goal}_2}}$ and $l \geq 4$, we find that each robot can reach its goal with probability 1 (since time is not an issue, they can collaborate to avoid crashing).

We next consider the probability of the robots reaching their targets without crashing within a bounded number of steps. \figref{robot-deadline-fig} presents
both the value for the (zero-sum) formula $\coalition{\mathit{rbt}_1}\probop{\max=?}{\neg \mathsf{crash} \buntil \mathsf{goal}_1}$ and SWNE values for the formula $\nashop{\mathit{rbt}_1{:}\mathit{rbt}_2}{\max \geq 2}{\probop{}{\neg \mathsf{crash} \buntil \mathsf{goal}_1}{+}\probop{}{\neg \mathsf{crash} \buntil \mathsf{goal}_2}}$, for a range of step bounds and grid sizes.
When there is only one route to each goal within the bound (along the diagonal), i.e., when $k=l{-}1$, in the SWNE profile both robots take this route. In odd grids, there is a high chance of crashing, but also a chance one will deviate and the other reaches its goal. Initially, as the bound $k$ increases, for odd grids the SWNE values for the robots are not equal (see \figref{robot-deadline-fig} right). Here, both robots following the diagonal does not yield a NE profile. First, the chance of crashing is high, and therefore the probability of the robots satisfying their objectives is low. Therefore it is advantageous for a robot to switch to a longer route as this will increase the probability of satisfying its objective, even taking into account that there is a greater chance it will run out of steps and changing its route will increase the probability of the other robot satisfying its objective by a greater amount (as the other robot will still be following the diagonal). Dually, both robots taking a longer route is not an NE profile, since if one robot switches to the diagonal route, then the probability of satisfying its objective will increase. It follows that, in a SWNE profile, one robot has to follow the diagonal and the other take a longer route.
As expected, if we compare the results, we see that the robots can improve their chances of reaching their goals by collaborating.

The next properties we consider concern the minimum expected number of steps for the robots to reach their goal. In \figref{robot-expected-fig} we have plotted the values corresponding to the formula $\coalition{\mathit{rbt}_2}\rewop{r_\mathit{steps}}{\min=?}{\future \mathsf{goal}_2}$ and SCNE values for the individual players for $\nashop{\mathit{rbt}_1{:}\mathit{rbt}_2}{\min=?}{\rewop{r_\mathit{steps}}{}{\future \mathsf{goal}_1}{+}\rewop{r_\mathit{steps}}{}{\future \mathsf{goal}_2}}$ as the grid size $l$ varies. The results again demonstrate that the players can gain by collaborating.

\begin{figure}[t]
\centering
\begin{subfigure}{.49\textwidth}
\vspace*{0.2cm}
\centering
\tiny{
\begin{tikzpicture}
\begin{axis}[
    title style={yshift=-2ex},
    title={$\;$},
    ylabel={Probability},
    xlabel={$k$},
    xmin=0, xmax=10,
    xtick={0,1,2,3,4,5,6,7,8,9,10},
    ymin=0, ymax=1,
    ytick={0,0.25,0.5,0.75,1},
    legend pos=south east,
    legend style={fill=none, draw=none, 
                    at={(1,0)}, anchor=south east, 
                    nodes={scale=0.9, transform shape}},
    ymajorgrids=true,
    grid style=dashed,
    height=4.25cm,
    width=0.95\textwidth,
    legend entries={
                $l{=}3$,
                $l{=}4$,
                $l{=}5$,
                $l{=}6$,      
                $l{=}7$                
                }]
\addlegendimage{mark=square*,red,mark size=1.5pt}
\addlegendimage{mark=o*,blue,mark size=1.5pt}
\addlegendimage{mark=triangle*,magenta,mark size=1.5pt}
\addlegendimage{mark=+,cyan,mark size=1.5pt}
\addlegendimage{mark=square,orange,mark size=1.5pt}
]

\addplot[mark=square*,red,mark size=1.5pt] table [x=k, y=minmax3, col sep=comma]{rbts_deadline_minmax_players.csv};
\addplot[mark=o,blue,mark size=1.5pt] table [x=k, y=minmax4, col sep=comma]{rbts_deadline_minmax_players.csv};
\addplot[mark=triangle*,magenta,mark size=1.5pt] table [x=k, y=minmax5, col sep=comma]{rbts_deadline_minmax_players.csv};
\addplot[mark=+,cyan,mark size=1.5pt] table [x=k, y=minmax6, col sep=comma]{rbts_deadline_minmax_players.csv};
\addplot[mark=o,orange,mark size=1.5pt] table [x=k, y=minmax7, col sep=comma]{rbts_deadline_minmax_players.csv};
\addplot[dashed,mark=square,orange,mark size=1.5pt] table [x=k, y=minmax7, col sep=comma]{rbts_deadline_minmax_players.csv};
\end{axis}
\end{tikzpicture}
}
\vspace*{-0.2cm}
\caption{Zero-sum objectives}
\end{subfigure}\begin{subfigure}{.49\textwidth}
\centering
\tiny{
\begin{tikzpicture}
\begin{axis}[
    title style={yshift=-2ex},
    title={$\mathit{robot}_{\scale{.75}{1}}$ (solid lines) and $\mathit{robot}_{\scale{.75}{2}}$ (dashed lines)},
    ylabel={Probability},
    xlabel={$k$},
    xmin=0, xmax=10,
    xtick={0,1,2,3,4,5,6,7,8,9,10},
    ymin=0, ymax=1,
    ytick={0,0.25,0.5,0.75,1},
    legend pos=south east,
    legend style={fill=none, draw=none, 
                    at={(1,0)}, anchor=south east, 
                    nodes={scale=0.9, transform shape}},
    ymajorgrids=true,
    grid style=dashed,
    height=4.25cm,
    width=0.95\textwidth,
    legend entries={
                $l{=}3$,
                $l{=}4$,
                $l{=}5$,
                $l{=}6$,      
                $l{=}7$                
                }]
\addlegendimage{mark=square*,red,mark size=1.5pt}
\addlegendimage{mark=o*,blue,mark size=1.5pt}
\addlegendimage{mark=triangle*,magenta,mark size=1.5pt}
\addlegendimage{mark=+,cyan,mark size=1.5pt}
\addlegendimage{mark=square,orange,mark size=1.5pt}
]

\addplot[mark=square*,red,mark size=1.5pt] table [x=k, y=p1_3, col sep=comma]{rbts_deadline_nash_players.csv};
\addplot[dashed,mark=square*,red,mark size=1.5pt] table [x=k, y=p2_3, col sep=comma]{rbts_deadline_nash_players.csv};
\addplot[mark=o,blue,mark size=1.5pt] table [x=k, y=p1_4, col sep=comma]{rbts_deadline_nash_players.csv};
\addplot[dashed,mark=o*,blue,mark size=1.5pt] table [x=k, y=p2_4, col sep=comma]{rbts_deadline_nash_players.csv};
\addplot[mark=triangle*,magenta,mark size=1.5pt] table [x=k, y=p1_5, col sep=comma]{rbts_deadline_nash_players.csv};
\addplot[dashed,mark=triangle*,magenta,mark size=1.5pt] table [x=k, y=p2_5, col sep=comma]{rbts_deadline_nash_players.csv};
\addplot[mark=+,cyan,mark size=1.5pt] table [x=k, y=p1_6, col sep=comma]{rbts_deadline_nash_players.csv};
\addplot[dashed,mark=square,cyan,mark size=1.5pt] table [x=k, y=p2_6, col sep=comma]{rbts_deadline_nash_players.csv};
\addplot[mark=o,orange,mark size=1.5pt] table [x=k, y=p1_7, col sep=comma]{rbts_deadline_nash_players.csv};
\addplot[dashed,mark=square,orange,mark size=1.5pt] table [x=k, y=p2_7, col sep=comma]{rbts_deadline_nash_players.csv};

\end{axis}
\end{tikzpicture}
}
\vspace*{-0.2cm}
\caption{Nonzero-sum objectives}
\end{subfigure}\vspace*{-0.2cm}
\caption{Robot coordination: probability of reaching the goal without crashing.}\label{robot-deadline-fig}
\vspace*{-0.2cm}
\end{figure}
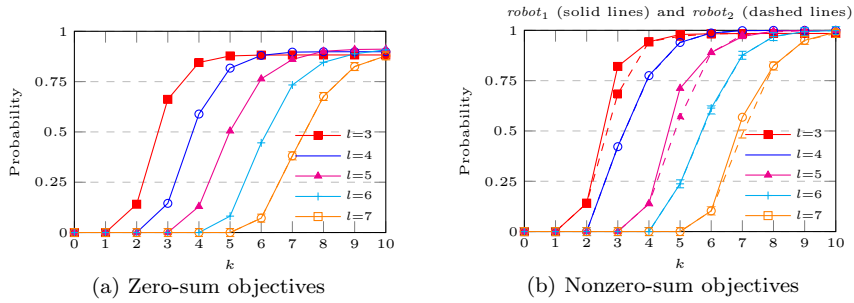
\begin{figure}[t]
\centering
\begin{subfigure}{.49\textwidth}
\centering
\tiny{
\begin{tikzpicture}
\begin{axis}[
    ylabel={Expected Cost},
    xlabel={$l$},
    xmin=3, xmax=16,
    xtick={4,6,8,10,12,14,16},
    ymin=2, ymax=19.5,
    ytick={2,4,6,8,10,12,14,16,18},
    legend pos=north west,
    legend style={fill=none},
    ymajorgrids=true,
    grid style=dashed,
    height=4.5cm,
    width=0.95\textwidth,
]
\addplot[mark=square*,red,mark size=1.5pt] table [x=l, y=minmax, col sep=comma]{rbts_expected_players.csv};

\end{axis}
\end{tikzpicture}
}
\vspace*{-0.2cm}
\caption{Zero-sum objectives}
\end{subfigure}\begin{subfigure}{.49\textwidth}
\centering
\tiny{
\begin{tikzpicture}
\begin{axis}[
    ylabel={Expected Cost},
    xlabel={$l$},
    xmin=3, xmax=16,
    xtick={4,6,8,10,12,14,16},
    ymin=2, ymax=19.5,
    ytick={2,4,6,8,10,12,14,16,18},
    legend pos=north west,
    legend style={fill=none, draw=none, 
                    at={(1,0)}, anchor=south east, 
                    nodes={scale=0.9, transform shape}},
    ymajorgrids=true,
    grid style=dashed,
    height=4.5cm,
    width=0.95\textwidth,
    legend entries={
                $\mathit{robot}_{\scale{.75}{1}}$,
                $\mathit{robot}_{\scale{.75}{2}}$
                },
legend pos=north west]
\addlegendimage{mark=square*,red,mark size=1.5pt}
\addlegendimage{mark=*,blue,mark size=1.5pt}
]

\addplot[mark=square*,red,mark size=1.5pt] table [x=l, y=player1, col sep=comma]{rbts_expected_players.csv};
\addplot[mark=*,blue,mark size=1.5pt] table [x=l, y=player2, col sep=comma]{rbts_expected_players.csv};

\end{axis}
\end{tikzpicture}
}
\vspace*{-0.2cm}
\caption{Nonzero-sum objectives}
\end{subfigure}\vspace*{-0.2cm}
\caption{Robot coordination: expected steps to reach the goal.}\label{robot-expected-fig}
\vspace*{-0.4cm}
\end{figure}
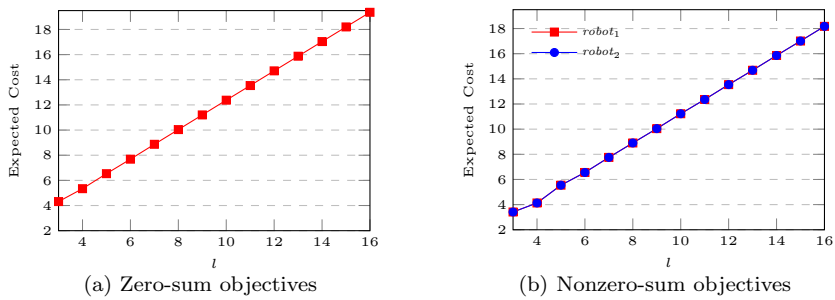

\startpara{Futures market investors} This case study is a model of a futures market investor~\cite{MM07},
which represents the interactions between investors and a stock market.
For the TSG model of~\cite{MM07}, in successive months, a single investor chooses whether to invest, next the market decides whether to bar the investor, with the restriction that the investor cannot be barred two months in a row or in the first month, and then the values of shares and a cap on values are updated probabilistically.

We have built and analysed several CSGs variants of the model,
analysing optimal strategies for investors under adversarial conditions.
First, we made a single investor and market take their decisions concurrently,
and verified that this yielded no additional gain for the investor (see~\cite{files}).
This is because the market and investor have the same information,
and so the market knows when it is optimal for the investor to invest without needing to see its decision. We next modelled two competing investors who simultaneously decide whether to invest (and, as above, the market simultaneously decides which investors to bar).
If the two investors cash in their shares in the same month, then their profits are reduced.
We also consider several distinct profit models:
`normal market', `later cash-ins', `later cash-ins with fluctuation' and `early cash-ins'.
The first is from~\cite{MM07} and the remaining reward models either postponing cashing in shares or the early chasing in of shares
(see~\cite{files} for details).
The CSG has 3 players: one for each investor and one representing the market who decides on the barring of investors. We study both the maximum profit of one investor and the maximum combined profit of both investors. 
For comparison, we also build a TSG model in which the investors first take turns to decide whether to invest (the ordering decided by the market) and then the market decides on whether to bar any of the investors.  

\figref{fig:2inv:normal} shows the maximum expected value over a fixed number of months under the `normal market' for both the profit of first investor and the combined profit of the two investors. For the former, we show results for the formulae $\coalition{i_1} \rewop{\mathit{profit}_1}{\max=?}{\future \mathsf{cashed\_in}_1}$, corresponding to the first investor acting alone, and $\coalition{i_1,i_2} \rewop{\mathit{profit}_{1,2}}{\max=?}{\future \mathsf{cashed\_in}_{1,2}}$ when in a coalition with the second investor. We plot the corresponding results from the TSG model for comparison.
\figref{fig:2inv:laterboth} shows the maximum expected combined profit for the other two profit models.

When investors cooperate to maximise the profit of the first, results for the CSG and TSG models coincide. This follows from the discussion above since all the second investor can do is make sure it does not invest at the same time as the first. For the remaining cases and given sufficient months, there is always a strategy in the concurrent setting that outperforms all turn-based strategies.
The increase in profit for a single investor in the CSG model is due to the fact that, as the investors decisions are concurrent, the second cannot ensure it invests at the same time as the first, and hence decreases the profit of the first. In the case of combined profit, the difference arises because, although the market knows when it is optimal for one investor to invest, in the CSG model the market does not know which one will, and therefore may choose the wrong investor to bar. %

\begin{figure}[t]
\centering
\begin{subfigure}{.49\textwidth}
\begin{center}
\tiny{
\begin{tikzpicture}
\begin{axis}[
    ylabel={Max profit for first investor},
    xlabel={Number of months},
    xmin=1, xmax=9,
    xtick={1, 2, 3, 4, 5, 6, 7, 8, 9},
    ymin=2.5, ymax=5.5,
    ytick={2.5,3,3.5,4,4.5,5,5.5},
    xtick pos=left,
    ytick pos=left,
    ymajorgrids=true,
    grid style=dashed,
    height=4.5cm,
    width=0.95\textwidth,
    legend entries={
                \textit{CSG $\coalition{\mathit{i1}}$},
                \textit{TSG $\coalition{\mathit{i1}}$},
                \textit{CSG $\coalition{\mathit{i1},\mathit{i2}}$},
                \textit{TSG $\coalition{\mathit{i1},\mathit{i2}}$}                },
    legend style={fill=none, draw=none, 
                    at={(1,0)}, anchor=south east, 
                    nodes={scale=0.9, transform shape}}
]
\addlegendimage{mark=square*,red,mark size=1.5pt}
\addlegendimage{mark=*,blue,mark size=1.5pt}
\addlegendimage{mark=triangle*,magenta,mark size=1.5pt}
\addlegendimage{mark=+,cyan,mark size=1.5pt}
\addlegendimage{thick}
]
legend entries={simulation},
\addlegendimage{only marks, mark=triangle*}

\addplot[mark=square*,red,mark size=1.5pt] table [x=n, y=csg1i1, col sep=comma]{2inv1normal.csv};
\addplot[mark=*,blue,mark size=1.5pt] table [x=n, y=stpg1i1, col sep=comma]{2inv1normal.csv};
\addplot[mark=triangle*,magenta,mark size=1.5pt] table [x=n, y=csg1i1i2, col sep=comma]{2inv1normal.csv};
\addplot[mark=+,cyan,mark size=1.5pt] table [x=n, y=stpg1i1i2, col sep=comma]{2inv1normal.csv};

\end{axis}
\end{tikzpicture}
}
\end{center}
\end{subfigure}
\begin{subfigure}{.49\textwidth}
\begin{center}
\tiny{
\begin{tikzpicture}
\begin{axis}[
    ylabel={Max combined profit},
    xlabel={Number of months},
    xmin=1, xmax=9,
    xtick={1, 2, 3, 4, 5, 6, 7, 8, 9},
    ymin=7, ymax=11,
    ytick={7.5,8,8.5,9,9.5,10,10.5,11},
    xtick pos=left,
    ytick pos=left,
    ymajorgrids=true,
    grid style=dashed,
    height=4.5cm,
    width=0.95\textwidth,
    legend entries={
                \textit{CSG $\coalition{\mathit{i1},\mathit{i2}}$},
                \textit{TSG $\coalition{\mathit{i1},\mathit{i2}}$}               },
    legend style={fill=none, draw=none, 
                    at={(1,0)}, anchor=south east, 
                    nodes={scale=0.9, transform shape}}
]
\addlegendimage{mark=square*,red,mark size=1.5pt}
\addlegendimage{mark=*,blue,mark size=1.5pt}
\addlegendimage{thick}
]
legend entries={simulation},
\addlegendimage{only marks, mark=triangle*}

\addplot[mark=square*,red,mark size=1.5pt] table [x=n, y=csg1i1i2, col sep=comma]{2inv2normal.csv};
\addplot[mark=*,blue,mark size=1.5pt] table [x=n, y=stpg1i1i2, col sep=comma]{2inv2normal.csv};

\end{axis}
\end{tikzpicture}
}
\end{center}
\end{subfigure}
\vspace*{-0.2cm}
\caption{Futures market investors: normal market.}
\label{fig:2inv:normal}
\vspace*{-0.4cm}
\end{figure}
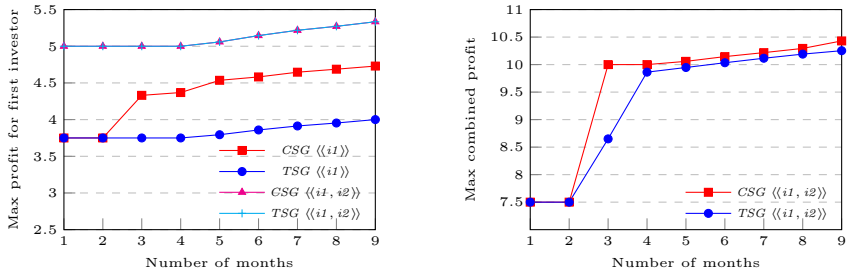\begin{figure}[t]
\centering

\begin{subfigure}{.49\textwidth}
\begin{center}
\tiny{
\begin{tikzpicture}
\begin{axis}[
    ylabel={Max combined profit},
    xlabel={Number of months},
    xmin=1, xmax=9,
    xtick={1, 2, 3, 4, 5, 6, 7, 8, 9},
    ymin=7, ymax=16,
    ytick={7,8,9,10,11,12,13,14,15,16},
    xtick pos=left,
    ytick pos=left,
    ymajorgrids=true,
    grid style=dashed,
    height=4.5cm,
    width=0.95\textwidth,
    legend entries={
                \textit{CSG $\coalition{\mathit{i1},\mathit{i2}}$},
                \textit{TSG $\coalition{\mathit{i1},\mathit{i2}}$}                },
    legend style={fill=none, draw=none, 
                    at={(1,0)}, anchor=south east, 
                    nodes={scale=0.9, transform shape}}
]
\addlegendimage{mark=square*,red,mark size=1.5pt}
\addlegendimage{mark=*,blue,mark size=1.5pt}
\addlegendimage{thick}
]
legend entries={simulation},
\addlegendimage{only marks, mark=triangle*}

\addplot[mark=square*,red,mark size=1.5pt] table [x=n, y=csg2i1i2, col sep=comma]{2inv2later.csv};
\addplot[mark=*,blue,mark size=1.5pt] table [x=n, y=stpg2i1i2, col sep=comma]{2inv2later.csv};

\end{axis}
\end{tikzpicture}
}
\end{center}
\end{subfigure}
\begin{subfigure}{.49\textwidth}
\begin{center}
\tiny{
\begin{tikzpicture}
\begin{axis}[
    ylabel={Max combined profit},
    xlabel={Number of months},
    xmin=1, xmax=9,
    xtick={1, 2, 3, 4, 5, 6, 7, 8, 9},
    ymin=9, ymax=27.5,
    ytick={9,10,12.5,15,17.5,20,22.5,25,27.5},
    xtick pos=left,
    ytick pos=left,
    ymajorgrids=true,
    grid style=dashed,
    height=4.5cm,
    width=0.95\textwidth,
    legend entries={
                \textit{CSG $\coalition{\mathit{i1},\mathit{i2}}$},
                \textit{TSG $\coalition{\mathit{i1},\mathit{i2}}$}                },
    legend style={fill=none, draw=none, at={(1,0)},
                    anchor=south east, 
                    nodes={scale=0.9, transform shape}}
]
\addlegendimage{mark=square*,red,mark size=1.5pt}
\addlegendimage{mark=*,blue,mark size=1.5pt}
\addlegendimage{thick}
]
legend entries={simulation},
\addlegendimage{only marks, mark=triangle*}

\addplot[mark=square*,red,mark size=1.5pt] table [x=n, y=csg1i1i2, col sep=comma]{2inv2laterfluct.csv};
\addplot[mark=*,blue,mark size=1.5pt] table [x=n, y=stpg1i1i2, col sep=comma]{2inv2laterfluct.csv};

\end{axis}
\end{tikzpicture}
}
\end{center}
\end{subfigure}
\vspace*{-0.2cm}
\caption{Futures market: later cash-ins without (left) and with (right) fluctuations.}
\label{fig:2inv:laterboth}
\vspace*{-0.4cm}
\end{figure}
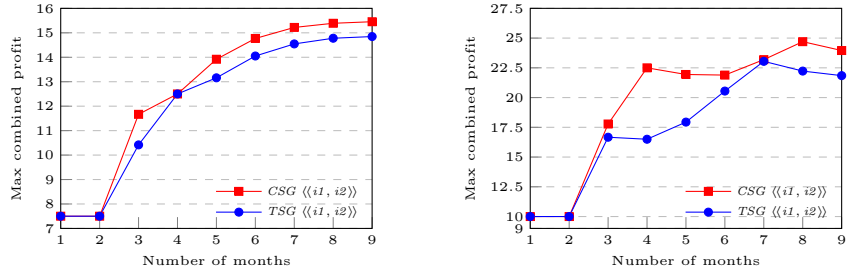

We performed strategy synthesis to study the optimal actions of investors.
By way of example, consider $\coalition{\mathit{i_1}}\rewop{\mathit{profit}_1}{\max=?}{\future \mathsf{cashed\_in}_1}$ over three months and for a normal market (see \figref{fig:2inv:normal} left). The optimal TSG strategy for the first investor is to invest in the first month (which the market cannot bar) ensuring an expected profit of $3.75$. The optimal (randomised) CSG strategy is to invest:
\begin{itemize}
\item
in the first month with probability ${\sim}0.4949$;
\item
in the second month with probability 1, if the second investor has cashed in;
\item
in the second month with probability ${\sim}0.9649$, if the second investor did not cash in at the end of the first month and the shares went up;
\item
in the second month with probability ${\sim}0.9540$, if the second investor did not cash in at the end of the first month and the shares went down;
\item
in the third month with probability 1 (this is the last month to invest).
\end{itemize}
Following this strategy, the first investor ensures an expected profit of ${\sim}4.33$.

We now make the market probabilistic, where the probability the market bars an individual investor equals $\mathit{pbar}$, and consider nonzero-sum properties of the form $\nashop{i_1{:}i_2}{\max=?}{\rewop{\mathit{profit}_1}{}{\future \mathsf{cashed\_in}_1}{+}\rewop{\mathit{profit}_2}{}{\future \mathsf{cashed\_in}_2}}$, in which each investor tries to maximise their individual profit, for different reward structures. In \figfigref{fig:2inv:nash01}{fig:2inv:nash05} we have plotted the results for the investors where the profit models of the investors follow a normal profile and where the profit models of the investors differ (`later cash-ins' for the first investor and `early cash-ins' for second), when $\mathit{pbar}$ equals $0.1$ and $0.5$ respectively. The results demonstrate that, given more time and a more predictable market, i.e., when $\mathit{pbar}$ is lower, the players can collaborate to increase their profits.

Performing strategy synthesis, we find that the strategies in the mixed profiles model are for the investor with an `early cash-ins' profit model to invest as soon as possible, i.e., it tries to invest in the first month and if this fails because it is barred, it will be able to invest in the second. On the other hand, for the investor with the `later cash-ins' profile, the investor will delay investing until the chances of the shares failing starts to increase or they reach the month before last and then invest (if the investor is barred in this month, they will be able to invest in the final month).

\begin{figure}[t]
\centering

\begin{subfigure}{.49\textwidth}
\begin{center}
\tiny{
\begin{tikzpicture}
\begin{axis}[
    ylabel={Profit},
    xlabel={Number of months},
    xmin=1, xmax=9,
    xtick={1, 2, 3, 4, 5, 6, 7, 8, 9},
    ymin=3.4, ymax=5.6,
    ytick={3.4,4.0,4.4,4.8,5.2,5.6},
    xtick pos=left,
    ytick pos=left,
    ymajorgrids=true,
    grid style=dashed,
    height=4.5cm,
    width=0.95\textwidth,
    legend entries={
                {$i_{\scale{.75}{1}}$ (normal market)},
                {$i_{\scale{.75}{2}}$ (normal market)}
                },
    legend style={fill = none, draw=none, 
            at={(1,0)}, anchor=south east, 
            nodes={scale=1.0, transform shape}},
]
\addlegendimage{mark=square*,red,mark size=1.5pt}
\addlegendimage{mark=*,blue,mark size=1.5pt}
\addlegendimage{thick}
]

\addplot[mark=square*,red,mark size=1.5pt] table [x=months, y=profit1_.1, col sep=comma]{invstrs_prft1_prtf2_up_v2.csv};
\addplot[mark=*,blue,mark size=1.5pt] table [x=months, y=profit2_.1, col sep=comma]{invstrs_prft1_prtf2_up_v2.csv};

\end{axis}
\end{tikzpicture}
}
\end{center}
\end{subfigure}
\begin{subfigure}{.49\textwidth}
\begin{center}
\tiny{
\begin{tikzpicture}
\begin{axis}[
    ylabel={Profit},
    xlabel={Number of months},
    xmin=1, xmax=9,
    xtick={1, 2, 3, 4, 5, 6, 7, 8, 9},
    ymin=3.4, ymax=9.5,
    ytick={3.4,4.5,5.5,6.5,7.5,8.5,9.5},
    xtick pos=left,
    ytick pos=left,
    ymajorgrids=true,
    grid style=dashed,
    height=4.5cm,
    width=0.95\textwidth,
    legend entries={
                {$i_{\scale{.75}{1}}$ (later cashin)},
                {$i_{\scale{.75}{2}}$ (earlier cashin)}
                },
    legend style={fill=none, draw=none, 
            at={(1,0)}, anchor=south east, 
            nodes={scale=1.0, transform shape}},
]
\addlegendimage{mark=square*,red,mark size=1.5pt}
\addlegendimage{mark=*,blue,mark size=1.5pt}
\addlegendimage{thick}
]

\addplot[mark=square*,red,mark size=1.5pt] table [x=months, y=profit1lc_.1, col sep=comma]{invstrs_prft1lc_prft2ec_up_v2.csv};
\addplot[mark=*,blue,mark size=1.5pt] table [x=months, y=profit2ec_.1, col sep=comma]{invstrs_prft1lc_prft2ec_up_v2.csv};

\end{axis}
\end{tikzpicture}
}
\end{center}
\end{subfigure}
\vspace*{-0.2cm}
\caption{Futures market: normal profiles (left) and mixed profiles (right) ($\mathit{pbar}{=}0.1$).}\label{fig:2inv:nash01}
\vspace*{-0.4cm}
\end{figure}
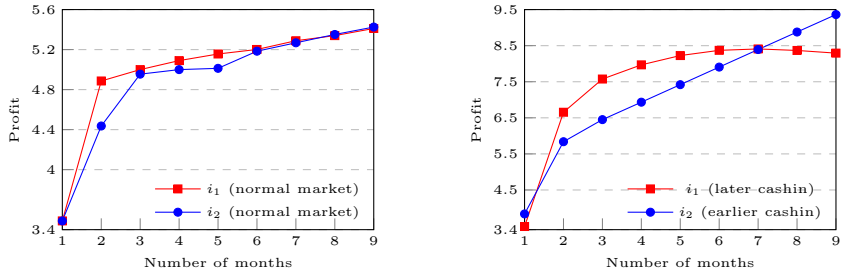

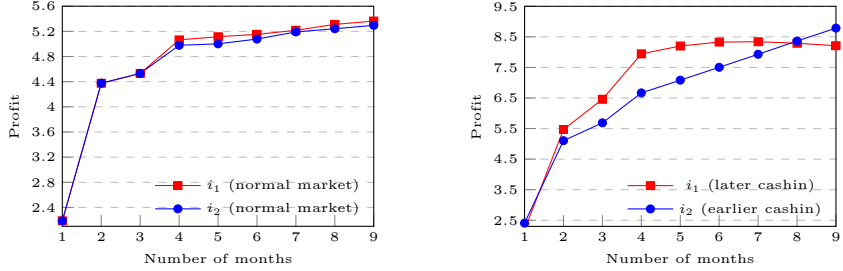
\begin{figure}[t]
\centering

\begin{subfigure}{.49\textwidth}
\begin{center}
\tiny{
\begin{tikzpicture}
\begin{axis}[
    ylabel={Profit},
    xlabel={Number of months},
    xmin=1, xmax=9,
    xtick={1, 2, 3, 4, 5, 6, 7, 8, 9},
    ymin=2.1, ymax=5.6,
    ytick={2.4,2.8,3.2,3.6,4.0,4.4,4.8,5.2,5.6},
    xtick pos=left,
    ytick pos=left,
    ymajorgrids=true,
    grid style=dashed,
    height=4.5cm,
    width=0.95\textwidth,
    legend entries={
                {$i_{\scale{.75}{1}}$ (normal market)},
                {$i_{\scale{.75}{2}}$ (normal market)}
                },
    legend style={fill=none, draw=none, 
            at={(1,0)}, anchor=south east, 
            nodes={scale=1.0, transform shape}}
            ]
\addlegendimage{mark=square*,red,mark size=1.5pt}
\addlegendimage{mark=*,blue,mark size=1.5pt}
\addlegendimage{thick}
]

\addplot[mark=square*,red,mark size=1.5pt] table [x=months, y=profit1_.5, col sep=comma]{invstrs_prft1_prtf2_up_v2.csv};
\addplot[mark=*,blue,mark size=1.5pt] table [x=months, y=profit2_.5, col sep=comma]{invstrs_prft1_prtf2_up_v2.csv};

\end{axis}
\end{tikzpicture}
}
\end{center}
\end{subfigure}
\begin{subfigure}{.49\textwidth}
\begin{center}
\tiny{
\begin{tikzpicture}
\begin{axis}[
    ylabel={Profit},
    xlabel={Number of months},
    xmin=1, xmax=9,
    xtick={1, 2, 3, 4, 5, 6, 7, 8, 9},
    ymin=2.3, ymax=9.5,
    ytick={2.5,3.5,4.5,5.5,6.5,7.5,8.5,9.5},
    xtick pos=left,
    ytick pos=left,
    ymajorgrids=true,
    grid style=dashed,
    height=4.5cm,
    width=0.95\textwidth,
    legend entries={
                {$i_{\scale{.75}{1}}$ (later cashin)},
                {$i_{\scale{.75}{2}}$ (earlier cashin)}
                },
    legend style={fill=none, draw=none, 
            at={(1,0)}, anchor=south east, 
            nodes={scale=1.0, transform shape}},
]
\addlegendimage{mark=square*,red,mark size=1.5pt}
\addlegendimage{mark=*,blue,mark size=1.5pt}
\addlegendimage{thick}
]

\addplot[mark=square*,red,mark size=1.5pt] table [x=months, y=profit1lc_.5, col sep=comma]{invstrs_prft1lc_prft2ec_up_v2.csv};
\addplot[mark=*,blue,mark size=1.5pt] table [x=months, y=profit2ec_.5, col sep=comma]{invstrs_prft1lc_prft2ec_up_v2.csv};

\end{axis}
\end{tikzpicture}
}
\end{center}
\end{subfigure}
\vspace*{-0.2cm}
\caption{Futures market: normal profiles (left) and mixed profiles (right) ($\mathit{pbar}{=}0.5$).}
\label{fig:2inv:nash05}
\vspace*{-0.4cm}
\end{figure}
\startpara{Trust models for user-centric networks} Trust models for user-centric networks were analysed previously using TSGs in~\cite{KPS13}. The analysis considered the impact of different parameters on the effectiveness of cooperation mechanisms between service providers. The providers share information on the measure of \emph{trust} for users in a \emph{reputation}-based setting. Each measure of trust is based on the service's previous interactions with the user (which previous services they paid for), and providers use this measure to block or allow the user to obtain services. 

In the original TSG model, a single user can either make a request to one of three service providers or buy the service directly by paying maximum price. If the user makes a request to a service provider, then the provider decides to accept or deny the request based on the user's trust measure. If the request was accepted, the provider would next decide on the price again based on the trust measure, and the user would then decide whether to pay for the service and finally the provider would update its trust measure based on whether there was a payment. This sequence of steps would have to take place before any other interactions occurred between the user and other providers. 
Here we consider CSG models allowing the user to make requests and pay different service providers simultaneously and for the different providers to execute requests concurrently. There are 7 players: one for the user's interaction with each service provider, one for the user buying services directly and one for each of the 3 service providers. Three trust models were considered. In the first, the trust level was decremented by 1 ($\mathit{td} = 1$) when the user does not pay, decremented by 2 in the second ($\mathit{td} = 2$) and reset to 0 in the third ($\mathit{td} = \mathit{inf}$). 

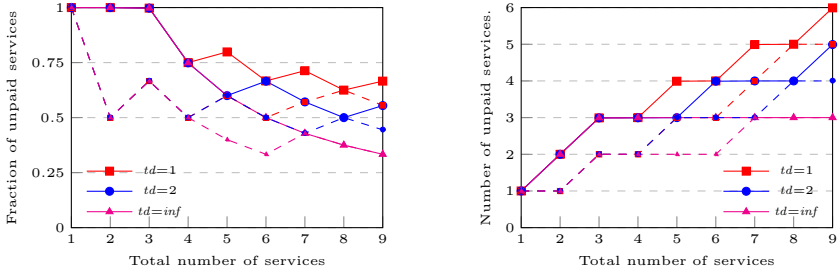
\begin{figure}[t]
\centering
\begin{subfigure}{.49\textwidth}
\begin{center}
\tiny{
\begin{tikzpicture}
\begin{axis}[
    title={},
    ylabel={Fraction of unpaid services},
    xlabel={Total number of services},
    xmin=1, xmax=9,
    xtick={1, 2, 3, 4, 5, 6, 7, 8, 9},
    ymin=0, ymax=1,
    ytick={0,.25,.50,.75,1},
    xtick pos=left,
    ytick pos=left,
    ymajorgrids=true,
    grid style=dashed,
    height=4.5cm,
    width=0.95\textwidth,
    legend entries={
                \textit{$\mathit{td}{=}1$},
                \textit{$\mathit{td}{=}2$},
                \textit{$\mathit{td}{=}\mathit{inf}$}
                },
    legend style={fill=none, draw=none, 
                    at={(0.4,0)}, anchor=south east, 
                    nodes={scale=0.9, transform shape}}
]
\addlegendimage{mark=square*,red,mark size=1.5pt}
\addlegendimage{mark=*,blue,mark size=1.5pt}
\addlegendimage{mark=triangle*,magenta,mark size=1.5pt}
\addlegendimage{thick}
]
legend entries={simulation},
\addlegendimage{only marks, mark=triangle*}

\addplot[mark=square*,red,mark size=1.5pt] table [x=k, y=red1, col sep=comma]{rewratio_3.csv};
\addplot[mark=*,blue,mark size=1.5pt] table [x=k, y=red2, col sep=comma]{rewratio_3.csv};
\addplot[mark=triangle*,magenta,mark size=1.5pt] table [x=k, y=red0, col sep=comma]{rewratio_3.csv};

\addplot[dashed,mark=square*,red,mark size=1pt] table [x=k, y=red1, col sep=comma]{rewratio_tsg.csv};
\addplot[dashed,mark=*,blue,mark size=1pt] table [x=k, y=red2, col sep=comma]{rewratio_tsg.csv};
\addplot[dashed,mark=triangle*,magenta,mark size=1pt] table [x=k, y=red0, col sep=comma]{rewratio_tsg.csv};

\end{axis}
\end{tikzpicture}
}
\end{center}
\end{subfigure}
\begin{subfigure}{.49\textwidth}
\begin{center}
\tiny{
\begin{tikzpicture}
\begin{axis}[
    title={},
    ylabel={Number of unpaid services.},
    xlabel={Total number of services},
    xmin=1, xmax=9,
    xtick={1, 2, 3, 4, 5, 6, 7, 8, 9},
    ymin=0, ymax=6,
    ytick={0, 1, 2, 3, 4, 5, 6},
    xtick pos=left,
    ytick pos=left,
    ymajorgrids=true,
    grid style=dashed,
    height=4.5cm,
    width=0.95\textwidth,
    legend entries={
                \textit{$\mathit{td}{=}1$},
                \textit{$\mathit{td}{=}2$},
                \textit{$\mathit{td}{=}\mathit{inf}$}
                },
    legend style={fill=none, draw=none, 
                    at={(1,0)}, anchor=south east, 
                    nodes={scale=0.9, transform shape}}
]
\addlegendimage{mark=square*,red,mark size=1.5pt}
\addlegendimage{mark=*,blue,mark size=1.5pt}
\addlegendimage{mark=triangle*,magenta,mark size=1.5pt}
\addlegendimage{thick}
]
legend entries={simulation},
\addlegendimage{only marks, mark=triangle*}

\addplot[mark=square*,red,mark size=1.5pt] table [x=k, y=red1, col sep=comma]{rewnopayed_3.csv};
\addplot[mark=*,blue,mark size=1.5pt] table [x=k, y=red2, col sep=comma]{rewnopayed_3.csv};
\addplot[mark=triangle*,magenta,mark size=1.5pt] table [x=k, y=red0, col sep=comma]{rewnopayed_3.csv};

\addplot[dashed,mark=square*,red,mark size=1pt] table [x=k, y=red1, col sep=comma]{rewnopaytsg.csv};
\addplot[dashed,mark=*,blue,mark size=1pt] table [x=k, y=red2, col sep=comma]{rewnopaytsg.csv};
\addplot[dashed,mark=triangle*,magenta,mark size=1pt] table [x=k, y=red0, col sep=comma]{rewnopaytsg.csv};

\end{axis}
\end{tikzpicture}
}
\end{center}
\end{subfigure}
\vspace*{-0.2cm}
\caption{User-centric network results (CSG/TSG values as solid/dashed lines).}
\label{fig:trust:unpaid_3}
\vspace*{-0.4cm}
\end{figure}
\figref{fig:trust:unpaid_3} presents results for the maximum fraction and number of unpaid services the user can ensure for each trust model, corresponding to the formulae $\coalition{\mathit{usr}} \rewop{\mathit{ratio}^-}{\min=?}{\future  \mathsf{finished}}$ and $\coalition{\mathit{usr}} \rewop{\mathit{unpaid}^-}{\min=?}{\future \mathsf{finished}}$ (to prevent not requesting any services and obtaining an infinite reward being the optimal choice of the user, we negate all rewards and find the minimum expected reward the user can ensure). 
The results for the original TSG model are included as dashed lines. The results demonstrate that the user can take advantage of the fact that in the CSG model it can request multiple services at the same time, and obtain more services without paying before the different providers get a chance to inform each other about non-payment. In addition, the results show that imposing %
a more severe penalty on the trust measure for non-payment reduces %
the %
number of services the user can obtain without paying.

\startpara{Aloha} This case study concerns three users trying to send packets using the slotted ALOHA protocol. In a time slot, if a single user tries to send a packet, there is a probability ($q$) that the packet is sent; as more users try and send, then the probability of success decreases. If sending a packet fails, the number of slots a user waits before resending is set according to an exponential backoff scheme. More precisely, each user maintains a backoff counter which it increases each time there is a failure (up to $b_{\max}$) and, if the counter equals $k$, randomly chooses the slots to wait from $\{0,1,\dots,2^k{-}1\}$. 

\begin{figure}[t]
\centering
\begin{subfigure}{.49\textwidth}
\centering
\tiny{
\begin{tikzpicture}
\begin{axis}[
    ylabel={Sum of Probabilities},
    xlabel={$D$},
    xmin=0, xmax=10,
    xtick={0,1,2,3,4,5,6,7,8,9,10},
    ymin=0, ymax=2,
    ytick={0,0.5,1,1.5,2.0,2.5},
    xtick pos=left,
    ytick pos=left,
    ymajorgrids=true,
    grid style=dashed,
    height=4.25cm,
    width=0.95\textwidth,
    legend entries={
                {$\snashop{\mathit{usr}_{\scale{.75}{1}}{:}\mathit{usr}_{\scale{.75}{2}}}{=?}{\mathtt{P}{+}\mathtt{P}}$},
                \textit{$\coalition{\mathit{\mathit{usr}_{\scale{.75}{1}}}}{\mathtt P_{\scale{.75}{\max}}}+{\mathtt P_{\scale{.75}{\max}}}$},
                \textit{$\coalition{\mathit{\mathit{usr}_{\scale{.75}{1}}}}{\mathtt P_{\scale{.75}{\max}}}+{\mathtt P_{\scale{.75}{\min}}}$}    
                },
    legend style={fill=none, draw=none, 
                    at={(1,0)}, anchor=south east, 
                    nodes={scale=0.8, transform shape}},
]
\addlegendimage{mark=square*,red,mark size=1.5pt}
\addlegendimage{mark=*,blue,mark size=1.5pt}
\addlegendimage{mark=triangle,orange,mark size=1.5pt}
]

\addplot[mark=square*,red,mark size=1.5pt] table [x=tmax, y=peq, col sep=comma]{alh_bckff_3p_tmax.csv};
\addplot[mark=*,blue,mark size=1.5pt] table [x=tmax, y=pmxmx, col sep=comma]{alh_bckff_3p_tmax.csv};
\addplot[mark=triangle,orange,mark size=1.5pt] table [x=tmax, y=pmxmn, col sep=comma]{alh_bckff_3p_tmax.csv};

\end{axis}
\end{tikzpicture}
}
\end{subfigure}
\begin{subfigure}{.49\textwidth}
\centering
\tiny{
\begin{tikzpicture}
\begin{axis}[
    ylabel={Probability},
    xlabel={$q$},
    xmin=0, xmax=1,
    xtick={0,0.1,0.2,0.3,0.4,0.5,0.6,0.7,0.8,0.9,1.0},
    ymin=0, ymax=1,
    ytick={0,0.2,0.4,0.6,0.8,1.0},
    xtick pos=left,
    ytick pos=left,
    ymajorgrids=true,
    grid style=dashed,
    height=4.25cm,
    width=0.95\textwidth,
    legend entries={
                $\mathit{usr}_{\scale{.75}{1}}$ (SWNE),
                $\mathit{usr}_{\scale{.75}{2}}$ (SWNE),
                $\mathit{usr}_{\scale{.75}{1}}$ (Max),
                $\mathit{usr}_{\scale{.75}{2}}$ (Max)
                },
    legend style={fill=none, draw=none, 
                at={(0,1)}, anchor=north west, 
                nodes={scale=0.8, transform shape}}
]
\addlegendimage{mark=square*,red,mark size=1.5pt}
\addlegendimage{mark=*,blue,mark size=1.5pt}
\addlegendimage{mark=triangle*,magenta,mark size=1.5pt}
\addlegendimage{mark=+,cyan,mark size=1.5pt}
]

\addplot[mark=square*,red,mark size=1.5pt] table [x=p11, y=player1, col sep=comma]{alh_bckff_3p_p11_players.csv};
\addplot[mark=*,blue,mark size=1.5pt] table [x=p11, y=player2, col sep=comma]{alh_bckff_3p_p11_players.csv};;

\addplot[mark=triangle*,magenta,mark size=1.5pt] table [x=p11, y=p1, col sep=comma]{alh_bckff_3p_p11_players.csv};
\addplot[mark=+,cyan,mark size=1.5pt] table [x=p11, y=p2, col sep=comma]{alh_bckff_3p_p11_players.csv};

\end{axis}
\end{tikzpicture}
}
\end{subfigure}
\vspace*{-0.2cm}
\caption{Aloha: $\nashop{\mathit{usr}_1{:}\{\mathit{usr}_2,\mathit{usr}_3\}}{\max=?}{\probop{}{\future (\mathsf{s}_1 {\wedge} t{\leq}D)}{+}\probop{}{\future (\mathsf{s}_2 {\wedge} \mathsf{s}_3 {\wedge} t{\leq}D)}}$}
\label{aloha-deadline-fig}
\vspace*{-0.4cm}
\end{figure}

We suppose that the three users are each trying to maximise the probability of sending their packet before a deadline $D$, with users 2 and 3 forming a coalition, which corresponds to the formula $\coalition{\mathit{usr}_1{:}\mathit{usr}_2{,}\mathit{usr}_3}_{\max=?}\probop{}{\future (\mathsf{sent}_1 \wedge t {\leq} D)}+\probop{}{\future (\mathsf{sent}_2 \wedge \mathsf{sent}_3 \wedge t {\leq} D)}$. \figref{aloha-deadline-fig} presents total values as $D$ varies (left) and individual values as $q$ varies (right). Through synthesis, we find the collaboration is dependent on $D$ and $q$. Given more time there is a greater chance for the users to collaborate by sending in different slots, 
while if $q$ is large it is unlikely users need to repeatedly send, so again can send in different slots. 
As the coalition has more messages to send, their probabilities are lower.
However, for the scenario with two users, the probabilities of the two users would still be different. In this case, although it is advantageous to initially collaborate and allow one user to try and send its first message, if the sending fails, given there is a bound on the time for the users to send, both users will try to send at this point as this is the best option for their individual goals.

We have also considered when the users try to minimise the expected time before their packets are sent, where users 2 and 3 form a coalition, represented by the formula $\nashop{\mathit{usr}_1{:}\mathit{usr}_2{,}\mathit{usr}_3}{\min=?}{\rewop{\mathit{time}}{}{\future  \mathsf{sent}_1}{+}\rewop{\mathit{time}}{}{\future (\mathsf{sent}_2 \wedge \mathsf{sent}_3)}}$. When synthesising the strategies we see that the players collaborate with the coalition of users 2 and 3, letting user 1 to try and send before sending their messages. However, if user 1 fails to send, then the coalition either lets user 1 try again in case the user can do so immediately, and otherwise the coalition attempts to send their messages.

Finally, we have analysed when the players collaborate to maximise the probability of reaching a state where they can then send their messages with probability 1 within $D$ time units (with users 2 and 3 in coalition), which is represented by the formula $\coalition{\mathit{usr}_1{,}\mathit{usr}_2{,}\mathit{usr}_3}\probopP_{\max=?}[ \future \coalition{\mathit{usr}_1{:}\mathit{usr}_2{,}\mathit{usr}_3}_{\min \geq 2} \probop{}{\future (\mathsf{sent}_1 \wedge t {\leq} D)} + \probop{}{\future (\mathsf{sent}_2 \wedge \mathsf{sent}_3 \wedge t {\leq} D)}]$.

\startpara{Intrusion detection policies} In~\cite{ZB09}, CSGs are used to model the interaction between an intrusion detection policy and attacker. The policy has a number of libraries it can use to detect attacks and the attacker has a number of different attacks which can incur different levels of damage if not detected. Furthermore, each library can only detect certain attacks. In the model, in each round the policy chooses a library to deploy and the attacker chooses an attack. A reward structure is specified representing the level of damage when an attack is not detected. The goal is to find optimal intrusion detection policies which correspond to finding a strategy for the policy that minimises damage, represented by synthesising a strategy for the formula $\coalition{\mathit{policy}} \rewop{\mathit{damage}}{\min=?}{\future (r = \mathit{rounds})}$. We have constructed CSG models with two players (representing the policy and the attacker) for the two scenarios outlined in~\cite{ZB09}.

\startpara{Jamming multi-channel radio systems} A CSG model for jamming multi-channel cognitive radio systems is presented in~\cite{ZLHB10}. The system consists of a number of channels ($\mathit{chans}$), which can be in an occupied or idle state. The state of each channel remains fixed within a time slot and between slots is Markovian (i.e.\ the state changes randomly based only on the state of the channel in the previous slot). A secondary user has a subset of available channels and at each time-slot must decide which to use. There is a single attacker which again has a subset of available channels and at each time slot decides to send a jamming signal over one of them. The CSG has two players: one representing the secondary user and the other representing the attacker. Through %
the zero-sum property $\coalition{\mathit{user}} \probop{\max=?}{\future (\mathit{sent} \geq \mathit{slots}{/}2)}$ we find the optimal strategy for the secondary user to maximize the probability that at least half their messages are sent against any possible attack.

\startpara{Medium Access Control} %
This case study extends the CSG model from \egref{mac-eg} to three users and assumes that the probability of a successful transmission is dependent on the number of users that try and send ($q_1 = 0.95$, $q_2 = 0.75$ and $q_3 = 0.5$). The energy of each user is bounded by $e_{\max}$. We suppose the first user acts in isolation and the remaining users form a coalition. The first nonzero-sum property we consider is $\nashop{p_1{:}p_2{,}p_3}{\max=?}{\rewop{\mathit{sent}_1}{}{\scumul{\leq k}}{+}\rewop{\mathit{sent}_{2,3}}{}{\scumul{\leq k}}}$, which corresponds to each coalition trying to maximise the expected number of messages they send over a bounded number of steps. On the other hand, the second property is  $\nashop{p_1{:}p_2{,}p_3}{\max=?}{\probop{}{\bfuture (\mathit{mess}_1 = s_{\max})}{+}\probop{}{\future (\mathit{mess}_2{+}\mathit{mess}_3 = 2{\cdot}s_{\max})}}$ and here the coalitions try to maximise the probability of successfully transmitting a certain number of messages ($s_{\max}$ for the first user and $2{\cdot}s_{\max}$ for the coalition of the second and third users), where in addition the first user has to do this in a bounded number of steps ($k$).

\startpara{Power Control} Our final case study is based on a model of power control in cellular networks from~\cite{BRE13}. In the model, phones emit signals over a cellular network and the signals can be strengthened by increasing the power level up to a bound ($\mathit{pow}_{\max}$). A stronger signal can improve transmission quality, but uses more energy and lowers the quality of other transmissions due to interference. We extend this model by adding a failure probability ($q_\mathit{fail}$) when a power level is increased and assume each phone has a limited battery capacity ($e_{\max}$). Based on~\cite{BRE13}, we associate a reward structure with each phone representing transmission quality dependent both on its power level and that of other phones due to interference. We consider the nonzero-sum property $\nashop{p_1{:}p_2}{\max=?}{\rewop{\mathit{r}_1}{}{\future (e_1 = 0)}{+}\rewop{\mathit{r}_2}{}{\future (e_2 = 0)}}$, where each user tries to maximise their expected reward before their phone's battery is empty. We have also analysed the property $\nashop{p_1{:}p_2}{\max=?}{\rewop{\mathit{r}_1}{}{\future (e_1 = 0)}+\rewop{\mathit{r}_2}{}{\scumul{\leq k}}}$, where the objective of the second user is to instead maximise their expected reward over a bounded number of steps ($k$).

\section{Conclusions}

In this paper, we have designed and implemented an approach for the automatic verification of a large subclass of CSGs.
We have extended the temporal logic rPATL to allow for the specification of equilibria-based (nonzero-sum) properties, where two players or coalitions with distinct goals can collaborate.
We have then proposed and implemented algorithms for verification and strategy synthesis using this extended logic, including both zero-sum and nonzero-sum properties, in the PRISM-games model checker. In the case of finite-horizon properties the algorithms are exact, while for infinite-horizon they are approximate using value iteration. We have also extended the PRISM-games modelling language, adding new features tailored to CSGs.
Finally, we have evaluated the approach on a range of case studies that have demonstrated
the benefits %
of CSG models compared %
to TSGs
and of nonzero-sum properties as a means to 
synthesise strategies that are collectively more beneficial for all players in a game.

The main challenge in implementing the model checking algorithms is efficiently solving matrix and bimatrix games at each state in each step of value iteration for zero-sum and nonzero-sum properties, respectively, which are non-trivial optimisation problems.
For bimatrix games, this furthermore requires finding an optimal equilibrium, which currently relies on iteratively restricting the solution search space.
Solution methods can be sensitive to floating-point arithmetic issues,
particularly for bimatrix games;
arbitrary precision representations may help here to alleviate these problems.

There are a number of directions for future work. First, we plan to consider additional properties such as multi-objective queries. We are also working on extending the implementation to consider alternative solution methods (e.g., policy iteration and using CPLEX~\cite{CPLEX} to solve matrix games) and a symbolic (binary decision diagram based) implementation and other techniques for Nash equilibria synthesis such as an MILP-based solution using regret minimisation. Lastly, we are considering extending the approach to partially observable strategies, multi-coalitional games, building on~\cite{KNPS20b}, and mechanism design.

\startpara{Acknowledgements} 
This project has received funding from the European Research Council (ERC)
under the European Union’s Horizon 2020 research and innovation programme
(grant agreement No.~834115) 
and the EPSRC Programme Grant on Mobile Autonomy (EP/M019918/1).

\bibliographystyle{spmpsci}
\bibliography{bib}

\appendix
\section{Convergence of Zero-Sum Reachability Reward Formulae}\label{2-app}

In this appendix we give a witness to the failure of convergence for value iteration when verifying zero-sum formulae with an infinite horizon reward objective if \assumref{game1-assum} does not hold.

\begin{figure}[h]
\vspace*{-0.4cm}
\begin{subfigure}{.4\textwidth}
\centering
\begin{tikzpicture}[->,>=stealth',shorten >=1pt,auto,node distance=2.8cm, semithick, scale=.40]
  \tikzstyle{every state}=[draw=black,text=black, initial text=]
\small 
\node[initial,state,minimum width=0.75cm,minimum height=0.75cm] (S)at(0,0) (s0) {$s_1$}; 

\node[state,minimum width=0.75cm,minimum height=0.75cm] (S)at(8,0) (s1) {$s_2$}; 

\node[state,minimum width=0.75cm,minimum height=0.75cm] (S)at(8,-6) (s3) {$t$};

\path [->] (s0.north east) [bend left]
edge node []  {$a_1,a_2$} (s1.north west);

\path [->] (s1.south west) [bend left]
edge node []  {$a_1,a_2$} (s0.south east);

\path [->] (s1.south) []
edge node []  {$a_1,b_2$} (s3.north);

\end{tikzpicture} 
\end{subfigure}
\begin{subfigure}{.6\textwidth}
\begin{align*}
\\
r_A(s, (c_1,c_2) ) & = \; \begin{cases}
-1 & \mbox{if $s = s_1$ and $(c_1,c_2) = (a_1,a_2)$} \\
1 & \mbox{if $s = s_2$ and $(c_1,c_2) = (a_1,a_2)$} \\
0 & \mbox{otherwise}
\end{cases} \\
r_S(s) & = \; 0 \; \mbox{for all $s$} \\ \\ 
\end{align*}
\end{subfigure}
\caption{Counterexample for zero-sum expected reachability reward properties.}\label{counter2-fig}
\end{figure}
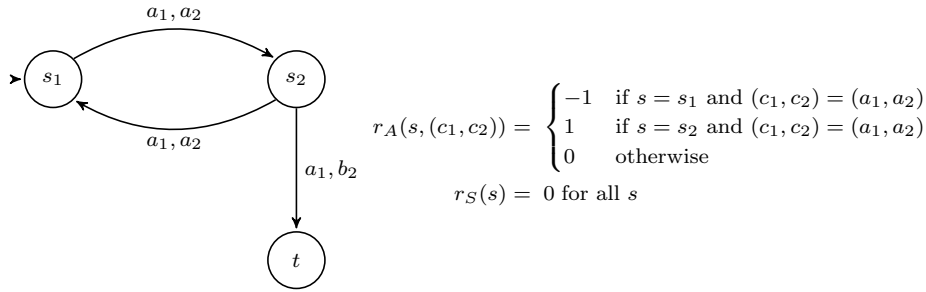

\noindent
Consider the CSG in \figref{counter2-fig} with players $p_1$ and $p_2$ and the zero-sum state formula $\phi=\coalition{p_1,p_2}\rewop{r}{\max=?}{\future \ap}$, where $\ap$ is the atomic proposition satisfied only by state $t$. Clearly, state $s_1$ does not reach either the target of the formula or an absorbing state with probability 1 under all strategy profiles, while the reward for the state-action pair $(s_1,(a_1,a_2))$ is negative. Applying the value iteration algorithm of \sectref{mc-sect}, we see that the values for state $s_1$ oscillate between $0$ and ${-}1$, while the values for state $s_2$ oscillate between $0$ and $1$.\section{Convergence of Nonzero-Sum Probabilistic Reachability Properties}\label{3-app}

In this appendix we give a witness to the failure of convergence for value iteration when verifying nonzero-sum formulae with infinite horizon probabilistic objectives if \assumref{game2-assum} does not hold.

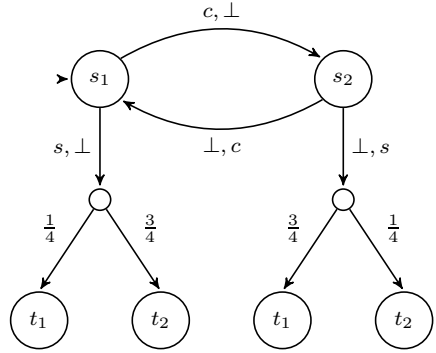
\begin{figure}[h]
\vspace*{-0.4cm}
\centering
\begin{tikzpicture}[->,>=stealth',shorten >=1pt,auto,node distance=2.8cm, semithick, scale=.40]

  \tikzstyle{every state}=[draw=black,text=black, initial text=]
\small 
\node[initial,state,minimum width=0.75cm,minimum height=0.75cm] (S)at(0,0) (s0) {$s_1$}; 

\node[state,minimum width=0.75cm,minimum height=0.75cm] (S)at(8,0) (s1) {$s_2$}; 

\node[state,minimum width=0cm,minimum height=0cm] (S)at(0,-4) (d1) {}; 

\node[state,minimum width=0cm,minimum height=0cm] (S)at(8,-4) (d2) {}; 

\node[state,minimum width=0.75cm,minimum height=0.75cm] (S)at(-2,-8) (t11) {$t_1$}; 

\node[state,minimum width=0.75cm,minimum height=0.75cm] (S)at(2,-8) (t12) {$t_2$}; 

\node[state,minimum width=0.75cm,minimum height=0.75cm] (S)at(6,-8) (t21) {$t_1$}; 

\node[state,minimum width=0.75cm,minimum height=0.75cm] (S)at(10,-8) (t22) {$t_2$};

\path [->] (s0.north east) [bend left]
edge node []  {$c,\bot$} (s1.north west);

\path [->] (s1.south west) [bend left]
edge node []  {$\bot,c$} (s0.south east);

\path [->] (s0.south) []
edge node [swap]  {$s,\bot$} (d1);

\path [->] (s1.south) []
edge node []  {$\bot,s$} (d2);

\path [->] (d1) []
edge node [swap]  {$\frac{1}{4}$} (t11.north);

\path [->] (d1) []
edge node []  {$\frac{3}{4}$} (t12.north);

\path [->] (d2) []
edge node [swap]  {$\frac{3}{4}$} (t21.north);

\path [->] (d2) []
edge node []  {$\frac{1}{4}$} (t22.north);

\end{tikzpicture} 
\caption{Counterexample for nonzero-sum probabilistic reachability properties.}\label{counter3-fig}
\end{figure}

\noindent
Consider the CSG in \figref{counter3-fig} with players $p_1$ and $p_2$ (an adaptation of a TSG example from~\cite{BMS16}) and the nonzero-sum state formula $\nashop{p_1{:}p_2}{\max=?}{\theta}$, where $\theta=\probop{}{\future \ap_1}{+}\probop{}{\future \ap_2}$ and $\ap_i$ is the atomic proposition satisfied only by the state $t_i$. Clearly, this CSG has a non-terminal end component as one can remain in $\{ s_1 , s_2 \}$ indefinitely or leave at any time.

Applying the value iteration algorithm of \sectref{mc-sect}, we have:
\begin{itemize}
\item In the first iteration $\V_{\game^C}(s_1,\theta,1)$ are the SWNE values of the bimatrix game:
\[
\mgame_1 \; = \; \kbordermatrix{ & \bot \cr
c & 0 \cr
s & \frac{1}{4}  \cr
} 
\quad \mbox{and} \quad
\mgame_2 \; = \; \kbordermatrix{ & \bot \cr
c & 0 \cr
s & \frac{3}{4}  \cr
}
\]
i.e. the values $(\frac{1}{4},\frac{3}{4})$, and $\V_{\game^C}(s_2,\theta,1)$ are the SWNE values of the bimatrix game:
\[
\mgame_1 \; = \; \kbordermatrix{ & c & s \cr
\bot\!\! & 0 & \frac{3}{4} \cr
}
\quad \mbox{and} \quad
\mgame_2 \; = \; \kbordermatrix{ & c & s \cr
\bot\!\! & 0 & \frac{1}{4} \cr
}
\]
i.e. the values $\
{(\frac{3}{4},\frac{1}{4})}$.
\item In the second iteration $\V_{\game^C}(s_1,\theta,2)$ are the SWNE values of the bimatrix game:
\[
\mgame_1 \; = \; \kbordermatrix{ & \bot \cr
c & \frac{3}{4}  \cr
s & \frac{1}{4}  \cr
}
\quad \mbox{and} \quad
\mgame_2 \; = \; \kbordermatrix{ & \bot \cr
c & \frac{1}{4} \cr
s & \frac{3}{4}  \cr
}
\]
i.e. the values $(\frac{3}{4},\frac{1}{4})$, and $\V_{\game^C}(s_2,\theta,2)$ are the SWNE values of the bimatrix games:
\[
\mgame_1 \; = \; \kbordermatrix{ & c & s \cr
\bot\!\! & \frac{1}{4} & \frac{3}{4} \cr
}
\quad \mbox{and} \quad
\mgame_2 \; = \; \kbordermatrix{ & c & s \cr
\bot\!\! & \frac{3}{4} & \frac{1}{4} \cr
}
\]
i.e. the values $(\frac{1}{4},\frac{3}{4})$.
\item 
In the third iteration $\V_{\game^C}(s_1,\theta,3)$ are the SWNE values of the bimatrix game:
\[
\mgame_1 \; = \; \kbordermatrix{ & \bot \cr
c & \frac{1}{4}  \cr
s & \frac{1}{4}  \cr
}
\quad \mbox{and} \quad
\mgame_2 \; = \; \kbordermatrix{ & \bot \cr
c & \frac{3}{4} \cr
s & \frac{3}{4}  \cr
}
\]
i.e. the values $(\frac{1}{4},\frac{3}{4})$, and $\V_{\game^C}(s_2,\theta,3)$ are the SWNE values of the bimatrix game:
\[
\mgame_1 \; = \; \kbordermatrix{ & c & s \cr
\bot\!\! & \frac{3}{4} & \frac{3}{4} \cr
}
\quad \mbox{and} \quad
\mgame_2 \; = \; \kbordermatrix{ & c & s \cr
\bot\!\! & \frac{1}{4} & \frac{1}{4} \cr
}
\]
i.e. the values $(\frac{3}{4},\frac{1}{4})$. 
\item
In the fourth iteration $\V_{\game^C}(s_1,\theta,4)$ are the SWNE values of the bimatrix game:
\[
\mgame_1 \; = \; \kbordermatrix{ & \bot \cr
c & \frac{3}{4}  \cr
s & \frac{1}{4}  \cr
}
\quad \mbox{and} \quad
\mgame_2 \; = \; \kbordermatrix{ & \bot \cr
c & \frac{1}{4} \cr
s & \frac{3}{4}  \cr
}
\]
i.e. the values $(\frac{3}{4},\frac{1}{4})$, and $\V_{\game^C}(s_2,\theta,4)$ are the SWNE values of the bimatrix game:
\[
\mgame_1 \; = \; \kbordermatrix{ & c & s \cr
\bot\!\! & \frac{3}{4}  & \frac{3}{4} \cr
}
\quad \mbox{and} \quad
\mgame_2 \; = \; \kbordermatrix{ & c & s \cr
\bot\!\! & \frac{1}{4} & \frac{1}{4} \cr
}
\]
i.e. the values $(\frac{3}{4},\frac{1}{4})$.
\end{itemize}
As can be seen the values computed at each iteration for the states $s_1$ and $s_2$ will oscillate between $(\frac{1}{4},\frac{3}{4})$ and $(\frac{3}{4},\frac{1}{4})$.\section{Convergence of Nonzero-Sum Expected Reachability Properties}\label{4-app}

In this appendix we give a witness to the failure of convergence for value iteration when verifying nonzero-sum formulae with infinite horizon reward objectives if \assumref{game3-assum} does not hold.

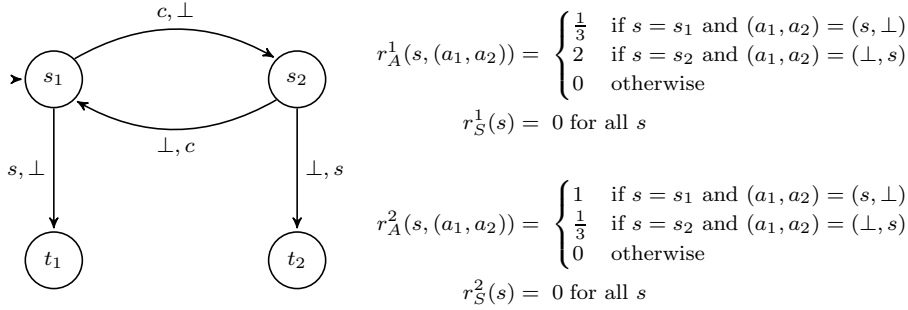
\begin{figure}[h]
\vspace*{-0.8cm}
\begin{subfigure}{.4\textwidth}
\centering
\begin{tikzpicture}[->,>=stealth',shorten >=1pt,auto,node distance=2.8cm, semithick, scale=.40]
  \tikzstyle{every state}=[draw=black,text=black, initial text=]
\small 
\node[initial,state,minimum width=0.75cm,minimum height=0.75cm] (S)at(0,0) (s0) {$s_1$}; 

\node[state,minimum width=0.75cm,minimum height=0.75cm] (S)at(8,0) (s1) {$s_2$}; 

\node[state,minimum width=0.75cm,minimum height=0.75cm] (S)at(0,-6) (s2) {$t_1$}; 

\node[state,minimum width=0.75cm,minimum height=0.75cm] (S)at(8,-6) (s3) {$t_2$}; 

\path [->] (s0.north east) [bend left]
edge node []  {$c,\bot$} (s1.north west);

\path [->] (s1.south west) [bend left]
edge node []  {$\bot,c$} (s0.south east);

\path [->] (s0.south) []
edge node [swap]  {$s,\bot$} (s2.north);

\path [->] (s1.south) []
edge node []  {$\bot,s$} (s3.north);

\end{tikzpicture} 
\end{subfigure}
\begin{subfigure}{.6\textwidth}
\begin{align*}
\\
r_A^1(s, (a_1,a_2) ) & = \;\begin{cases}
\frac{1}{3} & \mbox{if $s = s_1$ and $(a_1,a_2) = (s,\bot)$} \\
2 & \mbox{if $s = s_2$ and $(a_1,a_2) = (\bot,s)$} \\
0 & \mbox{otherwise}
\end{cases} \\
r_S^1(s) & = \; 0 \; \mbox{for all $s$} \\ \\
r_A^2(s, (a_1,a_2) ) & = \; \begin{cases}
1 & \mbox{if $s = s_1$ and $(a_1,a_2) = (s,\bot)$} \\
\frac{1}{3} & \mbox{if $s = s_2$ and $(a_1,a_2) = (\bot,s)$} \\
0 & \mbox{otherwise}
\end{cases} \\
r_S^2(s) & = \; 0 \; \mbox{for all $s$} \\
\end{align*}
\end{subfigure}
\vspace*{-0.4cm}
\caption{Counterexample for nonzero-sum expected reachability properties.}\label{counter4-fig}
\end{figure}

\noindent
Consider the CSG in \figref{counter4-fig} with players $p_1$ and $p_2$ (which again is an adaptation of a TSG example from~\cite{BMS16}) and the nonzero-sum state formula $\nashop{p_1{:}p_2}{\max=?}{\theta}$, where $\theta=\rewop{r_1}{}{\future \ap}{+}\rewop{r_2}{}{\future \ap}$ and $\ap$ is the atomic proposition satisfied only by the states $t_1$ and $t_2$. Clearly, there are strategy profiles for which the targets are not reached with probability 1. 

Applying the value iteration algorithm of \sectref{mc-sect}, we have:
\begin{itemize}
\item In the first iteration $\V_{\game^C}(s_1,\theta,1)$ are the SWNE values of the bimatrix game:
\[
\mgame_1 \; = \; \kbordermatrix{ & \bot \cr
c & 0 \cr
s & \frac{1}{3}  \cr
} 
\quad \mbox{and} \quad
\mgame_2 \; = \; \kbordermatrix{ & \bot \cr
c & 0 \cr
s & 1  \cr
}
\]
i.e. the values $(\frac{1}{3},1)$, and $\V_{\game^C}(s_2,\theta,1)$ are the SWNE values of the bimatrix game:
\[
\mgame_1 \; = \; \kbordermatrix{ & c & s \cr
\bot\!\! & 0 & 2 \cr
}
\quad \mbox{and} \quad
\mgame_2 \; = \; \kbordermatrix{ & c & s \cr
\bot\!\! & 0 & \frac{1}{3} \cr
}
\]
i.e. the values $\
{(2,\frac{1}{3})}$.
\item In the second iteration $\V_{\game^C}(s_1,\theta,2)$ are the SWNE values of the bimatrix game:
\[
\mgame_1 \; = \; \kbordermatrix{ & \bot \cr
c & 2 \cr
s & \frac{1}{3}  \cr
}
\quad \mbox{and} \quad
\mgame_2 \; = \; \kbordermatrix{ & \bot \cr
c & \frac{1}{3} \cr
s & 1  \cr
}
\]
i.e. the values $(2,\frac{1}{3})$, and $\V_{\game^C}(s_2,\theta,2)$ are the SWNE values of the bimatrix games:
\[
\mgame_1 \; = \; \kbordermatrix{ & c & s \cr
\bot\!\! & \frac{1}{3} & 2 \cr
}
\quad \mbox{and} \quad
\mgame_2 \; = \; \kbordermatrix{ & c & s \cr
\bot\!\! & 1 & \frac{1}{3} \cr
}
\]
i.e. the values $(\frac{1}{3},1)$.
\item 
In the third iteration $\V_{\game^C}(s_1,\theta,3)$ are the SWNE values of the bimatrix game:
\[
\mgame_1 \; = \; \kbordermatrix{ & \bot \cr
c & \frac{1}{3}  \cr
s & \frac{1}{3}  \cr
}
\quad \mbox{and} \quad
\mgame_2 \; = \; \kbordermatrix{ & \bot \cr
c & 1 \cr
s & 1  \cr
}
\]
i.e. the values $(\frac{1}{3},1)$, and $\V_{\game^C}(s_2,\theta,3)$ are the SWNE values of the bimatrix game:
\[
\mgame_1 \; = \; \kbordermatrix{ & c & s \cr
\bot\!\! & 2 & 2 \cr
}
\quad \mbox{and} \quad
\mgame_2 \; = \; \kbordermatrix{ & c & s \cr
\bot\!\! & \frac{1}{3} & \frac{1}{3} \cr
}
\]
i.e. the values $(2,\frac{1}{3})$. 
\item
In the fourth iteration $\V_{\game^C}(s_1,\theta,4)$ are the SWNE values of the bimatrix game:
\[
\mgame_1 \; = \; \kbordermatrix{ & \bot \cr
c & 2  \cr
s & \frac{1}{3}  \cr
}
\quad \mbox{and} \quad
\mgame_2 \; = \; \kbordermatrix{ & \bot \cr
c & \frac{1}{3} \cr
s & 1  \cr
}
\]
i.e. the values $(2,\frac{1}{3})$, and $\V_{\game^C}(s_2,\theta,4)$ are the SWNE values of the bimatrix game:
\[
\mgame_1 \; = \; \kbordermatrix{ & c & s \cr
\bot\!\! & \frac{1}{3}  & 2 \cr
}
\quad \mbox{and} \quad
\mgame_2 \; = \; \kbordermatrix{ & c & s \cr
\bot\!\! & 1 & \frac{1}{3} \cr
}
\]
i.e. the values $(\frac{1}{3},1)$.
\end{itemize}
As can be seen the values computed during value iteration oscillate for both $s_1$ and $s_2$.

\end{document}